\documentclass[3p]{elsarticle}

\usepackage{amsmath}
\usepackage{amsthm}
\usepackage{braket}
\usepackage{lipsum}
\usepackage{amsfonts}
\usepackage{graphicx}
\usepackage{epstopdf}
\usepackage{algorithmic}
\usepackage{multirow}
\usepackage{subcaption}
\usepackage{amsopn}
\usepackage{algorithm}
\usepackage{color}
\usepackage{hyperref}

\newcommand{\bE}{{\mathbb{E}}}
\newcommand{\bI}{{\mathbb{I}}}
\newcommand{\bN}{{\mathbb{N}}}
\newcommand{\bR}{{\mathbb{R}}}

\newcommand{\cA}{{\mathcal{A}}}
\newcommand{\cB}{{\mathcal{B}}}
\newcommand{\cC}{{\mathcal{C}}}
\newcommand{\cD}{{\mathcal{D}}}
\newcommand{\cE}{{\mathcal{E}}}
\newcommand{\cF}{{\mathcal{F}}}
\newcommand{\cH}{{\mathcal{H}}}
\newcommand{\cM}{{\mathcal{M}}}
\newcommand{\cN}{{\mathcal{N}}}
\newcommand{\cR}{{\mathcal{R}}}
\newcommand{\cV}{{\mathcal{V}}}
\newcommand{\cCM}{{\mathcal{C}'}}

\newcommand{\bsb}{{\boldsymbol{b}}}

\newcommand{\bsg}{{\boldsymbol{g}}}
\newcommand{\bsm}{{\boldsymbol{m}}}
\newcommand{\bsu}{{\boldsymbol{u}}}
\newcommand{\bsv}{{\boldsymbol{v}}}
\newcommand{\bsw}{{\boldsymbol{w}}}
\newcommand{\bsx}{{\boldsymbol{x}}}
\newcommand{\bsy}{{\boldsymbol{y}}}
\newcommand{\bsz}{{\boldsymbol{z}}}

\newcommand{\bsepsilon}{{\boldsymbol{\epsilon}}}
\newcommand{\bszero}{{\boldsymbol{0}}}

\newcommand{\bsbeta}{{\boldsymbol{\beta}}}
\newcommand{\bspsi}{{\boldsymbol{\psi}}}
\newcommand{\bsphi}{{\boldsymbol{\phi}}}
\newcommand{\bsvarphi}{{\boldsymbol{\varphi}}}
\newcommand{\bseta}{{\boldsymbol{\eta}}}
\newcommand{\bstheta}{{\boldsymbol{\theta}}}

\newcommand{\red}[1]{{\color{red}{#1}}}

\newtheorem{theorem}{Theorem}
\newtheorem{assumption}{Assumption}
\newtheorem{remark}{Remark}

\usepackage{lipsum}
\makeatletter
\def\ps@pprintTitle{%
 \let\@oddhead\@empty
 \let\@evenhead\@empty
 \def\@oddfoot{}%
 \let\@evenfoot\@oddfoot}
\makeatother
\begin{document}

\biboptions{longnamesfirst}

\title{Accurate, scalable, and efficient Bayesian optimal experimental design with derivative-informed neural operators
}

\author[1]{Jinwoo Go}
\ead{jgo31@gatech.edu}

\author[1]{Peng Chen}
\ead{pchen402@gatech.edu}


\affiliation[1]{organization={School of Computational Science and Engineering},
addressline={Georgia Institute of Technology},
postcode={30332 GA},
city={Atlanta},
country={USA}}

\begin{abstract}
We consider optimal experimental design (OED) problems in selecting the most informative observation sensors to estimate model parameters in a Bayesian framework. Such problems are computationally prohibitive when the parameter-to-observable (PtO) map is expensive to evaluate, the parameters are high-dimensional, and the optimization for sensor selection is combinatorial and high-dimensional. To address these challenges, we develop an accurate, scalable, and efficient computational framework based on derivative-informed neural operators (DINO). We propose to use derivative-informed dimension reduction to reduce the parameter dimensions, based on which we train DINO with derivative information as an accurate and efficient surrogate
for the PtO map and its derivative. Moreover, we derive DINO-enabled efficient formulations in computing the maximum a posteriori (MAP) point, the eigenvalues of approximate posterior covariance, and three commonly used optimality criteria for the OED problems. Furthermore, we provide detailed error analysis for the approximations of the MAP point, the eigenvalues, and the optimality criteria. We also propose a modified swapping greedy algorithm for the sensor selection optimization and demonstrate that the proposed computational framework is scalable to preserve the accuracy for increasing parameter dimensions and achieves high computational efficiency, with an over 1000$\times$ speedup accounting for both offline construction and online evaluation costs, compared to high-fidelity Bayesian OED solutions for a three-dimensional nonlinear convection-diffusion-reaction example with tens of thousands of parameters. 
\end{abstract}

\maketitle

\section{Introduction}

Mathematical modeling and computational simulation become increasingly crucial for understanding and optimizing complex systems in many scientific and engineering fields. In practical applications, various uncertainties arise in the modeling and simulation either because of the lack of knowledge of the systems (epistemic uncertainty) or the inherent randomness in the systems (aleatoric uncertainty) that cannot be reduced. To reduce the epistemic uncertainty commonly represented as uncertain model parameters, it is essential to acquire knowledge or data of the systems by proper experimental design and calibrate the model parameters by either deterministic or statistical/Bayesian inference from the experimental or observational data. However, these experiments may be costly or take an enormous time or space. This is where optimal experimental design (OED) comes into play. OED designs experiments to acquire the most informative data under budget/time/space constraints \cite{Atkinson2011, pukelsheim2006optimal, AlexanderianPetraStadlerEtAl14, HuanMarzouk13, WuChenGhattas23a}.

In this work, we consider OED for Bayesian inverse problems constrained by large-scale models described by partial differential equations (PDEs) under infinite/high-dimensional uncertain parameters, see \cite{Alexanderian21, AlexanderianPetraStadlerEtAl14, AlexanderianGloorGhattas16, AlexanderianSaibaba18, BeckDiaEspathEtAl18, WuChenGhattas23a, WuChenGhattas23, WuOLearyRoseberryChenEtAl23} and references therein. For the design optimality criteria, we consider both variance-based alphabetic optimality \cite{AlexanderianPetraStadlerEtAl14, AlexanderianGloorGhattas16, AlexanderianSaibaba18}, and information theory-based optimality criterion known as mutual information or
expected information gain (EIG) \cite{BeckDiaEspathEtAl18,WuChenGhattas23a,WuOLearyRoseberryChenEtAl23}. The variance-based optimality criteria seek to minimize the uncertainty represented by summary statistics of the posterior covariance, such as A-optimality (the trace of posterior covariance) and D-optimality (the determinant of posterior covariance), while the information theory-based EIG seeks to maximize the average information gained from all possible realizations of the experimental data. Mathematically, the EIG is defined as the expectation of Kullback--Leibler divergence between the posterior and the prior distributions over the data or marginal likelihood distribution, which involves the integration with respect to both the posterior and data distributions.

Several challenges are faced in solving OED with these optimality criteria for PDE-constrained large-scale Bayesian inverse problems, 
including: 
(1) each evaluation of the optimality criteria requires the computation of a parameter-to-observable (PtO) map---the map from the uncertain parameters to the observational quantity of interest of the system---at many data and parameter samples; (2) each computation of the PtO map involves a solution of the model, which can be very expensive for large-scale PDE models; (3) the dimensions of the uncertainty parameter space and the experiment design space can be very high or infinite as considered in our work, which brings the curse-of-dimensionality, i.e., the computational/sampling complexity increase rapidly (exponentially) with respect to the dimensions, for many conventional methods; and (4) the combinatorial optimization with respect to the experimental design may be highly non-convex, as in the example of optimal placement of sensors from a large number of candidate sensor locations.

To address these challenges, various computational
methods have been developed over the last decade, including (1) sparse polynomial chaos approximation of PtO maps
\cite{HuanMarzouk13,HuanMarzouk14},
(2) Laplace approximation of non-Gaussian posterior distributions
\cite{AlexanderianPetraStadlerEtAl16,
BeckDiaEspathEtAl18, LongScavinoTemponeEtAl13a, LongMotamedTempone15,
BeckMansourDiaEspathEtAl20}, (3) low-rank approximation of prior-preconditioned Hessian of the data misfit term \cite{AlexanderianPetraStadlerEtAl14,AlexanderianGloorGhattas16,AlexanderianPetraStadlerEtAl16,SaibabaAlexanderianIpsen17,CrestelAlexanderianStadlerEtAl17,AttiaAlexanderianSaibaba18}, 
(4) reduced order models
\cite{Aretz-NellesenChenGreplEtAl20, AretzChenVeroy21,AretzChenDegenEtAl24} 
and deep neural networks \cite{WuOLearyRoseberryChenEtAl23} that serve as surrogate models of the PDEs or PtO maps, (5) variational inference and neural estimation in a fast approximation of the EIG or mutual information \cite{FosterJankowiakBinghamEtAl19,
KleinegesseGutmann20,GoIsaac22,ShenHuan2023}, and (6) efficient optimization methods based on gradient information \cite{AlexanderianPetraStadlerEtAl14,AlexanderianPetraStadlerEtAl16,HuanMarzouk14}, greedy algorithms \cite{JagalurMohanMarzouk21,helin2022edge,AretzChenVeroy21, AretzChenDegenEtAl24}, and swapping greedy algorithms \cite{WuChenGhattas23,WuChenGhattas23a,WuOLearyRoseberryChenEtAl23}. 

In this work, we consider the Laplace approximation of the posterior distribution \cite{AlexanderianPetraStadlerEtAl16,
BeckDiaEspathEtAl18, LongScavinoTemponeEtAl13a, LongMotamedTempone15,
BeckMansourDiaEspathEtAl20} and the low-rank approximation of the prior-preconditioned Hessian of the data misfit term \cite{AlexanderianPetraStadlerEtAl14,
AlexanderianGloorGhattas16,AlexanderianPetraStadlerEtAl16,
SaibabaAlexanderianIpsen17,CrestelAlexanderianStadlerEtAl17,AttiaAlexanderianSaibaba18}. Though achieving significant computational reduction, such approximations are still infeasible for large-scale PDE models with high-dimensional parameters. This is because (1) a Maximum-A-Posteriori (MAP) point needs to be computed in the Laplace approximation for each observed data and each given experimental design, which requires the solution of a large-scale PDE-constrained high-dimensional optimization problem; (2) the evaluation of the optimality criteria involves a low-rank factorization of the large-scale prior-preconditioned Hessian of the misfit at each MAP point; (3) both the MAP points and the low-rank factorizations need to be computed numerous times, e.g., millions of times for the combinatorial optimization of the high-dimensional experimental design by a greedy algorithm.

\textbf{Contributions.} We propose addressing the computational challenges in Bayesian OED by developing an accurate, scalable, and efficient computational framework based on the derivative-informed neural operator (DINO) \cite{OLeary-RoseberryChenVillEtAl24}. 
We propose to use derivative-informed dimension reduction to reduce the parameter dimensions and principal component analysis (PCA) to reduce the PtO dimensions with comparison to alternative methods. Based on the dimension reduction, we train DINO with derivative information as a neural operator surrogate of the PtO map and its derivative. This approach uses a relatively small amount of training data (PtO map and its Jacobian) to achieve high approximation accuracy of the neural operator. 
Moreover, we derive DINO-enabled efficient formulations in computing the maximum a posteriori (MAP) point, the eigenvalues of approximate posterior covariance, and three commonly used optimality criteria for the OED problems. 
Specifically, we compute the MAP point by solving a low-dimensional optimization problem, leveraging the DINO evaluation of the PtO map and its accurate Jacobian. We approximate the posterior covariance matrix and optimality criteria at the MAP point through an eigenvalue problem of the misfit Hessian in the reduced input and output subspace. 
We provide detailed error analysis for the approximations of the MAP point, the eigenvalues, and the optimality criteria by DINO under appropriate assumptions on the accuracy of the neural surrogate and dimension reduction.
 These computations facilitate rapid evaluations of optimality criteria for each data point and corresponding experimental design, accommodating any number of sensors. To this end, we propose modifying a swapping greedy algorithm developed in \cite{WuChenGhattas23a} to solve the Bayesian OED optimization problem, enhancing its optimality through a different initialization of the sensors using a greedy algorithm. We demonstrate our approach's high accuracy, scalability, and efficiency with two examples of two-dimensional (2D) and three-dimensional (3D) PDE models, involving hundreds of observables and tens of thousands of parameters. Specifically, we achieve an 80$\times$ speedup compared to the high-fidelity Bayesian OED solution for the 2D model and a 1148$\times$ speedup for the 3D model, including both the online evaluation and offline construction costs of the surrogates. 

\textbf{Related work.}
Our work is related to the recent work \cite{WuChenGhattas23a} in using Laplace and low-rank approximations and \cite{WuOLearyRoseberryChenEtAl23} in using surrogate models than other works listed above. We point out the significant difference that in \cite{WuChenGhattas23a}, the MAP point in the Laplace approximation is replaced by the corresponding prior sample point to save the prohibitive cost of computing the MAP point, which is a limitation since the optimality criteria may not be well correlated for the two points. Moreover, the authors assume that observation noise is uncorrelated to enable an offline-online decomposition of the low-rank approximation. In \cite{WuOLearyRoseberryChenEtAl23}, a neural network surrogate is constructed to compute only the PtO map and the EIG by a double loop Monte Carlo (DLMC) sampling. The surrogate is not built using the additional Jacobian information; thus, it is inaccurate and not used to compute the Jacobian and the related A-/D-optimality criteria, as well as the Laplace approximation-based EIG. It is limited to relatively large observation noise because of the pollution of the neural network approximation errors in the accurate computation of the likelihood function, which is crucial for the DLMC sampling. In contrast, our work is not subject to these limitations. However, we remark that our method is limited to the low-rankness or fast spectral decay of the PtO map and the Hessian of the data misfit term as in \cite{WuChenGhattas23a,WuOLearyRoseberryChenEtAl23}. The fast spectral decay is the intrinsic property of many high-dimensional problems. It has been exploited in various other contexts besides Bayesian OED, e.g., Bayesian inference \cite{Bui-ThanhBursteddeGhattasEtAl12, Bui-ThanhGhattasMartinEtAl13, ChenVillaGhattas17, ChenWuChenEtAl19, ChenWuGhattas21, ChenGhattas20,WangChenLi2022}, model reduction for sampling and deep learning \cite{BashirWillcoxGhattasEtAl08, ChenGhattas19a, AlgerChenGhattas20, OLeary-RoseberryVillaChenEtAl22}, and optimization under uncertainty \cite{AlexanderianPetraStadlerEtAl17,ChenVillaGhattas19, ChenHabermanGhattas21,LuoOLearyRoseberryChenEtAl23}.

The rest of the paper is organized as follows:
Section \ref{sec:BOED} presents Bayesian OED problems with three optimality criteria constrained by PDEs to infer infinite-dimensional parameters from noisy observations. Section \ref{sec:discretization} then discusses solving these Bayesian OED problems using high-fidelity approximation methods, including finite element discretization, Laplace approximation, and low-rank approximations of the posterior. In Section \ref{sec:neural}, we introduce DINO surrogates incorporating proper dimension reduction techniques, efficient computation of optimality criteria, and a modified swapping greedy algorithm to accelerate the solution of Bayesian OED problems. Section \ref{sec:ErrorAnalysis} provides error analyses for the approximations of the MAP point, the eigenvalues of the approximate posterior covariance, and the optimality criteria under appropriate assumptions. We present numerical results in Section \ref{sec:numerical}, validating these assumptions and demonstrating the accuracy, scalability, and efficiency of our proposed method in solving Bayesian OED problems constrained by large-scale PDEs. Finally, we conclude the paper in Section \ref{sec:conclusion}.

\section{Problem statement}
\label{sec:BOED}
We present Bayesian inverse problems constrained by PDEs with infinite-dimensional parameters, Bayesian OED problems in the context of optimal sensor placement, and the challenges for their solutions.

\subsection{Bayesian inverse problem}

We consider Bayesian inverse problems to infer uncertain parameters in mathematical models described by PDEs written in an abstract form as 
\begin{equation}\label{eq:PDE}
\cR(u,m) = 0 \text{  in } \cV', 
\end{equation}
where $u \in \cV$ is the state variable in a separable Banach space $\cV$ defined in physical domain $D\in\bR^{d}$, e.g., in dimension $d = 1,2,3$; $m$ represents the uncertain parameter. Specifically, we consider $m$ as a random field parameter in a Hilbert space $\cM$ defined in $D$, which is infinite-dimensional. The operator $\cR: \cV \times \cM \to \cV'$ represents the strong form of the PDE, where $\cV'$ is the dual of $\cV$. 

Let  $\cB: \cV \to \bR^{d}$ denote an observation operator that maps the state variable $u$ to a vector of $d_s$-dimensional observation functionals, with $d_s \in \bN$. We define a corresponding PtO map $\cF: \cM \to \bR^{d_s}$ as $\cF(m) := \cB(u(m))$. 
In particular, we consider the observation made in $d_s$ sensor locations $(\bsx_1, ..., \bsx_{d_s})$ with $\bsx_i \in D$ indicating the $i$-th sensor location in domain $D$. In the OED problem, we can only choose $r_s \in \bN$ out of the $d_s$ candidate sensors because of the budget constraint, where typically $d_s \gg r_s$. We use a design matrix $\xi \in \bR^{d_s \times r_s}$ to represent the selection of $r_s$ sensors out of the $d_s$ candidates, with $\xi_{ij}$ = 1 if the $j$-th sensor is selected at the $i$-th location $\bsx_i$, and $\xi_{ij} = 0$ if not. Moreover, we consider that one location can not be placed with more than one sensor. Mathematically, we define the admissible design matrix set $\Xi$ as 
\begin{equation}\label{eq:Xi}
    \Xi = \left\{\xi \in \bR^{d_s \times r_s}: \xi_{ij} \in \{0,1\}, \ \sum_{i=1}^{d_s}\xi_{ij} = 1, \ \sum_{j=1}^{r_s}\xi_{ij} \in \{0,1\} \right\}.
\end{equation}
For a given design matrix $\xi \in \Xi$, we denote the corresponding PtO map as $\cF_\xi = \xi^T\cF: \cM \to \bR^{r_s}$. We consider that the
noisy observation data $\bsy$ is corrupted with an additive observation noise $\bsepsilon$, given as
\begin{equation}\label{eq:data}
    \bsy = \cF_\xi(m) + \boldsymbol{\epsilon},
\end{equation}
where we assume the observation noise obeys Gaussian distribution $\boldsymbol{\epsilon} \sim \cN(\bszero,\Gamma_\text{noise})$ with zero mean and covariance $\Gamma_\text{noise} \in \bR^{r_s \times r_s}$. Under this assumption, the likelihood of the data $\bsy$ for a parameter $m$ and design matrix $\xi$ is given by  
\begin{equation}\label{eq:clikelihood}
  \pi_\text{like}(\bsy|m,\xi) = \frac{1}{\sqrt{(2\pi)^r |\Gamma_\text{noise}|}} \exp\left(-\Phi(\bsy, m, \xi)\right),
\end{equation}
where $|\Gamma_\text{noise}|$ denotes the determinant of $\Gamma_\text{noise}$, and the potential function $\Phi(m, \bsy, \xi)$ represents a data-model misfit term given as 
\begin{equation}
  \Phi(\bsy, m, \xi) = \frac{1}{2}||\bsy - \cF_\xi(m)||^2_{\Gamma_\text{noise}^{-1}} = \frac{1}{2} (\bsy - \cF_\xi (m))^T \Gamma_\text{noise}^{-1} (\bsy - \cF_\xi (m)).
\end{equation}

Given data $\bsy$ in \eqref{eq:data}, the inverse problem is to infer the parameter $m$ constrained by the model \eqref{eq:PDE}. In a Bayesian framework, the inference is to update the probability distribution of the parameter $m$ from its prior distribution $\mu_\text{prior}$ to the posterior distribution $\mu_\text{post}$ conditioned on the observation data $\bsy$ and design matrix $\xi$.
We consider a Gaussian prior $\mu_\text{prior} = \cN(m_\text{prior},\cC_\text{prior})$ with mean $m_\text{prior}$ and covariance operator $\cC_\text{prior}:\cM \rightarrow \cM'$. 
In particular, we use a common choice of Mat\'ern covariance operator $\cC_\text{prior} = \cA^{-\alpha}$ \cite{VillaPetraGhattas21}, where $\cA=-\gamma\Delta + \kappa I$ with homogeneous Neumann boundary condition, $\Delta$ is the Laplacian, $I$ denotes the identity, and the constants $\alpha, \gamma, \kappa > 0$ collectively determine the smoothness, correlation, and variance of the parameter. We require $\alpha > d/2$ for the covariance to be of trace class \cite{stuart2010inverse}, where $d$ is the dimension of the domain $D$. Given this prior, Bayes' rule states that the posterior of the parameter written in the form of Radon--Nikodym derivative satisfies
\begin{equation}
    \frac{d\mu_\text{post}^{\bsy,\xi}(m)}{d\mu_\text{prior}(m)} = \frac{1}{\pi(\bsy|\xi)} \pi_\text{like}(\bsy|m,\xi),
\end{equation}
where $\pi(\bsy|\xi)$ is known as model evidence or marginal likelihood distribution given by 
\begin{equation}\label{eq:evidence}
  \pi(\bsy|\xi) = \int_\cM \pi_\text{like}(\bsy|m,\xi) d\mu_\text{prior}(m),
\end{equation}
which is typically computationally intractable in infinite or high dimensions.

\subsection{Bayesian optimal experimental design}
The Bayesian OED problem in our context seeks the optimal placement of $r_s$ sensors out of $d_s$ candidate locations that provide maximum information to infer the parameter. For this task, we consider several widely used optimality criteria, including
the variance-based A-optimality criterion and D-optimality criterion, and the information theory-based EIG or mutual information. 

The variance-based A-/D-optimality criteria are defined for a Laplace approximation of the posterior distribution 
given by 
$\hat{\mu}_\text{post}^{\bsy, \xi} =  \cN(m_\text{MAP}^{\bsy, \xi},\cC_\text{post}^{\bsy, \xi})$. Here, $m_\text{MAP}^{\bsy, \xi}$ denotes the MAP point defined as 
\begin{equation}\label{eq:c-map}
    m_\text{MAP}^{\bsy, \xi} := \arg \min_{m\in \cCM}\frac{1}{2}||\bsy - \cF_\xi(m)||^2_{\Gamma_\text{noise}^{-1}}+\frac{1}{2}||m-m_\text{prior}||^2_{\cC_\text{prior}^{-1}},
\end{equation}
where $\cCM$ = range$(\cC_\text{prior}^{1/2})$ is the Cameron–Martin space \cite{da2006introduction} of the Gaussian measure $\mu_\text{prior}$, which is a subspace of $\cM$.
The covariance operator $\cC_\text{post}$ is given by
\begin{equation}\label{eq:c-post-cov}
  \cC_\text{post}^{\bsy, \xi} = (\cH_\text{misfit}^{\bsy, \xi}(m_\text{MAP}^{\bsy, \xi}) + \cC_\text{prior}^{-1})^{-1},    
\end{equation}
where $\cH_\text{misfit}^{\bsy, \xi}(m_\text{MAP}^{\bsy, \xi})$ is the Hessian of data-model misfit term $\Phi(m,\bsy,\xi)$ with respect to $m$ evaluated at $m_\text{MAP}^{\bsy, \xi}$. A common choice for $\cH_\text{misfit}^{\bsy, \xi}(m_\text{MAP}^{\bsy, \xi})$ in infinite-/high-dimensional Bayesian OED is to use its Gauss--Newton approximation 
\begin{equation}\label{eq:c-hmisfit}
  \cH_{\text{misfit}}^{\bsy, \xi}(m_{\text{MAP}}^{\bsy, \xi}) \approx \cH_{\text{misfit}}^{\text{GN}}(m_{\text{MAP}}^{\bsy, \xi}) := \nabla_m \cF_\xi(m_{\text{MAP}}^{\bsy, \xi})^T \Gamma_{\text{noise}}^{-1} \nabla_m \cF_\xi(m_{\text{MAP}}^{\bsy, \xi}),
\end{equation}
with the Jacobian $\nabla_m \cF_\xi(m_{\text{MAP}}^{\bsy, \xi}) : \cM \to \bR^{r_s}$ evaluated at the MAP point $m_{\text{MAP}}^{\bsy, \xi}$. This corresponds to a Fisher information matrix-based optimality criterion \cite{schraudolph2002fast}. 
The Gauss--Newton approximation of the Hessian is exact if the PtO map $\cF_\xi(m)$ is linear with respect to $m$, or it can be considered as a linear approximation of $\cF_\xi(m)$ at the MAP point, i.e.,
\begin{equation}
    \cF_\xi(m) \approx \cF_\xi(m_{\text{MAP}}^{\bsy, \xi}) + \nabla_m \cF_\xi(m_{\text{MAP}}^{\bsy, \xi})(m - m_{\text{MAP}}^{\bsy, \xi}).
\end{equation} 
We note that for the A-/D-optimality criteria, we need to obtain the MAP point $m_{\text{MAP}}^{\bsy, \xi}$ by solving the optimization problem \eqref{eq:c-map} and compute the Jacobian $\nabla_m \cF_\xi(m_{\text{MAP}}^{\bsy, \xi})$ at the MAP point, both of which are constrained by the PDE \eqref{eq:PDE} and depend on the realization of the data $\bsy$ and the design matrix $\xi$.

To this end, we consider the Bayesian OED problem with the A-optimality criterion defined as \cite{AlexanderianGloorGhattas16}
\begin{equation}\label{eq:c-a-opt} 
  \xi^*_A = \arg\min_{\xi \in \Xi}\bE_{\pi(\bsy|\xi)}[\text{trace}(\cC_\text{post}^{\bsy, \xi}(m^{\bsy, \xi}_{\text{MAP}}))],
\end{equation}
where trace$(\cdot)$ denotes the trace of the covariance operator $\cC_\text{post}^{\bsy, \xi}(m^{\bsy, \xi}_{\text{MAP}})$ defined in \eqref{eq:c-post-cov}, with the superscript representing the dependence on the data $\bsy$ observed at the sensors determined by the design matrix $\xi$. The expectation $\bE_{\pi(\bsy|\xi)}[\cdot]$ is taken with respect to the distribution of the data $\bsy$ conditioned on $\xi$. 
This expectation can be evaluated by the double integral $\bE_{m\sim \mu_\text{prior}}\bE_{\bsy \sim \pi_\text{like}(\cdot|m, \xi)}[\cdot]$.

For the D-optimality criterion, we use the definition from \cite{AlexanderianGloorGhattas16} as

\begin{equation}\label{eq:c-d-opt}
  \xi^*_D = \arg\max_{\xi \in \Xi} \bE_{\pi(\bsy|\xi)}[\log \text{det}(\mathcal{I} + \tilde{\cH}^{\bsy, \xi}_{\text{misfit}})],
\end{equation}
with the prior-preconditioned Hessian of the misfit term $\tilde{\cH}^{\bsy, \xi}_{\text{misfit}} = \cC_{\text{prior}}^{1/2} \cH^{\bsy, \xi}_{\text{misfit}} (m^{\bsy, \xi}_{\text{MAP}})\cC_{\text{prior}}^{1/2}$ and the identity operator $\mathcal{I}$,
where det$(\cdot)$ denotes the determinant. Note that 
\begin{equation}
    \mathcal{I} + \tilde{\cH}^{\bsy, \xi}_{\text{misfit}} = \cC_\text{prior}^{1/2}\left(\cH_\text{misfit}^{\bsy, \xi}(m_\text{MAP}^{\bsy, \xi}) + \cC_\text{prior}^{-1}\right)\cC_\text{prior}^{1/2} = \cC_\text{prior}^{1/2}\left(\cC_\text{post}^{\bsy, \xi}(m^{\bsy, \xi}_{\text{MAP}})\right)^{-1}\cC_\text{prior}^{1/2},
\end{equation}
the prior-preconditioned inverse of the posterior covariance operator. 

    
Alternatively, the Bayesian OED with the information theory-based EIG seeks to maximize the expected information gained from the sensors, i.e.,
\begin{equation}\label{eq:c-eig}
  \xi^*_\text{EIG} = \arg\max_{\xi\in\Xi}\bE_{\pi(\bsy|\xi)}\left[D_\text{KL}(\mu_\text{post}^{\bsy,\xi}||\mu_\text{prior})\right],
\end{equation}
where the information gain is measured by the Kullback--Leibler (KL) divergence between the posterior distribution $\mu_\text{post}^{\bsy,\xi}$ and the prior distribution $\mu_\text{prior}$, defined as 
\begin{equation}
  D_\text{KL}(\mu_\text{post}^{\bsy,\xi}||\mu_\text{prior}) = \bE_{m \sim \mu_\text{post}^{\bsy, \xi}} \left[\log \left(\frac{d\mu_\text{post}^{\bsy, \xi}(m)}{d\mu_\text{prior}(m)} \right)\right].
\end{equation} 

Under the assumption of Gaussian prior $\mu_\text{prior} = \cN(m_\text{prior}, \cC_\text{prior})$ and using the Laplace approximation of the posterior $\hat{\mu}^{\bsy,\xi}_\text{prior} = \cN(m^{\bsy,\xi}_\text{MAP}, \cC^{\bsy,\xi}_\text{post})$, the analytical form of the KL divergence is given by \cite{AlexanderianGloorGhattas16}
\begin{equation}\label{eq:c-kl}
  \begin{split}
  D_\text{KL}(\hat{\mu}^{\bsy,\xi}_\text{post}||\mu_\text{prior}) = \frac{1}{2}\log \text{det}(I + \tilde{\cH}^{\bsy, \xi}_{\text{misfit}}) - \frac{1}{2}\text{trace}(\cH^{\bsy,\xi}_\text{misfit} \cC^{\bsy,\xi}_\text{post}) + \frac{1}{2}||m^{\bsy,\xi}_\text{MAP} - m_\text{prior}||^2_{\cC_\text{prior}^{-1}}.
  \end{split}
\end{equation}

Several computational challenges arise in solving the Bayesian OED problems with these optimality criteria, including: (1) the dimension of the uncertain parameter space and the design space can be infinite or very high, which brings the curse of dimensionality, i.e., the computational/sampling complexity increases exponentially with respect to the dimensions; (2) each evaluation of the optimality criteria (trace \eqref{eq:c-a-opt}, determinant \eqref{eq:c-d-opt}, EIG \eqref{eq:c-eig}) requires numerous expensive PDE solves for the computation of the PtO map $\cF_\xi$ and its Jacobian $\nabla_m \cF_\xi$; (3) the optimization with respect to experimental design $\xi$ is combinatorial and high-dimensional, which may require a large number of evaluations of the optimality criteria.


\section{High-fidelity approximations}\label{sec:discretization}
In this section, we follow \cite{AlexanderianPetraStadlerEtAl16, WuChenGhattas23a, VillaPetraGhattas21, IsaacPetraStadlerEtAl15} and present high-fidelity approximations using a finite element method to discretize the infinite-dimensional parameter field and solve the Bayesian OED problem based on Laplace and low-rank approximations. We highlight the specific computational challenges of high-fidelity approximations as the motivation and basics to develop our computational framework in Section \ref{sec:neural}. We also use the high-fidelity approximations of all the related quantities to demonstrate our proposed method's accuracy, scalability, and efficiency in Section \ref{sec:numerical}.

\subsection{High-fidelity discretization}\label{sec:high-fidelity-discretization}
The random field parameter $m$ lives in the infinite-dimensional Hilbert space $\cM$. We discretize $m$ by a finite element method in a finite-dimensional space $\cM_{d_\bsm} \subset \cM$ of dimension $d_\bsm$ with piecewise continuous Lagrange polynomial basis $\{\phi_j\}_{j=1}^{d_\bsm}$ such that 
$\phi_j(x_i) = \delta_{ij}$. Let $h$ denote the mesh size of the discretization. We approximate the parameter by $m_h \in \cM_{d_\bsm}$ defined as
\begin{equation}\label{eq:mFEM}
    m_h(x) = \sum_{j=1}^{d_\bsm}m_j\phi_j(x),
\end{equation}
where $\bsm = (m_1, m_2, ..., m_{d_\bsm})^T \in \bR^{d_\bsm}$ is a coefficient vector of $m_h$. The prior distribution for this discretized parameter $\bsm$ is Gaussian $\mu(\bsm) = \mathcal{N}(\bsm_\text{prior},\Gamma_\text{prior})$ with mean vector $\bsm_{\text{prior}}$ and covariance matrix $\Gamma_\text{prior}$ as discretized from the mean $m_{\text{prior}}$ and the covariance operator $\cC_\text{prior} = \mathcal{A}^{-\alpha}$ with $\mathcal{A} = -\gamma \Delta + \kappa I$. Let $M$ and $A$ denote the finite element mass matrix and stiffness matrix given by 
\begin{equation}\label{eq:mass}
  M_{ij} = \int_D \phi_i(x) \phi_j(x) dx, \ i,j = 1,...,d_\bsm,
\end{equation}
and
\begin{displaymath}
    A_{ij} = \int_D (\gamma \nabla \phi_i(x) \cdot \nabla \phi_j(x) + \kappa \phi_i(x) \phi_j(x)) dx, \ i,j = 1,...,d_\bsm,
\end{displaymath}
then the inverse of the covariance matrix $\Gamma_\text{prior} $ for $\alpha = 2$ is given by \cite{VillaPetraGhattas21}
\begin{equation}\label{eq:GammaInv}
    \Gamma^{-1}_\text{prior} = AM^{-1}A.
\end{equation}

\subsection{Laplace and low-rank approximations}
By the high-dimensional discretization, we can write the Laplace approximation of the posterior distribution of $\bsm$ as $\mathcal{N}(\bsm_\text{MAP}^{\bsy, \xi},\Gamma^{\bsy, \xi}_\text{post}(\bsm_\text{MAP}^{\bsy, \xi}))$, where the MAP point $\bsm_\text{MAP}^{\bsy, \xi}$ is obtained from solving the following finite/high-dimensional optimization problem corresponding to \eqref{eq:c-map} in the function space  
\begin{equation}\label{eq:MAP}
     \bsm_\text{MAP}^{\bsy, \xi} = \arg \min_{\bsm\in \bR^{d_\bsm}}\frac{1}{2}||\bsy- F_\xi(\bsm)||^2_{\Gamma_\text{noise}^{-1}}+\frac{1}{2}||\bsm-\bsm_\text{prior}||^2_{\Gamma_\text{prior}^{-1}},
\end{equation}
with $F_\xi$ denoting a map from $\bsm$ to the observables corresponding to a discrete version of the PtO map $\mathcal{F}_\xi$ for a given sensor selection $\xi$. To solve this optimization problem, we can employ, e.g., an inexact Newton--CG algorithm, which requires the computation of the derivative of $F_\xi(\bsm)$ with respect to $\bsm$.

The posterior covariance matrix $\Gamma_\text{post}^{\bsy, \xi}(\bsm_\text{MAP}^{\bsy, \xi})$ evaluated at $\bsm_\text{MAP}^{\bsy, \xi}$ is given by
\begin{equation}
    \Gamma^{\bsy, \xi}_\text{post}(\bsm_\text{MAP}^{\bsy, \xi}) = ({H}^{\bsy, \xi}_\text{misfit}(\bsm_\text{MAP}^{\bsy, \xi}) + \Gamma_\text{prior}^{-1})^{-1},\label{eq:disc-post-cov}
\end{equation}
where $H^{\bsy, \xi}_\text{misfit}(\bsm_\text{MAP}^{\bsy, \xi})$ is a discrete approximation of $\cH^{\bsy, \xi}_{\text{misfit}}(m_{\text{MAP}})$ in \eqref{eq:c-hmisfit} given by 
\begin{equation}\label{eq:GN-hessian}
H^{\bsy, \xi}_\text{misfit}(\bsm_\text{MAP}^{\bsy, \xi}) \approx H^{\text{GN}}_\text{misfit}(\bsm_\text{MAP}^{\bsy, \xi}) = (\nabla_\bsm F_\xi(\bsm_{\text{MAP}}^{\bsy, \xi}))^T \Gamma_{\text{noise}}^{-1} \nabla_\bsm F_\xi(\bsm_{\text{MAP}}^{\bsy, \xi}).
\end{equation}

The posterior covariance matrix $\Gamma_\text{post}$ is high-dimensional and involves the Jacobian $\nabla_\bsm F_\xi$ and its transpose which requires solving the PDE and its adjoint. To efficiently compute the posterior covariance, we employ a low-rank decomposition by first solving a generalized eigenvalue problem 
\begin{equation}\label{eq:disc-gen-eig}
  {H}^{\bsy, \xi}_\text{misfit}(\bsm_\text{MAP}^{\bsy, \xi})\bsw_i = \lambda_i\Gamma_\text{prior}^{-1}\bsw_i, \ \ \   i = 1,...,r, \\
\end{equation}
with eigenpairs $(\lambda_i, \bsw_i)$, $i = 1, \dots, r$, corresponding to the largest $r$ eigenvalues $\lambda_1 \geq \cdots \geq \lambda_r$, and $\bsw_i^T \Gamma_\text{prior}^{-1} \bsw_j = \delta_{ij}$ with $i,j = 1, \dots, r$. It requires $O(r)$ linearized PDE solves using a randomized SVD algorithm \cite{VillaPetraGhattas21}, see \cite{WuChenGhattas23a} for more details. Let $\Lambda_r = \text{diag}(\lambda_1, \dots, \lambda_r)$, $W_r = (\bsw_1, \dots, \bsw_r)$, and $V_r = \Gamma_\text{prior}^{-1/2} W_r$, we have the low-rank approximation from \eqref{eq:disc-gen-eig} 
\begin{equation}\label{eq:p-pcon-hess-r}
  \tilde{H}_\text{misfit}^{\bsy, \xi} (\bsm_\text{MAP}^{\bsy, \xi}) = \Gamma_\text{prior}^{1/2}{H}^{\bsy, \xi}_\text{misfit}(\bsm_\text{MAP}^{\bsy, \xi}) \Gamma_\text{prior}^{1/2} \approx V_r\Lambda_r V_r^T, 
\end{equation}
which is accurate for a relatively small $r$ as long as the eigenvalues decay fast with $\lambda_{r+1} \ll 1$. 
To this end, the posterior covariance can be approximated by 
\begin{equation}
\label{eq:cov-approx}
  \begin{split}
   \Gamma^{\bsy, \xi}_\text{post}(\bsm_\text{MAP}^{\bsy, \xi}) &= \Gamma_\text{prior}^{1/2}(I + \tilde{H}_\text{misfit}^{\bsy, \xi} (\bsm_\text{MAP}^{\bsy, \xi}))^{-1}\Gamma_\text{prior}^{1/2}\\
   & \approx \Gamma_\text{prior}^{1/2}(I + V_r\Lambda_r V_r^T)^{-1} \Gamma_\text{prior}^{1/2} \\
   & = \Gamma_\text{prior}^{1/2}(I - V_rD_rV_r^T)\Gamma_\text{prior}^{1/2}\\ 
   & = \Gamma_\text{prior} - W_{r}D_{r}W_{r}^T,
  \end{split} 
\end{equation}
where we factorized $\Gamma_{\text{prior}}^{1/2}$ in the first equality, used low-rank approximation \eqref{eq:p-pcon-hess-r} in the inequality, employed Sherman--Morrison--Woodbury formula for the second equality with $D_r = \text{diag}(\lambda_1/(1+\lambda_1), \dots, \lambda_r/(1+\lambda_r))$, and used the definition of $V_r$ with $\Gamma_\text{prior}^{1/2}V_r = W_r$ in the last equality. 

\subsection{High-fidelity Bayesian OED}
\label{sec:Optimality}
By applying the Laplace and low-rank approximations of the posterior, we can compute the approximate A-optimality optimal design corresponding to \eqref{eq:c-a-opt} as 
\begin{equation}\label{eq:a-final-approx-W}
    \begin{split}
  \xi_A^* 
  = \arg\min_\xi \bE_{\pi(\bsy|\xi)} \left[\text{trace}(\Gamma_\text{prior} - W_rD_{r}W_r^T)\right] = \arg\max_\xi \bE_{\pi(\bsy|\xi)} \left[\text{trace}(W_rD_{r}W_r^T)\right],    
    \end{split}
\end{equation}
where the equality is due to that $\Gamma_{\text{prior}}$ does not depend on $\bsy$ and $\xi$. For simplicity, we consider a simplified A-optimality criterion as $\text{trace}(V_r D_r V_r^T) = \text{trace} (D_r V_r^T V_r) = \text{trace} (D_r)$ as $V_r^T V_r = I_r$, so that 
\begin{equation}\label{eq:a-final-approx}
    \xi_A^* = \arg\max_{\xi} \bE_{\pi(\bsy|\xi)} \left[\sum_{i = 1}^r \frac{\lambda_i}{1+ \lambda_i}\right].
\end{equation}
Note that this definition of the A-optimality is consistent with the definition of the D-optimality \eqref{eq:c-d-opt} as to minimize the trace and log determinant of the same quantity $\Gamma_{\text{prior}}^{-1/2} \Gamma^{\bsy, \xi}_\text{post}(\bsm_\text{MAP}^{\bsy, \xi}) \Gamma_{\text{prior}}^{-1/2}$, respectively, which is equivalent to maximize the trace of $V_r D_r V_r^T$ and log determinant of $I + V_r \Lambda_r V_r^T$ by \eqref{eq:cov-approx}. This leads to the optimal design of the approximate D-optimality as 

\begin{equation}\label{eq:d-final-approx}
\begin{split}
  \xi_D^* 
  = \arg\max_\xi \bE_{\pi(\bsy|\xi)} \left[ \log \text{det}(I + V_r \Lambda_r V_r^T) \right] = \arg\max_{\xi} \bE_{\pi(\bsy|\xi)} \left[\sum_{i=1}^r \log (1+\lambda_i)\right].
\end{split}
\end{equation}




For the EIG criteria defined in \eqref{eq:c-eig} with the Laplace approximation and \eqref{eq:c-kl}, we can compute optimal design of the approximate EIG with the following form \cite{WuChenGhattas23a} 
\begin{equation}\label{eq:eig-final-approx}
  \xi_{\text{EIG}}^{*} = \arg\max_\xi \bE_{\pi(\bsy|\xi)} \left[\sum_{i=1}^r\log (1 + \lambda_i) - \sum_{i=1}^r \frac{\lambda_i}{1+\lambda_i}  + ||\bsm^{\bsy, \xi}_\text{MAP}-\bsm_\text{prior}||^2_{\Gamma_\text{prior}^{-1}} \right].
\end{equation} 
Note that the first two terms correspond to the D-optimality and A-optimality criteria. 

\subsection{Computational challenges}
We remark that both the eigenpairs $(\lambda_i, \bsw_i)$, $i = 1, \dots, r$, and the MAP point $\bsm^{\bsy, \xi}_{\text{MAP}}$ depend on the data sample $\bsy$ and the design $\xi$. The conditional expectation $\bE_{\pi(\bsy|\xi)}$ in the above three optimal design problems can be computed using sample average approximation (SAA) by drawing samples $\bsy^{(n)} = F_\xi(\bsm^{(n)}) + \bsepsilon^{(n)}$, $n = 1, \dots, N_s$, with $\bsm^{(n)}$ drawn from the prior distribution of $\bsm$ and $\varepsilon^n$ drawn from the observation noise distribution. Then we need to solve the optimization problem \eqref{eq:MAP} for the MAP point and the generalized eigenvalue problem \eqref{eq:disc-gen-eig} for the eigenpairs at each of the $N_s$ samples, which becomes computationally prohibitive for large $N_s$ and large-scale PDEs to solve. We point out that the authors of \cite{WuChenGhattas23a} replace the MAP point $\bsm^{\bsy^{(n)}, \xi}_\text{MAP}$ by the prior sample $\bsm^{(n)}$, which avoids solving the optimization problem \eqref{eq:MAP}. Moreover, they assume uncorrelated observation noise $\bsepsilon^{(n)}$ in each observation dimension, by which they only need to solve a generalized eigenvalue problem in high dimensions once and then solve $N_s$ eigenvalue problems of size $r \times r$ for each design $\xi$ in the design optimization. However, the MAP point may differ from the prior sample, especially for a small number of observables, which can make their method less effective. The computational method in solving the generalized eigenvalue problem also becomes not applicable for correlated observation noise. In this work, we propose to efficiently compute both the MAP point and the eigenpairs by using derivative-enhanced surrogate models, which completely avoids solving the PDE models during the optimization for the optimal design once the surrogate models are constructed. Moreover, we do not need the observation noise (as required in \cite{WuChenGhattas23a}) to be uncorrelated to use the surrogate models.

\section{Bayesian OED with DINO} 
\label{sec:neural}
In this section, we present our approach for scalable and efficient computation of the optimality criteria based on DINO \cite{OLeary-RoseberryChenVillEtAl24} and introduce a swapping greedy algorithm modified from \cite{WuChenGhattas23a} to optimize the experimental design. To achieve the scalability of the DINO approximation with accurate derivative evaluation, we employ derivative-informed dimension reduction methods \cite{OLeary-RoseberryVillaChenEtAl22, zahm2020gradient} to reduce the dimensions of both the input parameters and output observables. For the efficient computation of the optimality criteria, we formulate both the optimization of the MAP point problem \eqref{eq:MAP} and the generalized eigenvalue problem \eqref{eq:disc-gen-eig} in the derivative-informed subspaces. Additionally, we present a complexity analysis of the high-fidelity and DINO approximations regarding the number of PDE solves for the full Bayesian OED optimization.

\subsection{Derivative-informed dimension reduction} 
\label{sec:DIDR}

The random field parameter $m$ defined in the function space $\mathcal{M}$ is infinite-dimensional. By the high-fidelity discretization of $m$ in Section \ref{sec:high-fidelity-discretization}, the dimension $d_\bsm$ of discrete parameters $\bsm$ is often remarkably high. Meanwhile, the dimension $d_F$ of candidate sensors or observables $F$ is also set high to select the most informative sensors. The high dimensionality of input parameters and output observables makes building an accurate neural network surrogate of the PtO map challenging without a sufficiently large amount of training data and a significant training cost. To address this challenge, we present derivative-informed dimension reduction and a commonly used {\color{black}principal component analysis} (PCA) for comparison.



In both the computation of the MAP point by a Newton-CG optimizer to solve \eqref{eq:MAP} and the evaluation of the Gauss--Newton Hessian in \eqref{eq:GN-hessian}, we need to compute the derivative $\nabla_\bsm F(\bsm)$. This motivates a derivative-informed dimension reduction for the input parameter $\bsm$. Specifically, let 
\begin{equation}\label{eq:meta}
    \bsm(\bseta) = \bsm_{\text{prior}} + \Gamma_{\text{prior}}^{1/2} \bseta
\end{equation}
with the whitened noise $\bseta \sim \mathcal{N}(0, I)$. We can formally write the eigenvalue decomposition for input dimension reduction as  
\begin{equation}\label{eq:eigenias}
  \bE_\bseta \left[
    \nabla_\bseta F^T(\bsm(\bseta)) \,  \nabla_\bseta F(\bsm(\bseta)) 
  \right] \bsvarphi^{(i)} = \lambda_i \bsvarphi^{(i)}, \quad i = 1, \dots, r,
\end{equation}
with $\lambda_1 \geq \cdots \geq \lambda_r$ and $(\bsvarphi^{(i)})^T \bsvarphi^{(j)} = \delta_{ij}$, which are equivalent to compute \cite{OLeary-RoseberryVillaChenEtAl22, zahm2020gradient}
\begin{equation}\label{eq:gIAS}
  \bE_\bsm \left[
    \nabla_\bsm F^T(\bsm) \,  \nabla_\bsm F(\bsm) 
  \right] \bspsi^{(i)} = \lambda_i \Gamma_{\text{prior}}^{-1}\bspsi^{(i)}, \quad i = 1, \dots, r,
\end{equation} 
with $\bspsi^{(i)} = \Gamma_{\text{prior}}^{1/2} \bsvarphi^{(i)}$ and $(\bspsi^{(i)})^T \Gamma_{\text{prior}}^{-1} \bspsi^{(j)} = \delta_{ij}$. The expectation can be computed by SAA. To this end, the input parameter dimension reduction can be computed by the linear projection 
\begin{equation}\label{eq:ias}
  \bsm \approx \bsm_r : = \bsm_{\text{prior}} + \sum_{i = 1}^r \beta_i \bspsi^{(i)} \quad \text{where} \quad \beta_i = (\bsm -  \bsm_{\text{prior}})^T \Gamma_{\text{prior}}^{-1} \bspsi^{(i)}. 
\end{equation}
We also consider a derivative-informed output dimension reduction by writing 
\begin{equation}
  \bE_\bseta \left[
    \nabla_\bseta F(\bsm(\bseta)) \,  \nabla_\bseta F^T(\bsm(\bseta)) 
  \right] \bspsi^{(i)} = \lambda_i \bspsi^{(i)}, \quad i = 1, \dots, r,\label{eq:DIS}
\end{equation}
with $\lambda_1 \geq \cdots \geq \lambda_r$ and $(\bspsi^{(i)})^T \bspsi^{(j)} = \delta_{ij}$, which are equivalent to compute 
\begin{equation}\label{eq:gOAS}
  \bE_\bsm \left[
    \nabla_\bsm F(\bsm) \, \Gamma_{\text{prior}} \, \nabla_\bsm F^T(\bsm) 
  \right] \bspsi^{(i)} = \lambda_i \bspsi^{(i)}, \quad i = 1, \dots, r,
\end{equation}
where the expectation is computed by SAA. Then, the projected output observables can be computed by 
\begin{equation}\label{eq:oas}
  F \approx F_r := \bar{F} + \sum_{i = 1}^r \beta_i \bspsi^{(i)} \quad \text{where} \quad \beta_i = (F - \bar{F})^T \bspsi^{(i)},
\end{equation}
where $\bar{F}$ is a the sample mean of $F$. To distinguish the eigenvectors for the input and output dimension reduction, we use $\bspsi_\bsm^{(i)}$ for $\bsm$ and $\bspsi_F^{(i)}$ for $F$. We also name $\Psi_\bsm = (\bspsi_\bsm^{(1)}, \dots, \bspsi_\bsm^{(r)})$ and $\Psi_F = (\bspsi_F^{(1)}, \dots, \bspsi_F^{(r)})$ as the derivative-informed input subspace (DIS) basis and derivative-informed output subspace (DOS) basis.

We present a commonly used PCA dimension reduction to compare the derivative-informed dimension reduction. Given $N$ samples of the input-output pairs $(\bsm^{(n)}, F^{(n)})$, $n = 1, \dots, N$, where $F^{(n)}$ is computed by solving the PDE model at $\bsm^{(n)}$, we can perform dimension reduction of them by PCA or equivalently a truncated singular value decomposition (tSVD). Let matrix $B = (F^{(1)} - \bar{F}, \dots, F^{(N)}- \bar{F})$ with sample mean $\bar{F}$. The tSVD of $B$, which can be computed by a randomized SVD algorithm, is given by 
\begin{equation}\label{eq:svd}
  B \approx B_r := \Psi_r \Sigma_r \Phi_r^T,
\end{equation}
where $\Psi_r = (\bspsi^{(1)}, \dots, \bspsi^{(r)})$ and $\Phi_r = (\bsphi^{(1)}, \dots, \bsphi^{(r)})$ are the matrices with $r$ columns of the left and right singular vectors respectively corresponding to the $r$ largest singular values $\sigma_1 \geq \dots \geq \sigma_r$ with $\Sigma_r = \text{diag}(\sigma_1, \dots, \sigma_r)$. For any new $F$, we can perform the dimension reduction by the linear projection 
\begin{equation}\label{eq:pca}
  F \approx F_r := \bar{F} + \sum_{i = 1}^r \beta_i \bspsi^{(i)} \quad \text{where } \quad \beta_i = (F - \bar{F})^T \bspsi^{(i)}.
\end{equation}
Note that instead of using the sample-based PCA dimension reduction for the input parameter, we can also compute a discrete version of a truncated Karhunen--Lo\`eve expansion (KLE) for $\bsm \sim \mathcal{N}(\bsm_{\text{prior}}, \Gamma_{\text{prior}})$ as 
\begin{equation}\label{eq:kle}
  \bsm \approx \bsm_r := \bsm_{\text{prior}} + \sum_{i = 1}^r \beta_i \bspsi^{(i)} \quad \text{with} \quad \beta_i = \sqrt{\lambda_i} \eta_i,
\end{equation}
where $(\lambda_i, \bspsi^{(i)})$, $i = 1, \dots, r$, are the eigenpairs of the covariance matrix $\Gamma_\text{prior}$, and $\eta_i \sim \mathcal{N}(0, 1)$. 
The sample-based PCA projection converges to the KLE projection with increasing samples. 

\subsection{Derivative-informed neural operators}
\label{sec:DINN}

Based on the dimension reduction for the input and output, we construct a neural network surrogate $F_{\text{NN}}(\bsm)$ of the PtO map $F(\bsm)$, or the discrete observation operator (thus neural operator), as
\begin{equation}\label{eq:Fnn}
    F(\bsm) \approx F_{\text{NN}}(\bsm) :=  \cD_{\Psi_F} \circ \Phi_{\bstheta} \circ \cE_{\Psi_\bsm} (\bsm),
\end{equation}
where $\cE_{\Psi_\bsm} : \bR^{d_\bsm} \to \bR^{r_\bsm}$ represents an encoder defined by the linear projection with basis $\Psi_\bsm \in \bR^{d_\bsm \times r_\bsm}$ as in \eqref{eq:ias} by DIS or in \eqref{eq:kle} by KLE. The output of the encoder is the $r_\bsm$-dimensional coefficient vector $\cE_{\Psi_\bsm}(\bsm) = \bsbeta_\bsm = (\beta_1, \dots, \beta_{r_\bsm})^T$ defined in the linear projections. $\cD_{\Psi_F}: \bR^{r_F} \to \bR^{d_F}$ represents a decoder by the linear projection with basis $\Psi_F \in \bR^{d_F \times r_F}$ as in \eqref{eq:oas} by DOS or in \eqref{eq:pca} by PCA. Note that the full dimension of the observable vector $F$ is $d_F$, the same as the number of candidate sensors $d_s$. 
The map $\Phi_\bstheta: \bR^{r_\bsm} \to \bR^{r_F}$ is given by a neural network parametrized by parameters $\bstheta$, which takes the $r_\bsm$-dimensional input $\bsbeta_\bsm$ and maps it to a $r_F$-dimensional output to approximate the coefficient vector $\bsbeta_F$ of the linear projection of $F$ as in \eqref{eq:oas} by DOS or in \eqref{eq:pca} by PCA.
One example is given by 
\begin{equation}\label{eq:PhiTheta}
  \Phi_\bstheta (\bsz) := Z_L \bsz_L + \bsb_L, 
\end{equation}
where  
\begin{equation}
  \bsz_{l+1} = \bsz_{l} + W_l h_l(Z_l \bsz_{l} + \bsb_l), \quad l = 0, \dots, L-1,
\end{equation}
is a sequence of ResNets \cite{he2016deep} with input $\bsz_0 = \bsz$, $h_l$ represents elementwise nonlinear activation functions, $\bstheta = (W_0, Z_0, \bsb_0, \dots, W_{L-1}, Z_{L-1}, \bsb_{L-1}, Z_L, \bsb_L)$ collects all the neural network parameters. Note that the neural network size depends only on $r_F$ and $r_\bsm$, which is small, so a relatively small number of data would be sufficient for the training of the neural network.

To generate the training data, we draw $N_t$ random samples $\bsm^{(n)}$ and solve the PDE to obtain $F^{(n)}$, $n = 1, \dots, N_t$. Then we project them to their corresponding reduced basis and obtain $\bsbeta_\bsm^{(n)}$ and $\bsbeta_F^{(n)}$. Note that the samples in computing the projection basis in the last section do not need the same samples drawn here. In practice, we take a subset of samples drawn here to compute the projection basis first. To construct the neural network surrogate that is accurate not only for the PtO map $F(\bsm)$ but also for its Jacobian $J(\bsm) = \nabla_\bsm F(\bsm)$, we also compute a reduced Jacobian $J_r^{(n)} = \Psi_F^T J^{(n)} \Psi_\bsm \in \bR^{r_F \times r_\bsm}$ with the full Jacobian $J^{(n)} = \nabla_\bsm F^{(n)}$ at each of the $N$ samples. This computation requires solving $\min(r_F, r_\bsm)$ linearized PDEs corresponding to the PDE model \eqref{eq:PDE}. These linearized PDEs have the same linear operator $\partial_u \mathcal{R}$ or its adjoint at each sample $\bsm^{(n)}$. The computational cost can be amortized by, e.g., one LU factorization of the linear operator and repeated use of this factorization as a direct solver.

To this end, we define the loss function for the training of the neural network with two terms as
\begin{equation}\label{eq:loss}
  \ell(\bstheta) = \frac{1}{N_t}\sum_{n = 1}^{N_t} ||\bsbeta_F^{(n)} - \Phi_\bstheta(\bsbeta_\bsm^{(n)})||_2^2 + \lambda ||J_r^{(n)} - \nabla_\bsbeta \Phi_\bstheta(\bsbeta_\bsm^{(n)})||_2^2,   
\end{equation}
where the first term measures the difference between the projected PtO map and its neural network approximation in Euclidean norm, and the second term is informed by the derivative (thus DINO \cite{OLeary-RoseberryChenVillEtAl24}) and measures the difference of their Jacobians in Frobenius norm. Note that $\nabla_\bsbeta \Phi_\bstheta \in \bR^{r_F \times r_\bsm}$ represents the Jacobian of the neural network output $\Phi_\bstheta$ with respect to its input, which can be computed by automatic differentiation. The evaluation of the loss function and its gradient with respect to the parameters $\bstheta$ is very efficient since all quantities depend only on the small reduced dimensions $r_F$ and $r_\bsm$. We remark that the scaling parameter $\lambda$ in the loss function \eqref{eq:loss} can be properly tuned to balance the two terms. In practice, $\lambda = 1$ is sufficient to improve the PtO map's accuracy and its reduced Jacobian.

\subsection{Efficient computation of the MAP point and the eigenpairs with DINO}
\label{sec:nn-oed}
Once the neural network $\Phi_\bstheta$ defined in \eqref{eq:PhiTheta} is well trained with the derivative-informed loss function \eqref{eq:loss}, we obtain the optimized DINO surrogate $F_{\text{NN}}$ of the PtO map $F$ as in \eqref{eq:Fnn}. We can use it to solve problem \eqref{eq:MAP} for the MAP point and compute the eigenpairs in \eqref{eq:disc-gen-eig} that involve the Jacobian as in \eqref{eq:GN-hessian} using its reliable derivative information. We further discuss the error bound for this result in Section \ref{sec:ErrorAnalysis}. By employing this proposed method, we can significantly reduce the computational cost of these operations.

Specifically, we solve the following optimization problem to compute the MAP point, 
\begin{equation}\label{eq:rMAP}
  \bsbeta^{\bsy, \xi}_{\text{MAP}} = \arg\min_{\bsbeta}  \frac{1}{2} ||\bsy - \xi \,\mathcal{D}_{\Psi_F} \circ \Phi_\bstheta( \bsbeta)||^2_{\Gamma^{-1}_{\text{noise}}} + \frac{1}{2} ||\bsbeta||^2_{\Gamma^{-1}_\bsm},
\end{equation}
where $\Gamma^{-1}_\bsm = \Psi_\bsm^T\Gamma^{-1}_{\text{prior}}\Psi_\bsm \in \bR^{r_\bsm \times r_\bsm}$ can be computed once and used repeatedly in solving the optimization problem. When $\Psi_\bsm$ is taken as the DIS basis computed from \eqref{eq:gIAS}, we have $\Gamma^{-1}_\bsm = \bI$. Note that this optimization problem has a reduced dimension $r_\bsm \ll d_\bsm$, compared to the high $d_\bsm$-dimensional optimization problem \eqref{eq:MAP}. We propose to solve it using a gradient-based optimization algorithm, e.g., the quasi--Newton LBFGS algorithm, where the derivative of the neural network $\Phi_\bstheta(\bsbeta)$ with respect to the input $\bsbeta$ can be efficiently computed by automatic differentiation. Once $\bsbeta^{\bsy, \xi}_{\text{MAP}}$ is obtained, we can compute the approximate high-dimensional MAP point $\bsm_{\text{MAP}}^{\bsy, \xi}$ by the linear projection, e.g., DIS in \eqref{eq:ias}. Further error analysis for MAP point estimation is discussed in Section \ref{sec:ErrorAnalysis-map}







To compute the eigenpairs of \eqref{eq:disc-gen-eig} at the MAP point using the DINO surrogate $F_{\text{NN}}$ in \eqref{eq:Fnn}, we first observe that its Jacobian is given by 
\begin{equation}\label{eq:red-jac}
  \nabla_\bsm F_{\text{NN}}(\bsm_{\text{MAP}}^{\bsy, \xi}) = \Psi_F \, \nabla_\bsbeta \Phi_\bstheta (\bsbeta^{\bsy, \xi}_{\text{MAP}}) \, \Psi_\bsm^T \Gamma^{-1}_{\text{prior}} ,
\end{equation}
where we use the DIS basis $\Psi_\bsm$ for the input encoder, and $\Psi_F$, either PCA basis or DOS basis for the output decoder.
Second, by observing the similarity of the generalized eigenvalue problem \eqref{eq:disc-gen-eig} and the generalized eigenvalue problem \eqref{eq:gIAS} in computing the DIS basis, we propose to approximate the eigenvector $\bsw_i$ of \eqref{eq:disc-gen-eig} by the linear projection with DIS basis as $\hat{\bsw}_i = \Psi_\bsm \bsu_i \in \bR^{d_\bsm}$ with the coefficient vector $\bsu_i \in \bR^{r_\bsm}$, then the generalized eigenvalue problem \eqref{eq:disc-gen-eig} can be simplified (by multiplying $\Psi_\bsm$ on both sides) as 
\begin{equation}\label{eq:rEig}
  \hat{H}^{\bsy, \xi}_\text{misfit}(\bsbeta_\text{MAP}^{\bsy, \xi}) \bsu_i = \hat{\lambda}_i \bsu_i, \quad i = 1, \dots, r_\bsm, 
\end{equation}
where the reduced matrix $\hat{H}^{\bsy, \xi}_\text{misfit}(\bsbeta_\text{MAP}^{\bsy, \xi}) \in \bR^{r_\bsm \times r_\bsm}$ is given by 
\begin{equation}\label{eq:rEigen}
  \hat{H}^{\bsy, \xi}_\text{misfit}(\bsbeta_\text{MAP}^{\bsy, \xi}) =  (\nabla_\bsbeta \Phi_\bstheta (\bsbeta^{\bsy, \xi}_{\text{MAP}}))^T \, \Psi_F^T \, \xi^T \, \Gamma_{\text{noise}}^{-1} \xi \, \Psi_F \, \nabla_\bsbeta \Phi_\bstheta (\bsbeta^{\bsy, \xi}_{\text{MAP}}),
\end{equation}
which can be efficiently evaluated by automatic differentiation. Moreover, the eigenvalue problem \eqref{eq:rEigen} can also be efficiently solved in the reduced dimension $r_\bsm$. Once the MAP point and the eigenpairs are computed, we can evaluate all the optimality criteria as in Section \ref{sec:Optimality}. We also analyze the error bound for the eigenvalues in Section \ref{sec:ErrorAnalysis-eig} and for the optimality criteria in Section \ref{sec:error-Opt}.

\subsection{Swapping greedy algorithm}
This section presents a swapping greedy algorithm as a modification developed in \cite{WuChenGhattas23a} to optimize the experimental design for the optimal sensor placement problem, inspired by \cite{lourencco2003iterated, zhang2002feature}. Recall that we need to choose $r_s$ sensors out of $d_s$ candidate sensor locations. We first apply a standard greedy algorithm to place the sensor one by one from $1$ to $r_s$ sensors according to one of the optimality criteria, and then employ a swapping greedy algorithm to swap the selected $r_s$ sensors one by one with the rest of the sensors to improve the optimality criterion until a stopping condition is satisfied, see Algorithm \ref{alg:buildtree}. The optimality can be the trace/determinant of the posterior covariance matrix and EIG. More specifically, we initialize an empty sensor set in line \ref{line:S0} and then select and add the sensors $s_t^*$, $t = 1, \dots, r_s$, one by one by a greedy algorithm in line \ref{line:greedy-start} to \ref{line:greedy-end} to maximize one of the optimality criteria, where the optimality criteria are computed either by high-fidelity approximation as in Section \ref{sec:Optimality}, or by neural network acceleration as in Section \ref{sec:nn-oed}. Then, we start from the selected sensor set by the greedy algorithm and swap the sensors from it one by one with the rest of the candidate sensors in line \ref{line:swap-start} to \ref{line:swap-end}. We stop the swapping loop if the improved design optimality is smaller than tolerance or if the number of the swapping loops is larger than a maximum number $k_{\text{max}}$. 

We remark that a leverage score in \cite{WuChenGhattas23a} selects an initial sensor set with $r_s$ sensors before starting the swapping loop since the optimality criteria can only be evaluated in an online optimization stage for a fixed number ($r_s$) of sensors in its algorithm. In contrast, we can efficiently assess the optimality criteria for an arbitrary number of sensors, which allows us to initialize the sensor set by a greedy algorithm before swapping. We observe fewer swapping loops to converge in our example (1 to 2 loops are sufficient).

\begin{algorithm}[!htb]
\caption{Swapping greedy algorithm}
\label{alg:buildtree}
\begin{algorithmic}[1]
\STATE{\textbf{Input}: A set $S = \{s_1, \dots, s_{d_s}\}$ of $d_s$ candidate sensors, a budget of $r_s$ sensors, a maximum number of swapping loops $k_{\text{max}}$, and an optimality tolerance $\varepsilon_{\text{min}}$.}
\STATE{\textbf{Output}: A set $S^*$ of $r_s$ optimal sensors.}
\vskip 0.2cm
\STATE{Set an empty sensor set $S^0 = \emptyset$. \label{line:S0}} 
\FOR{$t = 1, \dots, r_s$ \label{line:greedy-start}} 
\STATE{Choose the optimal sensor $s_t^*$ at step $t$ as}
\STATE{$$
s_t^* = \arg\max_{s \in S \setminus S^{t-1}} \text{optimality}(\xi(\{s\} \cup S^{t-1})),
$$}
\STATE{where $\xi(\{s\} \cup S^{t-1})$ is the design matrix for the selected sensors $\{s\} \cup S^{t-1}$.}
\STATE{Update the set of selected sensors $S^t$ as}
\STATE{$$
S^t = \{s_t^*\} \cup S^{t-1}.
$$}
\ENDFOR \label{line:greedy-end}
\STATE{Set an initial loop number $k = 0$ and design improvement $\varepsilon = 2 \varepsilon_{\text{min}}$.}
\WHILE{$k \leq k_{\text{max}}$ and $\varepsilon > \varepsilon_{\text{min}}$ \label{line:swap-start}}
\STATE{Denote a set of $r_s$ selected sensors as $S^{r_s}_k = S^{r_s} = \{s_1, \dots, s_{r_s}\}$, and compute optimality$(\xi(S^{r_s}_k))$ for the corresponding experimental design $\xi(S^{r_s}_k)$.}
\FOR{$t = 1, \dots, r_s$}
\STATE{Select the candidate sensor $s_t^*$ to swap with the $t$-th sensor $s_t$ in $S^{r_s}$ as}
\STATE{$$
s_t^* = \arg\max_{s \in S\setminus S^{r_s}} \text{optimality}(\xi(S^{r_s} \setminus \{s_t\} \cup \{s\}))
$$}
\STATE{Update $S^{r_s}$ by swapping its $t$-th sensor $s_t$ with $s_t^*$.}
\ENDFOR
\STATE{Compute the optimality$(\xi(S^{r_s}))$ for the updated sensor set $S^{r_s}$, compute the incremental improvement of the optimality criterion}
\STATE{
$$
\varepsilon = \text{optimality}(\xi(S^{r_s})) - \text{optimality}(\xi(S^{r_s}_k)),
$$
}
\STATE{and update $k \leftarrow k + 1$.}
\ENDWHILE \label{line:swap-end}
\RETURN{$S^* = S^{r_s}$.}
\end{algorithmic}
\end{algorithm}

\subsection{Computational complexity}
\label{sec:complexity}
The total computation for the swapping greedy algorithm requires $O(k_s r_s d_s)$ evaluations of the optimality criterion to select $r_s$ sensors from $d_s$ candidates with $k_s \leq k_{\text{max}}$ loops of the swapping greedy algorithm. Note that the greedy algorithm for initialization takes $O(r_s d_s)$ evaluations of the optimality criterion. Each of the optimality criteria takes $N_s$ times of solving the optimization problem \eqref{eq:rMAP} to compute the MAP point and $N_s$ times of solving the eigenvalue problem \eqref{eq:rEig} to compute the eigenpairs, where $N_s$ is the number of samples in the SAA of the expectation in the optimality criteria. 

If we only use the high-fidelity approximation without neural network acceleration, by the same swapping greedy algorithm, we would have to solve $O(N_s k_s r_s d_s)$ times the optimization problem \eqref{eq:MAP} to compute the high-fidelity MAP point and solve $O(N_s k_s r_s d_s)$ times the high-fidelity generalized eigenvalue problem \eqref{eq:disc-gen-eig} to compute the eigenpairs. By using an inexact Newton-CG algorithm to solve the optimization problem \eqref{eq:MAP}, we need to solve $N_{\text{nt}}$ state PDEs and $N_{\text{nt}} N_{\text{cg}}$ linearized PDEs (for Hessian-vector product), with $N_{\text{nt}}$ Newton iterations and $N_{\text{cg}}$ (in average) CG iterations per Newton iteration. For the solution of the generalized eigenvalue problem \eqref{eq:disc-gen-eig} by a randomized SVD algorithm, we need to solve one state PDE and $O(r)$ linearized PDEs to compute $r$ eigenpairs. We remark that if a direct solver is used, e.g., by first LU factorization and then its repeated use to solve the linearized PDEs, then the cost $C_2$ of each linearized PDE solve becomes much smaller than the cost $C_1$ to solve the state PDE. As a result, the additional computational cost for the Jacobian data generation is comparable or smaller (depending on the nonlinearity of the state PDEs) than that for the PtO map data generation, as reported in our examples. The computational complexity in terms of the number of PDE solves is summarized in Table \ref{tb:comp}. 

For the neural network construction, we first need to build the basis for the input and output dimension reduction, with $O(N_\bsm)$ state PDE solves (using $O(N_\bsm)$ samples for SAA in \eqref{eq:gIAS}) and $O(N_\bsm r_\bsm)$ linearized PDE solves to get $r_\bsm$ bases for input dimension reduction, and $O(N_F)$ state and $O(N_F r_F)$ linearized PDE solves to get $r_F$ bases for output dimension reduction. To generate $N_t$ samples of the PtO map and its Jacobian in the reduced dimension for training the neural network surrogate, we need to solve 
$N_t$ state PDEs (for PtO map), and $O(N_t r_t)$ linearized PDEs (for Jacobian) with $r_t = \min(r_\bsm, r_F)$. As the neural network is constructed in the reduced dimensions with a relatively small $O(r_\bsm r_F)$ degrees of freedom, its training cost is negligible compared to the data generation. Once trained in an offline stage, the evaluation of the neural network and its Jacobian is much faster than the large-scale PDE solves and can be efficiently used to solve the Bayesian OED problem without the PDE solves in the online stage. The computational complexity in terms of the number of PDE solves is also summarized in Table \ref{tb:comp}. 

\begin{table}[h]
  \footnotesize
  \begin{center}
    \begin{tabular}{|c|c|c|c|}
      \hline
      \# PDE solves & High-fidelity \\ \hline
      MAP point &  $O(N_s k_s r_s d_s N_{nt} )C_1 + O(N_s k_s r_s d_s N_{nt} N_{cg})C_2$ \\ \hline
      Eigenvalue problem & $O(N_s k_s r_s d_s)C_1 + O(N_s k_s r_s d_s r)C_2$ \\ \hline \hline
      \# PDE solves & DINO \\ \hline
      Dimension reduction & $ O(N_\bsm + N_F)C_1 + O(N_\bsm r_\bsm + N_F r_F)C_2$ \\
      \hline
      Training data & $O(N_t)C_1 + O(N_t r_t)C_2$ \\ \hline
    \end{tabular}
    \caption{Computational complexity in terms of the number of PDE solves, each state PDE with cost $C_1$ and each linearized PDE with cost $C_2$. $N_s$: \# samples to compute optimality criterion, $k_s$: \# swapping loops, $r_s$: \# sensors to select, $d_s$: \# candidate sensors, $N_{nt}$: \# Newton iterations, $N_{cg}$: \# CG iterations per Newton iteration, $r$: \# eigenpairs, $N_t$: \# training samples, $r_t = \min(r_\bsm, r_F)$: the smaller number of the dimensions of the input projection $r_\bsm$ and the output projection $r_F$, $N_\bsm$ and $N_F$: \# samples used to construct the input and output dimension reduction.}\label{tb:comp}
  \end{center}
  \end{table}

\section{Error analysis}

This section provides detailed error analysis for the approximations of the MAP point, the eigenvalues of the approximate posterior covariance, and the optimality criteria under assumptions of the accuracy of the neural surrogate and the dimension reduction, and the smoothness of the objective function in computing the MAP point. 
We validate these assumptions in the numerical results presented in Section \ref{sec:numerical}. 

\label{sec:ErrorAnalysis}
\subsection{Approximation error of the MAP point}
\label{sec:ErrorAnalysis-map}

We consider the parameter $\bsm$ in its representation \eqref{eq:meta} as a function of the white noise $\bseta$. By slight abuse of notation, we write the forward map as $F_\xi(\bseta) = F_\xi(\bsm(\bseta))$. Let the objective function in \eqref{eq:MAP} for the high-fidelity MAP point be rewritten as 
\begin{equation}\label{eq:objMAP}
    L(\bseta) = \frac{1}{2}||\bsy - F_\xi(\bseta)||^2_{\Gamma_\text{noise}^{-1}}+\frac{1}{2}||\bseta||^2
\end{equation}
for any design $\xi$ and data $\bsy$. Let $\hat{F}_\xi$ denote the neural network approximation of the forward map $F_\xi$ given in \eqref{eq:Fnn}. Let $\hat{L}$ denote the corresponding objective function in \eqref{eq:rMAP}, which can be equivalently written as a function of $\bseta$ as 
\begin{equation}\label{eq:hatobjMAP}
    \hat{L}(\bseta) = \frac{1}{2}||\bsy- \hat{F}_\xi(\bseta)||^2_{\Gamma_\text{noise}^{-1}}+\frac{1}{2}||\Phi_\bsm^T \bseta||^2,
\end{equation} 
where $\Phi_\bsm = (\bsvarphi^{(1)}, \dots, \bsvarphi^{(r)})$ with $\bsvarphi^{(i)}$ as the eigenvectors of \eqref{eq:eigenias}, which satisfies $\Phi_\bsm^T \Phi_\bsm = I_r$. Moreover, we denote the MAP point obtained from the high-fidelity optimization problem \eqref{eq:MAP} as $\bsm^* = \bsm(\bseta^*)$ and the MAP point obtained from the projected optimization problem \eqref{eq:rMAP} as $\hat{\bsm}^* = \bsm(\hat{\bseta}^*)$, with the MAP points $\bseta^*$ of \eqref{eq:objMAP} and $\hat{\bseta}^*$ of \eqref{eq:hatobjMAP}, which satisfy the stationary conditions 
\begin{equation}\label{eq:staionaryM}
  0 = \nabla_\bseta L(\bseta^*) =  - (\bsy - F_\xi(\bseta^*))^T\Gamma_\text{noise}^{-1}\nabla_\bseta F_\xi(\bseta^*) + \bseta^*
\end{equation}
and 
\begin{equation}\label{eq:staionaryMhat}
  0 = \nabla_\bseta \hat{L}(\hat{\bseta}^*) = - (\bsy - \hat{F}_\xi(\hat{\bseta}^*))^T\Gamma_\text{noise}^{-1}\nabla_\bseta \hat{F}_\xi(\hat{\bseta}^*) + P_r \hat{\bseta}^*,
\end{equation}
respectively, where $P_r = \Phi_\bsm \Phi_\bsm^T$ is the projection matrix. To bound the approximation error of the MAP point obtained in \eqref{eq:rMAP} and the eigenvalues obtained in \eqref{eq:rEig} by the neural network approximation of the forward map,
we make the following assumptions on the objective function $L$ and $\hat{L}$, 
the approximation error of the forward map and its Jacobian, and the dimension reduction error of the MAP point. We use $||\cdot||$ to denote the Euclidean norm for vectors and the Frobenius norm for matrices. 

\begin{assumption}\label{ass:assumptions}
Let $\bseta^*$ and $\hat{\bseta}^*$ be the MAP points of \eqref{eq:objMAP} and \eqref{eq:hatobjMAP},  respectively.
    \begin{itemize}
        \item[A1] We assume that the gradients  $\nabla_\bseta L$ and $\nabla_\bseta \hat{L}$ allow convergent Taylor expansion around $\bseta^*$ and $\hat{\bseta}^*$\label{ass:ass-1}, 
  \begin{equation}\label{eq:Taylor1}
    \nabla_\bseta L(\bseta) = \nabla_\bseta L(\bseta^*) + \nabla_\bseta^2 L(\bseta^*)(\bseta - \bseta^*) + \delta(\bseta) (\bseta - \bseta^*),
  \end{equation}
  for $\bseta$ such that $||\bseta - \bseta^*|| \leq R_\bseta$, and \begin{equation}\label{eq:Taylor2}
    \nabla_\bseta \hat{L}(\bseta) = \nabla_\bseta \hat{L}(\hat{\bseta}^*) +\nabla_\bseta^2 \hat{L}(\hat{\bseta}^*)(\bseta - \hat{\bseta}^*) + \hat{\delta}(\bseta) (\bseta - \hat{\bseta}^*),
  \end{equation}
  for $\bseta$ such that $||\bseta - \hat{\bseta}^*|| \leq R_\bseta$, for some radius $R_\bseta$, where $\delta(\bseta) \in \mathbb{R}^{d_\bseta \times d_\bseta}$ and $\hat{\delta}(\bseta) \in \mathbb{R}^{d_\bseta \times d_\bseta}$ satisfy $\lim_{\bseta \to \bseta^*}||\delta(\bseta)|| = 0$ and $\lim_{\bseta \to \hat{\bseta}^*}||\hat{\delta}(\bseta)|| = 0$.
 In addition, we assume that
 \begin{equation}\label{eq:lowerbound}
     ||(\nabla_\bseta^2 \hat{L}(\hat{\bseta}^*) + \nabla_\bseta^2 L(\bseta^*)) (\bseta^* - \hat{\bseta}^*) + (\hat{\delta}(\bseta^*) + \delta(\hat{\bseta}^*)) (\bseta^* - \hat{\bseta}^*)|| \geq c_h ||\bseta^* - \hat{\bseta}^*||
 \end{equation}
 for some constant $c_h > 0$.    
        \item[A2] Moreover, we assume that the forward map $F_\xi$ and its Jacobian $\nabla_\bseta F_\xi$ are bounded and well approximated by the bounded neural network approximation $\hat{F}_\xi$ and its Jacobian $\nabla_\bseta \hat{F}_\xi$ in the sense that there exist small constants $\varepsilon_1 > 0$ and $\varepsilon_2 > 0$ such that\label{ass:ass-2}
  \begin{equation}\label{eq:approxFJ}
    ||F_\xi(\bseta) - \hat{F}_\xi(\bseta)|| \leq \varepsilon_1 \quad \text{ and } \quad ||\nabla_\bseta F_\xi(\bseta) - \nabla_\bseta \hat{F}_\xi(\bseta)|| \leq \varepsilon_2
  \end{equation}
  for $\bseta$ such that $||\bseta - \bseta^*|| \leq R_\bseta$ and $||\bseta - \hat{\bseta}^*|| \leq R_\bseta$, in particular for $\bseta = \bseta^*$ and $\bseta = \hat{\bseta}^*$.

    \item[A3] Furthermore, we assume that there exists a small constant $\varepsilon_3 > 0$ such that the projection error  satisfies\label{ass:ass-3}
  \begin{equation}
||(I - P_r)\bseta^*|| \leq \varepsilon_3.
  \end{equation}
    \end{itemize}
\end{assumption}

\begin{remark}
    We remark that Assumption A1 is mild and satisfied under the stronger assumption that 
 the Hessian $\nabla_{\bseta}^2 L(\bseta^*)$ and $\nabla_\bseta^2 \hat{L}(\hat{\bseta}^*)$ at the MAP points are positive definite with the smallest eigenvalues $\lambda^*_{\text{min}} \geq c_\lambda > 0$ and $\hat{\lambda}^*_{\text{min}} \geq c_\lambda > 0$ for some constant $c_\lambda$, and $||\hat{\delta}(\bseta^*)|| \leq c_\delta$ and $||\delta(\hat{\bseta}^*)|| \leq c_\delta$ for some $c_\delta < c_\lambda$, which leads to $c_h \geq 2(c_\lambda - c_\delta) > 0$ for $c_h$ in \eqref{eq:lowerbound}. 
This holds for linear forward map $F_\xi(\bseta) = A \, \bseta +\bsb$, in which case $\delta(\bseta) = 0$ and $\hat{\delta}(\bseta) = 0$, and the Hessian $\nabla_\bseta^2 L(\bseta^*) = A^T \Gamma_{\text{noise}}^{-1} A + I$ is positive definite. However, it may not be true for highly nonlinear inverse problems such that the objective function $L$ and $\hat{L}$ become highly nonconvex and allow multiple local minima far away from each other. 
\end{remark}
\begin{remark}
Assumption A2 is reasonable since the neural network approximation $\hat{F}_\xi$ is trained for both the function and its Jacobian in the DINO framework. Assumption A3 holds true for problems with intrinsic low-dimensionality of the input parameter, or the forward map is only sensitive in an intrinsically low-dimensional subspace spanned by $\Phi_\bsm$. In this case, the regularization term $||\bseta||$ in \eqref{eq:objMAP} would push $||(I - P_r)\bseta^*||$ close to zero as $F_\xi(\bseta) = F_\xi(P_r \bseta + (I - P_r) \bseta)$ is not sensitive to $(I - P_r)\bseta$. We will 
numerically investigate the decay of the approximation errors with increasing training data to demonstrate A2 and the dimension reduction error with increasing reduced dimension to demonstrate A3 in detail in Section \ref{sec:numerical}.
\end{remark}

\begin{theorem}
\label{thm:1}
  Under Assumption \ref{ass:assumptions}, if $||\bseta^* - \hat{\bseta}^*|| \leq R_\bseta$, we have the following error bound for the approximation of high-fidelity MAP point by the surrogate MAP point 
  \begin{equation}\label{eq:MAPbound}
    ||\bseta^* - \hat{\bseta}^*|| \leq c_1\varepsilon_1 + c_2\varepsilon_2 + c_3 \varepsilon_3 =: \varepsilon_\bsm,
  \end{equation}
  where $c_1 > 0, c_2>0,$ and $c_3 > 0$ are positive constants independent of $\varepsilon_1, \varepsilon_2, $ and $\varepsilon_3$.
\end{theorem}
Note that this bound implies the bound for the MAP points in the parameter space measured in $\Gamma_\text{prior}^{-1}$-norm, as $||\bsm^* - \hat{\bsm}^*||_{\Gamma^{-1}_\text{prior}} = ||\bseta^* - \hat{\bseta}^*||$, which appears in the EIG-optimality criteria in \eqref{eq:eig-final-approx}. Moreover, we have $||\bsm^* - \hat{\bsm}^*||_M \leq C_\bsm ||\bsm^* - \hat{\bsm}^*||_{\Gamma^{-1}_\text{prior}}$ for some constant $C_\bsm \leq 1/\kappa$ by the definition of $\Gamma^{-1}_\text{prior}$ in \eqref{eq:GammaInv}.

\begin{proof}
  By A1 in Assumption \ref{ass:assumptions}, and $||\bseta^* - \hat{\bseta}^*|| \leq R_\bseta$, we can evaluate the Taylor expansions \eqref{eq:Taylor1} and \eqref{eq:Taylor2} at $\bseta = \hat{\bseta}^*$ and $\bseta = \bseta^*$, respectively, and obtain 
\begin{equation}\label{eq:gradL}
    \nabla_\bseta L(\hat{\bseta}^*) = \nabla_\bseta^2 L(\bseta^*) (\hat{\bseta}^* - \bseta^*) + \delta(\hat{\bseta}^*)(\hat{\bseta}^* - \bseta^*)
\end{equation}
  and \begin{equation}\label{eq:gradhatL}
    \nabla_\bseta \hat{L}(\bseta^*) = \nabla_\bseta^2 \hat{L}(\hat{\bseta}^*) (\bseta^* - \hat{\bseta}^*) + \hat{\delta}(\bseta^*)(\bseta^* - \hat{\bseta}^*),
  \end{equation}
where we used the stationary conditions \eqref{eq:staionaryM} and \eqref{eq:staionaryMhat}. Subtracting \eqref{eq:gradL} from \eqref{eq:gradhatL}, we obtain
\begin{equation}\label{eq:hessgraddiff}
  \nabla_\bseta \hat{L}(\bseta^*) - \nabla_\bseta L(\hat{\bseta}^*) = (\nabla_\bseta^2 \hat{L}(\hat{\bseta}^*) + \nabla_\bseta^2 L(\bseta^*)) (\bseta^* - \hat{\bseta}^*) + (\delta(\hat{\bseta}^*) + \hat{\delta}(\bseta^*))(\bseta^* - \hat{\bseta}^*).
\end{equation}
By definition of $\hat{L}$ and $L$, we have for the left hand side of the above equation 
\begin{equation}\label{eq:gradLdiff}
  \begin{split}
  \nabla_\bseta \hat{L}(\bseta^*) - \nabla_\bseta L(\hat{\bseta}^*) & =  - (\bsy - \hat{F}_\xi(\bseta^*))^T\Gamma_\text{noise}^{-1}\nabla_\bseta \hat{F}_\xi(\bseta^*) + P_r\bseta^* \\
  & + (\bsy - F_\xi(\hat{\bseta}^*))^T\Gamma_\text{noise}^{-1}\nabla_\bseta F_\xi(\hat{\bseta}^*) - \hat{\bseta}^*.
  \end{split}
\end{equation}  
Let us consider the first term on the right hand side of the above equation. We have
\begin{equation}
  \begin{split}
    & - (\bsy - \hat{F}_\xi(\bseta^*))^T\Gamma_\text{noise}^{-1}\nabla_\bseta \hat{F}_\xi(\bseta^*) \\
  & = - (F(\bseta^*) - \hat{F}_\xi(\bseta^*))^T \Gamma_\text{noise}^{-1}\nabla_\bseta \hat{F}_\xi(\bseta^*) - (\bsy - F_\xi(\bseta^*))^T\Gamma_\text{noise}^{-1}\nabla_\bseta \hat{F}_\xi(\bseta^*)\\
  & = - A^T \Gamma_\text{noise}^{-1}\nabla_\bseta \hat{F}_\xi(\bseta^*) - (\bsy - F_\xi(\bseta^*))^T\Gamma_\text{noise}^{-1} (\nabla_\bseta \hat{F}_\xi(\bseta^*) - \nabla_\bseta F_\xi(\bseta^*)) \\
  & \quad - (\bsy - F_\xi(\bseta^*))^T\Gamma_\text{noise}^{-1}\nabla_\bseta F_\xi(\bseta^*)\\
  & = - A^T \Gamma_\text{noise}^{-1}\nabla_\bseta \hat{F}_\xi(\bseta^*) - (\bsy - F_\xi(\bseta^*))^T\Gamma_\text{noise}^{-1}B - \bseta^*,
  \end{split}
\end{equation}
where we added and subtracted $F_\xi(\bseta^*)$ in the first equality, added and subtracted $\nabla_\bseta F_\xi(\bseta^*)$ in the second equality while also setting $A = F(\bseta^*) - \hat{F}_\xi(\bseta^*)$, and used the stationary condition \eqref{eq:staionaryM} in the third equality while setting $B = \nabla_\bseta \hat{F}_\xi(\bseta^*) - \nabla_\bseta F_\xi(\bseta^*)$. 
Similarly, we have for the third term in the right hand side of \eqref{eq:gradLdiff} by adding and subtracting $\hat{F}_\xi(\hat{\bseta}^*)$ and $\nabla_\bseta \hat{F}_\xi(\hat{\bseta}^*)$ and using the stationary condition \eqref{eq:staionaryMhat} to obtain 
\begin{equation}
  \begin{split}
  & (\bsy - F_\xi(\hat{\bseta}^*))^T\Gamma_\text{noise}^{-1}\nabla_\bseta F_\xi(\hat{\bseta}^*)\\
  & = C^T \Gamma_\text{noise}^{-1}\nabla_\bseta F_\xi(\hat{\bseta}^*) + (\bsy - \hat{F}_\xi(\hat{\bseta}^*))^T\Gamma_\text{noise}^{-1} D + P_r\hat{\bseta}^*,\\
  \end{split}
\end{equation}
where $C = \hat{F}_\xi(\hat{\bseta}^*) - F_\xi(\hat{\bseta}^*)$ and $D = \nabla_\bseta F_\xi(\hat{\bseta}^*)  - \nabla_\bseta \hat{F}_\xi(\hat{\bseta}^*)$. Substituting the above two equations into \eqref{eq:gradLdiff}, we obtain
\begin{equation}\label{eq:gradLdiff2}
  \begin{split}
  & \nabla_\bseta \hat{L}(\bseta^*) - \nabla_\bseta L(\hat{\bseta}^*) \\
  & = - A^T \Gamma_\text{noise}^{-1}\nabla_\bseta \hat{F}_\xi(\bseta^*) - (\bsy - F_\xi(\bseta^*))^T\Gamma_\text{noise}^{-1}B -  (\bseta^* - P_r \bseta^*)\\
  & \quad + C^T \Gamma_\text{noise}^{-1}\nabla_\bseta F_\xi(\hat{\bseta}^*) + (\bsy - \hat{F}_\xi(\hat{\bseta}^*))^T\Gamma_\text{noise}^{-1} D + (P_r \hat{\bseta}^* - \hat{\bseta}^*)\\
  & \leq \bar{c}_1 \varepsilon_1 + \bar{c}_2 \varepsilon_2 + \bar{c}_3\varepsilon_3,
  \end{split}
\end{equation}
where $\varepsilon_1$ and $\varepsilon_2$ are the bounds for the errors of the forward map and its Jacobian by the neural network approximation in \eqref{eq:approxFJ}, and $\varepsilon_3$ is the bound for the approximation error of the projection operator. Note that $ P_r \hat{\bseta}^* - \hat{\bseta}^* = 0$ in \eqref{eq:gradLdiff2} as $\hat{\bseta}^*$ is obtained in the projected subspace so that $P_r(\hat{\bseta}^*) = \hat{\bseta}^*$. The result follows by setting $\bar{c}_3 = 1$ and 
\begin{equation}
  \begin{split}
    \bar{c}_1 & = ||\Gamma_\text{noise}^{-1}||(||\nabla_\bseta \hat{F}_\xi(\bseta^*)|| + ||\nabla_\bseta F_\xi(\hat{\bseta}^*)||), \\
    \bar{c}_2 & = ||\Gamma_\text{noise}^{-1}||(2||\bsy|| + ||F_\xi(\bseta^*)|| + ||\hat{F}_\xi(\hat{\bseta}^*)||),
  \end{split}
\end{equation}
which are bounded by the assumption that the forward map and its Jacobian and their neural network approximations are bounded. For the right hand side of \eqref{eq:hessgraddiff}, we have the lower bound \eqref{eq:lowerbound} by assumption.
To this end, a combination of this lower bound and \eqref{eq:gradLdiff2} concludes by setting $c_i = \bar{c}_i/c_h$ with $i = 1, 2, 3$.
\end{proof}

\subsection{Approximation error of the eigenvalues}
\label{sec:ErrorAnalysis-eig}
We consider the approximation of the eigenvalues of the generalized eigenvalue problem \eqref{eq:disc-gen-eig} by those of the eigenvalue problem \eqref{eq:rEig} with neural network approximation of the forward map. 
\begin{theorem}
\label{thm:2}

Let $\lambda_i$, $i = 1, \dots, r$, denote the leading eigenvalues of the generalized eigenvalue problem \eqref{eq:disc-gen-eig}, and $\hat{\lambda}_i$, $i = 1, \dots, r$, denote those of the eigenvalue problem \eqref{eq:rEig}. Under Assumption \ref{ass:assumptions} and the assumption 
\begin{equation}\label{eq:Lipschitz}
    ||\nabla_\bseta F_\xi (\bseta^*) - \nabla_\bseta F_\xi(\hat{\bseta}^*)|| \leq L_F ||\bseta^* - \hat{\bseta}^*||
\end{equation}
for some Lipschitz constant $L_F$, we have 
\begin{equation}\label{eq:boundeigenvalue}
    |\lambda_i - \hat{\lambda}_i| \leq \hat{c}_1 \varepsilon_1 + \hat{c}_2 \varepsilon_2 + \hat{c}_3 \varepsilon_3 =: \varepsilon_\lambda, \quad i = 1, \dots r, 
\end{equation}
for some constants $\hat{c}_1$, $\hat{c}_2$, and $\hat{c}_3$ independent of $\varepsilon_1$, $\varepsilon_2$, and $\varepsilon_3$. 
\end{theorem}
\begin{proof}
    Let $H$ denote the prior-preconditioned Hessian of the misfit evaluated at the MAP point as
    \begin{equation}
        H = (\nabla_\bseta F_\xi(\bseta^*))^T \Gamma_{\text{noise}}^{-1} \nabla_\bseta F_\xi(\bseta^*),
    \end{equation}
then the eigenvalues of $H$, given by the eigenvalue problem   
\begin{equation}\label{eq:eigenH}
    H \bsv_i = \lambda_i \bsv_i, \quad i = 1, \dots, r
\end{equation}
are the same as the eigenvalues of the generalized eigenvalue problem \eqref{eq:disc-gen-eig}, and the eigenvectors satisfy $\bsv_i = \Gamma_{\text{prior}}^{-1/2} \bsw_i$, $i = 1, \dots, r$. Let $\hat{H}$ denote the approximate prior-preconditioned Hessian of the misfit evaluated at the approximate MAP point as
\begin{equation}
    \hat{H} = (\nabla_\bseta \hat{F}_\xi(\hat{\bseta}^*))^T \Gamma_{\text{noise}}^{-1} \nabla_\bseta \hat{F}_\xi(\hat{\bseta}^*),
\end{equation}
then the eigenvalues of $\hat{H}$, given by the eigenvalue problem   
\begin{equation}\label{eq:eigenHhat}
    \hat{H} \hat{\bsv}_i = \hat{\lambda}_i \hat{\bsv}_i, \quad i = 1, \dots, r
\end{equation}
are the same as the eigenvalues of the projected eigenvalue problem \eqref{eq:rEig}, and the eigenvectors $\hat{\bsv}_i = \Gamma_{\text{prior}}^{-1/2} 
 \Psi_\bsm \bsu_i$. This results from approximating the forward map with its approximate Jacobian $\nabla_\bseta \hat{F}_\xi $ given as in \eqref{eq:red-jac}. Let us first bound the Jacobians as 
 \begin{equation}
     ||\nabla_\bseta F_\xi(\bseta^*) - \nabla_\bseta \hat{F}_\xi(\hat{\bseta}^*)|| \leq ||\nabla_\bseta F_\xi(\bseta^*) - \nabla_\bseta F_\xi(\hat{\bseta}^*)|| + ||\nabla_\bseta F_\xi(\hat{\bseta}^*) - \nabla_\bseta \hat{F}_\xi(\hat{\bseta}^*)||
 \end{equation}
by a triangle inequality, which leads to the bound 
\begin{equation}\label{eq:JacBound}
     ||\nabla_\bseta F_\xi(\bseta^*) - \nabla_\bseta \hat{F}_\xi(\hat{\bseta}^*)|| \leq L_F ||\bseta^* - \hat{\bseta}^*|| + \varepsilon_2 \leq c_1 L_F \varepsilon_1 + (c_2 L_F + 1) \varepsilon_2 + c_3 L_F \varepsilon_3,
 \end{equation}
for which we used the Lipschitz continuity assumption in \eqref{eq:Lipschitz} and assumption \eqref{eq:approxFJ} in the first inequality and the bound \eqref{eq:MAPbound} for the MAP point approximation in the second inequality. Moreover, we have 
\begin{equation}\label{eq:HessBound}
    ||H - \hat{H}|| \leq c_H ||\nabla_\bseta F_\xi(\bseta^*) - \nabla_\bseta \hat{F}_\xi(\hat{\bseta}^*)||,
\end{equation}
with $c_H = ||\Gamma_{\text{prior}}|| \, ||\Gamma^{-1}_{\text{noise}}|| (||\nabla_\bseta F_\xi(\bseta^*)|| + ||\nabla_\bseta \hat{F}_\xi(\hat{\bseta}^*)||) $, 
for which we added and subtracted a term $ (\nabla_\bseta \hat{F}_\xi(\hat{\bseta}^*))^T \Gamma_{\text{noise}}^{-1} \nabla_\bseta {F}_\xi({\bseta}^*)$ in $H - \hat{H}$. By Weyl's theorem \cite[Ch.\ 5.7]{bai2000templates} for the perturbation of Hermitian matrices, which both $H$ and $\hat{H}$ belong to, we have 
\begin{equation}
    |\lambda_i - \hat{\lambda}_i| \leq ||H - \hat{H}||_{\text{op}}, \quad i = 1, \dots, r,
\end{equation}
which concludes by using the bounds \eqref{eq:JacBound} and \eqref{eq:HessBound}, and the fact that the operator norm is smaller than the Frobenius norm, i.e., 
$||H - \hat{H}||_{\text{op}} \leq ||H - \hat{H}||$,
with constants 
\begin{equation}\label{eq:hatc}
    \hat{c}_1 = c_1 L_F c_H, \quad \hat{c}_2 = (c_2 L_F + 1) c_H, \quad \hat{c}_3 = c_3 L_F c_H.
\end{equation}
\end{proof}

\subsection{Approximation error of the optimality criteria}\label{sec:error-Opt}

\begin{theorem}

    Under Assumption \ref{ass:assumptions} and the assumptions in Theorems \ref{thm:2}, we have the approximation errors for the A-optimality in \eqref{eq:a-final-approx} as 
    \begin{equation}
        \left| \sum_{i=1}^r \frac{\lambda_i}{1+\lambda_i} - \sum_{i=1}^r \frac{\hat{\lambda}_i}{1+\hat{\lambda}_i}\right| \leq r \varepsilon_\lambda,
    \end{equation}
with the rank $r$ and the bound $\varepsilon_\lambda$ in \eqref{eq:boundeigenvalue}, for the D-optimality in \eqref{eq:d-final-approx} as 
\begin{equation}
\left|\sum_{i=1}^r\log(1+\lambda_i) - \sum_{i=1}^r\log(1+\hat{\lambda}_i)\right| \leq r \varepsilon_\lambda,
\end{equation}
and for the EIG-optimality in \eqref{eq:eig-final-approx} as 
\begin{equation}
    \left|\mathcal{E} - \hat{\mathcal{E}} \right| \leq 2r \varepsilon_\lambda + c_\bsm \varepsilon_\bsm, 
\end{equation}
with $c_\bsm$ defined in \eqref{eq:cMAP}, and the information gain 
\begin{equation}\label{eq:cE}
    \mathcal{E} = \sum_{i=1}^r\log (1 + \lambda_i) - \sum_{i=1}^r \frac{\lambda_i}{1+\lambda_i}  + ||\bsm^{\bsy, \xi}_{\text{MAP}}-\bsm_{\text{prior}}||^2_{\Gamma_{\text{prior}}^{-1}}.
\end{equation}
$\hat{\mathcal{E}}$ is the approximation of  \eqref{eq:cE} with the approximate eigenvalues and MAP point.
\end{theorem}

\begin{proof}

For the A-optimality, we have 
\begin{equation}
    \left| \sum_{i=1}^r \frac{\lambda_i}{1+\lambda_i} - \sum_{i=1}^r \frac{\hat{\lambda}_i}{1+\hat{\lambda}_i}\right| = \left| \sum_{i=1}^r \frac{\lambda_i - \hat{\lambda}_i}{(1+\lambda_i)(1+\hat{\lambda}_i)} \right| \leq \sum_{i=1}^r  \frac{|\lambda_i - \hat{\lambda}_i|}{(1+\lambda_i)(1+\hat{\lambda}_i)}  \leq r \varepsilon_r,
\end{equation}
where we used $\lambda_i > 0, \hat{\lambda}_i > 0$, $i = 1, \dots, r$, and the bound \eqref{eq:boundeigenvalue}.

For the D-optimality, we have 
\begin{equation}
    \left|\sum_{i=1}^r\log(1+\lambda_i) - \sum_{i=1}^r\log(1+\hat{\lambda}_i)\right| = \left| \sum_{i=1}^r\log\left(1+ \frac{\lambda_i-\hat{\lambda}_i}{1+\hat{\lambda}_i}\right) \right| \leq \sum_{i=1}^r \frac{|\lambda_i - \hat{\lambda}_i|}{1+\hat{\lambda}_i} \leq r \varepsilon_\lambda,
\end{equation}
where we used $\log(1+x) \leq |x|$ for $x > -1$, $|\lambda_i - \hat{\lambda}_i| < (1+\hat{\lambda}_i)$, and $\hat{\lambda}_i > 0$, $i = 1, \dots, r$, and the bound \eqref{eq:boundeigenvalue}.

For the EIG-optimality, we have for the last term 
\begin{equation}
    \left|||\bsm^* - \bsm_{\text{prior}}||_{\Gamma_{\text{prior}}^{-1}}^2  - ||\hat{\bsm}^* - \bsm_{\text{prior}}||_{\Gamma_{\text{prior}}^{-1}}^2\right| \leq c_\bsm ||\bsm^* - \hat{\bsm}^*||_{\Gamma^{-1}_{\text{prior}}} \leq c_\bsm \varepsilon_\bsm
\end{equation}
where we used the bound \eqref{eq:MAPbound} with 
\begin{equation}\label{eq:cMAP}
c_\bsm = ||\bsm^* - \bsm_{\text{prior}}||_{\Gamma_{\text{prior}}^{-1}}  + ||\hat{\bsm}^* - \bsm_{\text{prior}}||_{\Gamma_{\text{prior}}^{-1}},    
\end{equation}
which concludes using triangular inequality and the bound for the A-optimality and D-optimality. 

\end{proof}

\section{Numerical result}
\label{sec:numerical}
This section demonstrates our proposed computational framework's accuracy, scalability, and efficiency for two Bayesian OED problems. To evaluate the accuracy of neural network surrogates, we consider the approximation of the PtO map, its Jacobian, the MAP point, the eigenvalues, and three different optimality criteria for Bayesian OED. We also verify that the assumptions outlined in Assumption~\ref{ass:assumptions} align with our numerical results, validating the theoretical foundation of our approach. For scalability, we illustrate how accuracy is preserved with increasing parameter dimensions while quantifying the computational savings in terms of required PDE solves. To demonstrate efficiency, we compare the computational cost of the neural network surrogates with that of the high-fidelity approximation. This comparison includes both the cost of offline training and the surrogate's online evaluation in selecting optimal sensors, providing a comprehensive assessment of the framework's computational advantages.

\subsection{Test problems}
We consider two test problems, one of a linear diffusion problem and the other of a nonlinear convection-diffusion-reaction problem with a nonlinear reaction term, both with infinite-dimensional parameter fields and leading to nonlinear Bayesian inverse problems. We present the two examples in this subsection and the numerical results side by side in the following subsections for easy comparison.

\subsubsection{Linear diffusion problem}
For the first test problem, we consider Bayesian OED for optimal sensor placement to infer a permeability field in pressure-driven Darcy flow in a physical domain $D = (0,1)^2$, which is governed by the following diffusion equation prescribed with suitable boundary conditions
\begin{equation}\label{eq:diffusion} 
\begin{split}
 -\nabla \cdot ( e^m \nabla u) &= 0 \text{ in } D, \\
 u &= 1 \text{ on } \partial D_\text{top}, \\
 u &= 0 \text{ on } \partial D_\text{bottom}, \\
 e^m\nabla u \cdot \mathbf{n} &= 0 \text{ on } \partial D_\text{sides},    
\end{split}   
\end{equation}
where the state variable $u$ represents the pressure field of the Darcy flow driven by the pressure difference with Dirichlet boundary conditions $u = 1$ on the top boundary $\partial D_\text{top}$ and $u = 0$ on the bottom boundary $\partial D_\text{bottom}$. A homogeneous Neumann boundary condition is assumed on the two sides $\partial D_\text{sides}$, where $\mathbf{n}$ denotes the unit-length outward normal for the side boundaries. $e^m$ represents a positive permeability field with log-normal distribution, i.e.,
the model parameter $m \sim \mathcal{N}(0, \mathcal{C})$ is assumed be to be Gaussian random field with Mat\'ern covariance $\mathcal{C} = \cA^{-2}=(-\gamma\Delta + \kappa I)^{-2}$, where we set $\gamma = 0.1$, and $\kappa=0.5$. We use a finite element method (FEM) with piecewise linear polynomials to approximate the parameter and pressure fields discretized using a uniform mesh of size $64\times 64$. The top of Figure \ref{fig:poisson_setup} shows a random sample of the parameter, the corresponding state as a solution of the diffusion equation, and its pointwise observation with Gaussian additive noise following $\cN(0,0.1)$ at the $50$ candidate sensor locations. We consider 50 candidate sensors in the lower part of the physical domain and select 5 sensors from the 50 candidates.

\begin{figure}[!htb]
  \centering      
  \begin{subfigure}{0.32\textwidth}
        \includegraphics[width=\linewidth]{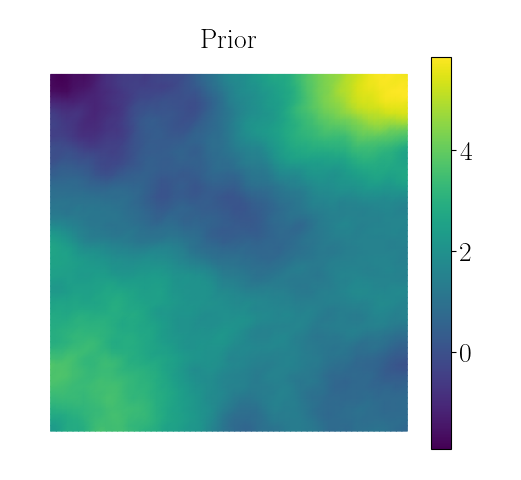}
    \end{subfigure}
    \hfill
    \begin{subfigure}{0.32\textwidth}
        \includegraphics[width=\linewidth]{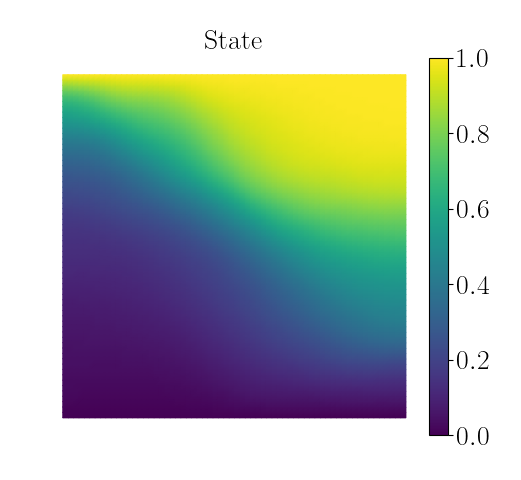}
    \end{subfigure}
    \hfill
    \begin{subfigure}{0.32\textwidth}
        \includegraphics[width=\linewidth]{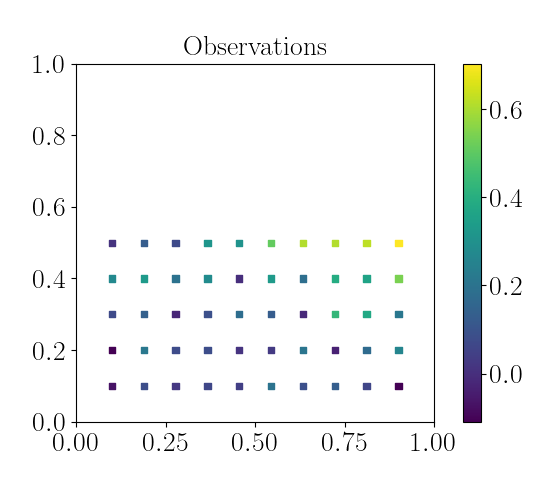}
    \end{subfigure}

    \vskip -0.3cm

    \begin{subfigure}{0.32\textwidth}
      \includegraphics[width=\linewidth]{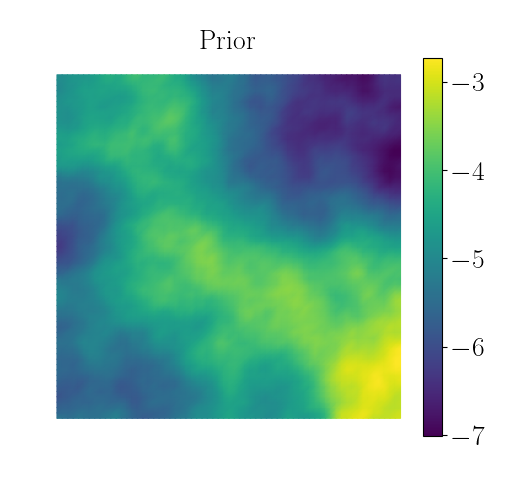}
    \end{subfigure}
    \hfill
    \begin{subfigure}{0.32\textwidth}
        \includegraphics[width=\linewidth]{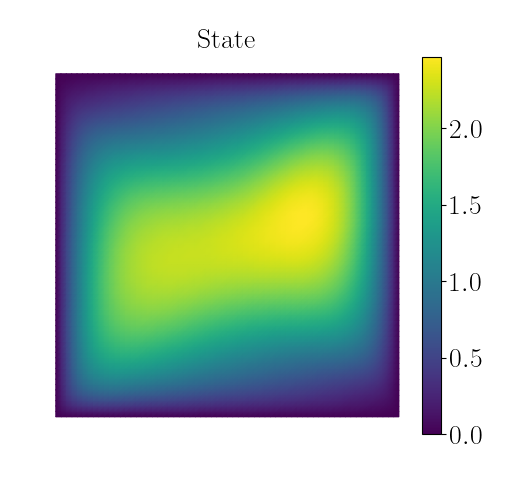}
    \end{subfigure}
    \hfill
    \begin{subfigure}{0.32\textwidth}
        \includegraphics[width=\linewidth]{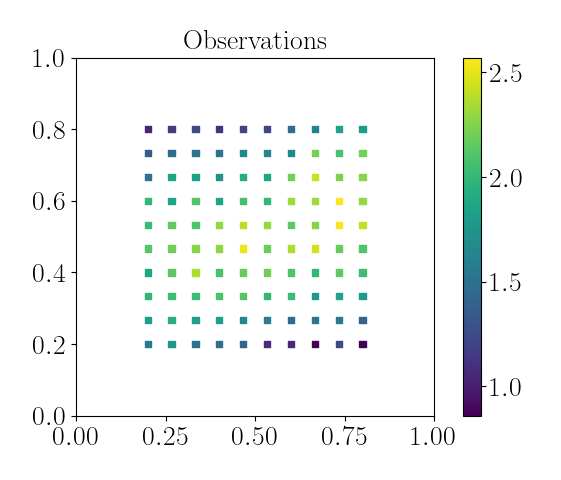}
    \end{subfigure}  
  \caption{A random prior sample $m \sim \cN(0, \cC)$ (left), the high-fidelity solution $u(m)$ by finite element approximation (middle), and the observation data $\bsy $ at all the candidate observation points (right). Top: diffusion problem. Bottom:  CDR problem.}\label{fig:poisson_setup}
\end{figure}

\subsubsection{Nonlinear convection-diffusion-reaction problem}

For the second test problem, we consider Bayesian OED to infer a reaction coefficient field in mass transfer problem in the domain  $D = (0,1)^d$, $d = 2, 3$, which is governed by the nonlinear (semilinear) convection-diffusion-reaction (CDR) equation
\begin{equation}\label{eq:cdr_equation}
\begin{split}
-\nu \Delta u + \bsv \cdot \nabla u + e^m u^3 = f \text{ in } D&, \\
u = 0 \text{ on } \partial D,&  
\end{split}
\end{equation}
where a homogeneous Dirichlet boundary condition is prescribed on the whole boundary $\partial D$. The source term is given as a Gaussian bump $f(\bsx) = \max(0.5, \exp(-25||\bsx-\bsx_s||_2^2))$ at $\bsx_s = (0.7, 0.7)$ (2D) and $\bsx_s = (0.7, 0.7, 0.7)$ (3D). 
The velocity field $\bsv$ in the convection term is obtained as the solution of the following steady-state Navier--Stokes equations 
\begin{equation}
    \begin{split}
- \nu \Delta \bsv  + \bsv \cdot \nabla \bsv + \nabla p = 0 \text{ in } D&,\\
\nabla\cdot \bsv = 0 \text{ in } D&, \\
\bsv = \bsg \text{ on } \partial D&,
    \end{split}
\end{equation}
where $\bsg = (0, 1)$ (2D) and $\bsg = (0, 1, 0)$ (3D) on left wall, $\bsg=(0, -1)$ (2D) and $\bsg=(0, -1, 0)$ (3D) on the right wall, and zero elsewhere. We set the viscosity $\nu = 0.01$. We use a finite element method to solve the equations with piecewise quadratic polynomials for the approximation of the velocity field $\bsv$ and piecewise linear polynomials for the approximation of the pressure field $p$, the parameter field $m$, and the concentration field $u$. The finite element solution $u$ at a random parameter sample $m$ with the same Gaussian prior distribution as in the last example, and its corresponding observation data $\bsy$ at 100 candidate observation locations are shown at the bottom plots of Figure \ref{fig:poisson_setup}. We seek to choose 10 sensors out of the 100 candidates for the Bayesian OED problem.



\subsection{Dimension reduction}
We reduce the dimension for the high-dimensional input parameters in both application problems and the output observables in the second problem since the output dimension $d_s = 50$ is relatively small in the first problem. We first use a uniform mesh of size $64\times 64$, which leads to a discretization of the infinite-dimensional input parameter $m$ to a $4,225$-dimensional parameter vector $\bsm$ by finite element approximation as in \eqref{eq:mFEM}. We compute the eigenpairs for the KLE projection \eqref{eq:kle} and the DIS projection \eqref{eq:ias} by solving the eigenvalue problem for $\Gamma_{\text{prior}}$ and the generalized eigenvalue problem \eqref{eq:gIAS}, respectively. We use $1,024$ random samples to approximate the expectation in \eqref{eq:gIAS}. The eigenvalues and eigenvectors of KLE and DIS are shown in Figure \ref{fig:Poisson-eigenvalues} and \ref{fig:poisson-combined}. From Figure \ref{fig:Poisson-eigenvalues}, we observe that the DIS eigenvalues are smaller and also decay faster than the KLE eigenvalues for both problems. From the second and third rows of Figure \ref{fig:poisson-combined}, we see that the DIS basis functions are different for the two problems, and the first DIS basis functions expose the features of the solutions of the two PDE problems as shown in Figure \ref{fig:poisson_setup}, while the KLE basis functions are the same for the two problems as shown in the first row.

\begin{figure}[htbp]
    \centering
    \includegraphics[width=0.4\linewidth]{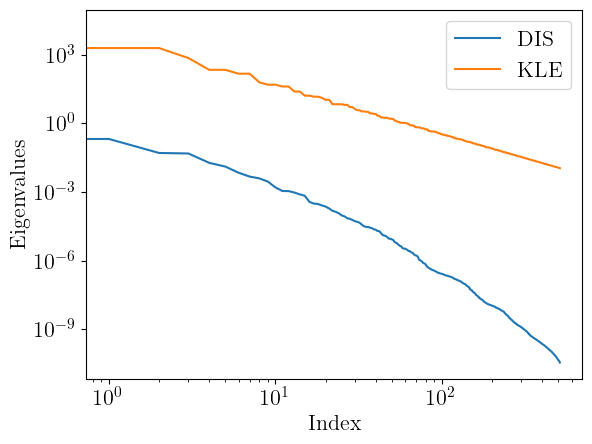}
    \includegraphics[width=0.4\linewidth]{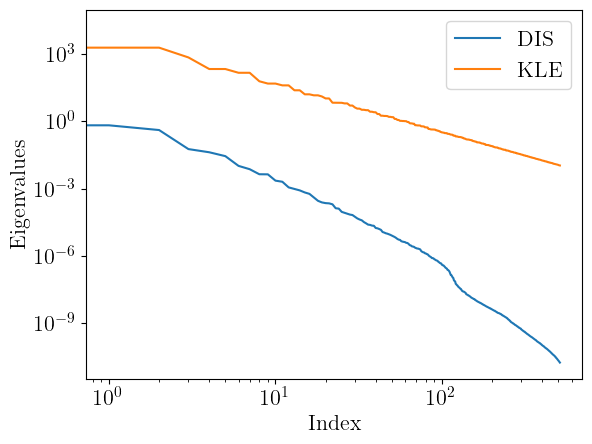}
    \caption{Decay of the eigenvalues of the prior covariance $\Gamma_\text{prior}$ used in the KLE approximation of the parameter \eqref{eq:kle} and the eigenvalues of the generalized eigenvalue problem \eqref{eq:gIAS} used in DIS approximation of the parameter \eqref{eq:ias}. Left: for the diffusion problem. Right: for the  CDR problem.}
    \label{fig:Poisson-eigenvalues}
\end{figure}

\begin{figure}[!htbp]
    \centering
    \begin{subfigure}[b]{\textwidth}
        \centering
        \includegraphics[width=0.9\textwidth]{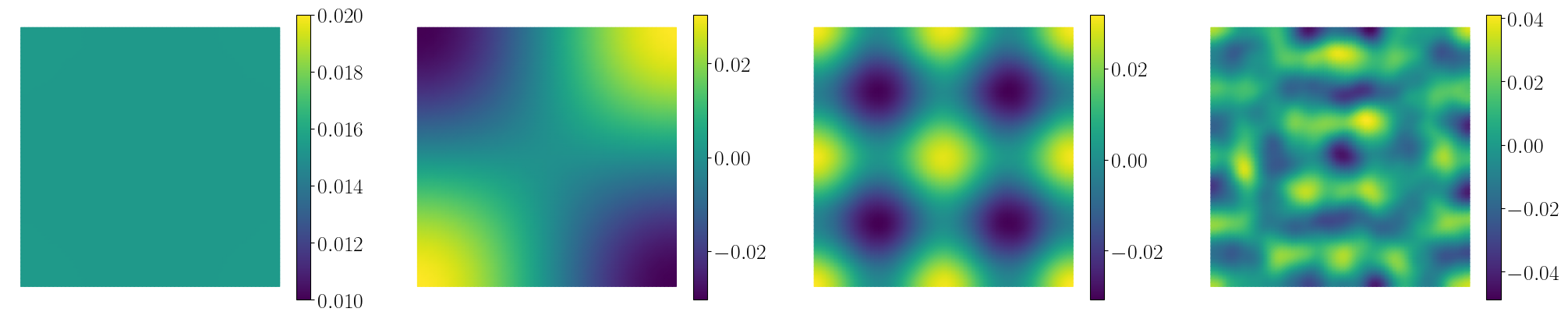}
    \end{subfigure}

    
    \begin{subfigure}[b]{\textwidth}
        \centering
        \includegraphics[width=0.9\textwidth]{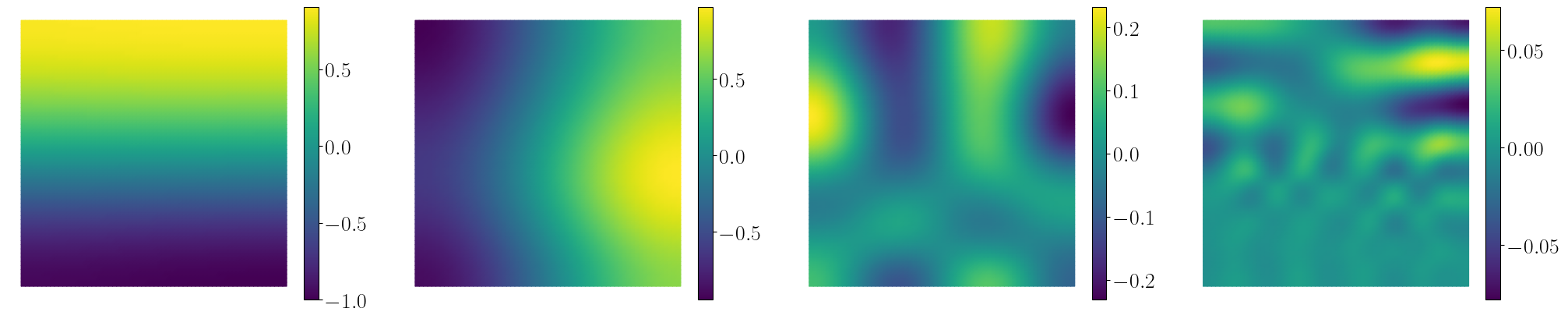}
    \end{subfigure}
    

  
  \begin{subfigure}[b]{\textwidth}
      \centering
      \includegraphics[width=0.9\textwidth]{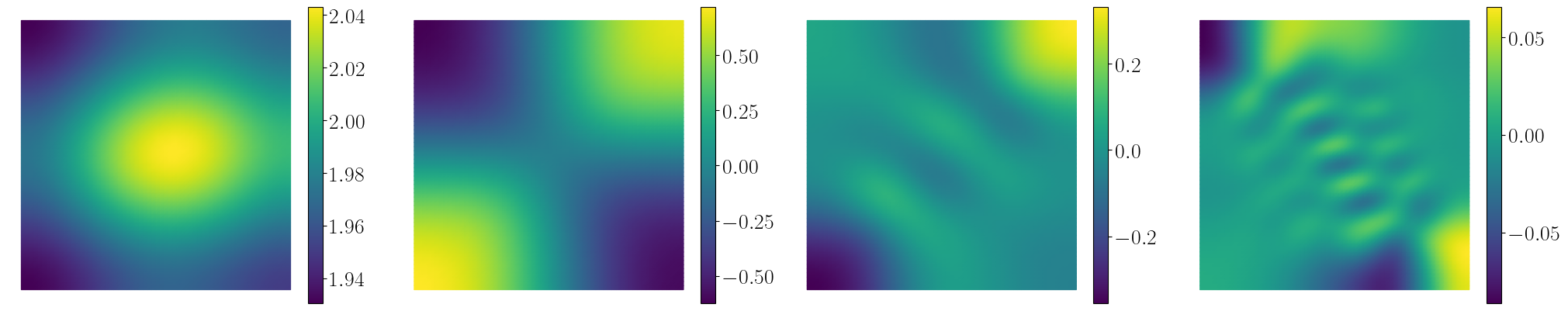}
  \end{subfigure}
    \caption{KLE bases 1, 4, 16, 64 (top). DIS bases 1, 4, 16, 64 for the diffusion (middle) and CDR (bottom) problems.}
    \label{fig:poisson-combined}
\end{figure}


\begin{figure}[!ht]
  \centering      
    
    \begin{subfigure}{0.32\textwidth}
      \includegraphics[width=\linewidth]{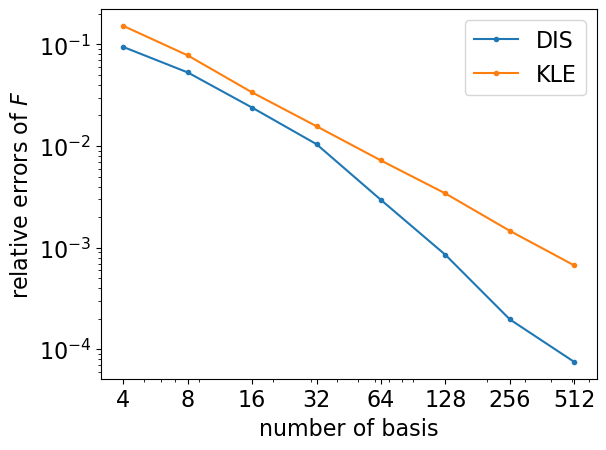}
    \end{subfigure}
    \hfill
    \begin{subfigure}{0.32\textwidth}
        \includegraphics[width=\linewidth]{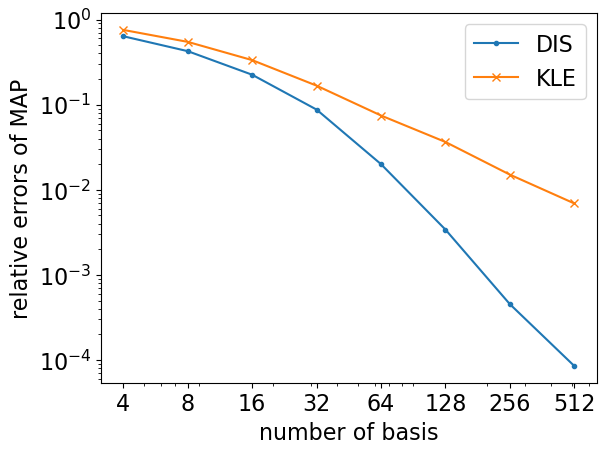}
    \end{subfigure}
    \hfill
     \begin{subfigure}{0.32\textwidth}
        \includegraphics[width=\linewidth]{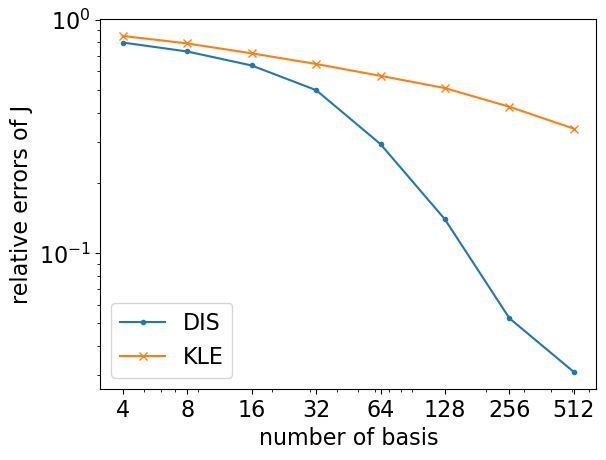}
    \end{subfigure}

    \begin{subfigure}{0.32\textwidth}
      \includegraphics[width=\linewidth]{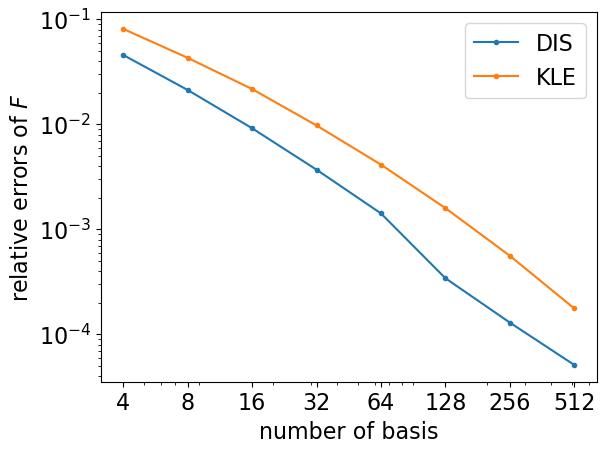}
    \end{subfigure}
    \hfill
    \begin{subfigure}{0.32\textwidth}
      \includegraphics[width=\linewidth]{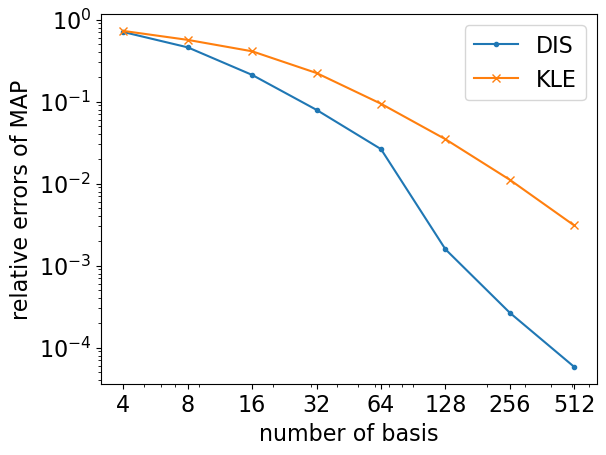}
    \end{subfigure}
    \hfill 
    \begin{subfigure}{0.32\textwidth}
      \includegraphics[width=\linewidth]{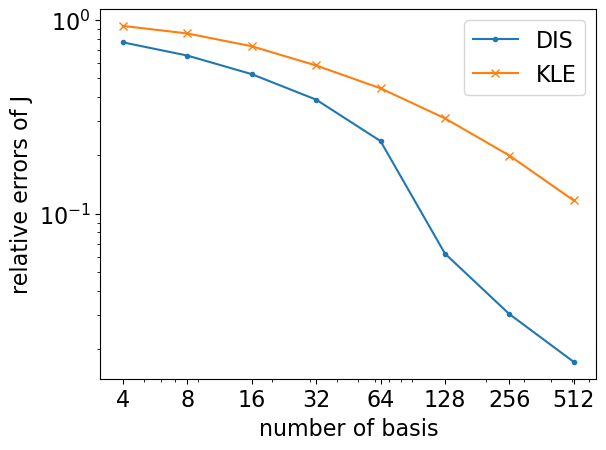}
    \end{subfigure}

  \begin{subfigure}{0.32\textwidth}
      \includegraphics[width=\linewidth]{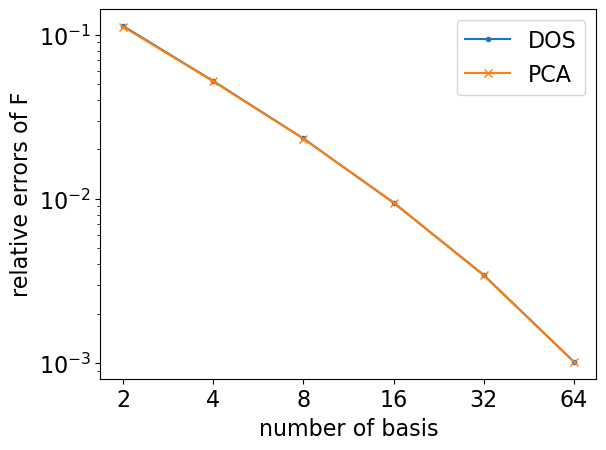}
  \end{subfigure}
  \hfill
  \begin{subfigure}{0.32\textwidth}
      \includegraphics[width=\linewidth]{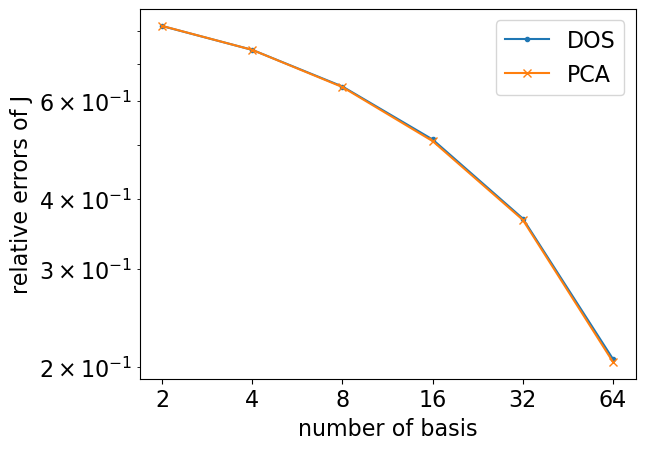}
  \end{subfigure}

  \caption{Relative errors of the input projection (top and middle rows for the two 2D examples, respectively) and output projection (bottom row for the  CDR problem) for the observables (left), the MAP point (middle), and the Jacobian (right). We use DIS \eqref{eq:ias} and KLE \eqref{eq:kle} basis for input projection, and DOS \eqref{eq:oas} and PCA \eqref{eq:pca} basis for output projection. 
  }\label{fig:poisson_recon}
\end{figure}

The relative errors of different input and output projection methods are shown in Figure \ref{fig:poisson_recon}, 
Specifically, we compute the mean of the relative input projection errors for the observables and the Jacobian as 
\begin{equation}\label{eq:errorFJ}
  \bE_{\bsm} \left[ \frac{||F(\bsm) - F(P_r(\bsm))||_2}{||F(\bsm)||_2}\right] \text{ and } \bE_{\bsm} \left[ \frac{||J(\bsm) - J(P_r(\bsm))||_F}{||J(\bsm)||_F}\right]
\end{equation}
where $P_r$ denotes a linear projector of the parameter vector $\bsm$ by either KLE in \eqref{eq:kle} or DIS in \eqref{eq:ias}, $||\cdot||_2$ and $||\cdot||_F$ denote the Euclidean norm for the observable vector and the Frobenius norm for the Jacobian matrix, respectively. More explicitly, we evaluate the projected Jacobian matrix as 
\begin{equation}
  J(P_r(\bsm)) = J (\bsm) \Psi^{\text{KLE}}_\bsm (\Psi^{\text{KLE}}_\bsm)^T \text{ or } J(P_r(\bsm)) = J (\bsm) \Psi^{\text{DIS}}_\bsm (\Psi^{\text{DIS}}_\bsm)^T \Gamma_{\text{prior}}^{-1}
\end{equation}
for the KLE basis $\Psi^{\text{KLE}}_\bsm$ or the DIS basis $\Psi^{\text{DIS}}_\bsm$. The expectation in \eqref{eq:errorFJ} is evaluated by SAA using 100 
samples. We also compute the relative input projection error for the MAP point in \eqref{eq:MAP} as
\begin{equation}
  \bE_{\bsy} \left[ \frac{||\bsm_{\text{MAP}}^{\bsy, \xi} - P_r(\bsm_{\text{MAP}}^{\bsy, \xi})||_{\Gamma_{\text{prior}}^{-1}}}{||\bsm_{\text{MAP}}^{\bsy, \xi}||_{\Gamma_{\text{prior}}^{-1}}}\right]\label{eq:MAP-re}
\end{equation}
which is consistent with A3 in Assumption \ref{ass:assumptions}.
The expectation is evaluated by SAA using 100 samples of $\bsy = F(\bsm) + \bsepsilon$ by drawing random samples $\bsm$ and $\bsepsilon$. From the first and second rows of Figure \ref{fig:poisson_recon}, we can observe that using the DIS basis yields smaller projection errors for all three quantities. Moreover, the DIS basis is particularly advantageous for the approximation of Jacobian, which plays a key role in evaluating the optimality criteria. On the other hand, as shown in the third row of Figure \ref{fig:poisson_recon}, the output projection used in the CDR problem is less sensitive to the projection basis of PCA in \eqref{eq:pca} and DOS in \eqref{eq:oas}, as measured by the average (using  100 
samples) of the relative projection error for the observables and the Jacobian as 
\begin{equation}\label{eq:errorFJout}
  \bE_{\bsm} \left[ \frac{||F(\bsm) - F_r(\bsm)||_2}{||F(\bsm)||_2}\right] \text{ and } \bE_{\bsm} \left[ \frac{||J(\bsm) - \Psi_F \Psi_F^T J(\bsm)||_F}{||J(\bsm)||_F}\right],
\end{equation}
where the projected observables $F_r$ are defined in \eqref{eq:pca} or \eqref{eq:oas} with corresponding PCA basis $\Psi_F = \Psi^{\text{PCA}}_F$ or DOS basis $\Psi_F = \Psi^{\text{DOS}}_F$. For the following computations, we use 128
DIS basis functions for input projection in both problems and 30 
PCA basis functions for output projection in the second problem, leading to less than $1\%$ error in the input and output dimension reduction for the PtO map.

Note that 128 DIS basis functions yield smaller than $1\%$ projection error for the MAP point. This low projection error aligns with A3 in Assumption~\ref{ass:assumptions}, which posits that the projection error for the MAP point is negligible. This result further validates our choice of the DIS basis for dimension reduction in the parameter space, especially in approximating the Jacobian and the MAP point.

\subsection{Accuracy and scalability of DINO approximations}

We construct the DINO surrogates as presented in Section \ref{sec:DINN} using PyTorch \cite{paszke2017automatic}. Specifically, we construct the encoder by the DIS projection with the reduced dimension $r_\bsm = 128$ as in \eqref{eq:ias}. For the CDR problem, we also construct the decoder by the PCA projection with the reduced dimension $r_F = 30$ as in \eqref{eq:pca}. We connect the encoder layer with one linear layer, three ResNet layers, and another linear layer before the decoder layer. The two linear layers map the input and output dimensions to the same dimension as in the ResNet layers, for which we choose 100. We use sigmoid as the activation function inside the ResNet layers and tanh for all other layers. This specific architecture is among the best (with the smallest generalization errors) of several ones we tested with different numbers of layers, dimensions, and activation functions.
We use an Adam optimizer \cite{kingma2014adam} with a learning rate of $0.001$. For the diffusion model, we train the neural network with Jacobian for 100 epochs and without Jacobian for 200 epochs. For the CDR model, we train the neural network with Jacobian for 200 epochs and without Jacobian for 300 epochs. We select these numbers of epochs based on when the validation error stops decreasing. 
We use the libraries FEniCS \cite{alnaes2015fenics}, hIPPYlib \cite{villa2021hippylib}, and hIPPYflow \cite{o2021hippyflow} to generate the high-fidelity data, using a mesh of size 64$\times$64 unless otherwise stated in scalability test, and project the data to get $(\bsbeta_\bsm, \bsbeta_F, J_r)$ used in the loss function \eqref{eq:loss} for training and testing.


\subsubsection{The accuracy of the neural network approximations}

\begin{figure}[!htb]
  \centering
  \includegraphics[width=0.4\textwidth]{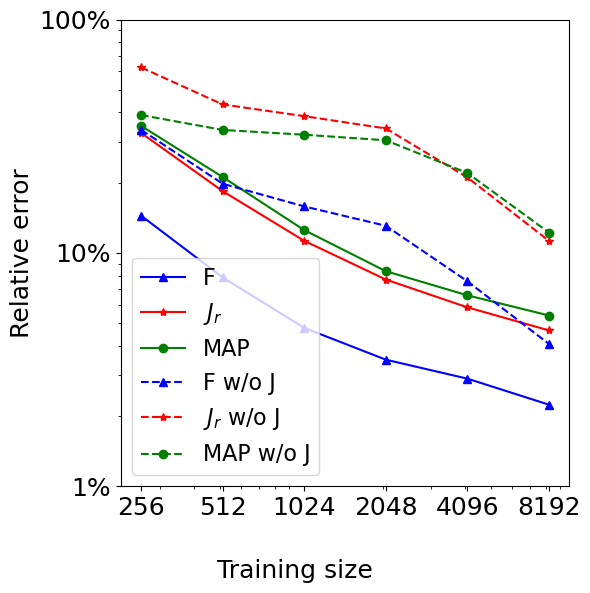}
  \includegraphics[width=0.4\linewidth]{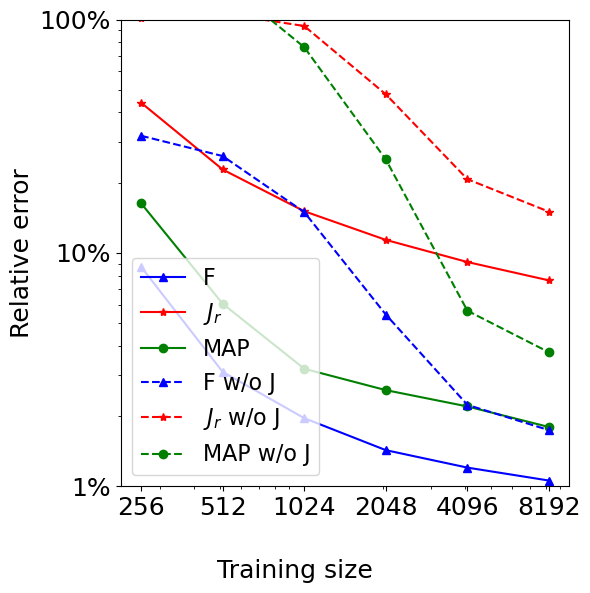}
  \caption{Mean of relative errors with increasing training size for the neural network approximations of the PtO map $F$, the reduced Jacobian ${J_r}$, and the MAP point $\bsm_\text{MAP}$. The solid and dashed line results are for the neural networks trained with and without the Jacobian. Diffusion (left) and CDR (right) problems.} \label{fig:perf-dino-nn}
\end{figure}

We train the neural networks {\color{black}ten} times, starting from different initializations with different training sizes. Then, we compute for each case the neural network approximations of the PtO map $F$, the reduced Jacobian $J_r$, and the MAP points by solving the optimization problem \eqref{eq:rMAP}. The relative errors of the PtO map and the reduced Jacobian are defined as 
\begin{equation}\label{eq:errorFJNN}
  \bE_{\bsm} \left[ \frac{||F(\bsm) - F_{\text{NN}}(\bsm)||_2}{||F(\bsm)||_2}\right] \text{ and } \bE_{\bsm} \left[ \frac{||J_r(\bsm) - \nabla_\bsbeta \Phi_\bstheta(\bsbeta_\bsm)||_F}{||J_r(\bsm)||_F}\right],
\end{equation}
averaged over 1024 test samples, and the relative error of the MAP point in $L_2$-norm is defined as
\begin{equation}
  \bE_{\bsy} \left[ \frac{||\bsm_{\text{MAP}}^{\bsy, \xi} - \bsm_r(\bsbeta_{\text{MAP}}^{\bsy, \xi})||_M}{||\bsm_{\text{MAP}}^{\bsy, \xi}||_M}\right]\label{eq:MAP-NN-L2},
\end{equation}
averaged over 200 samples, with $\bsm_r = \bsm_{\text{prior}} + \Psi^{\text{DIS}}_\bsm \bsbeta_{\text{MAP}}^{\bsy, \xi}$, $||\bsm||_M = \sqrt{\bsm^T M \bsm}$ with the mass matrix $M$ defined in \eqref{eq:mass}. We also consider the error in $\Gamma_{\text{prior}}^{-1}$-norm given by 
\begin{equation}
  \bE_{\bsy} \left[ \frac{||\bsm_{\text{MAP}}^{\bsy, \xi} - \bsm_r(\bsbeta_{\text{MAP}}^{\bsy, \xi})||_{\Gamma_\text{prior}^{-1}}}{||\bsm_{\text{MAP}}^{\bsy, \xi}||_{\Gamma_\text{prior}^{-1}}}\right]\label{eq:MAP-NN},
\end{equation}
as reported in Table \ref{tab:metrics_models}.
We use PyTorch's LBFGS optimizer with 100 maximum iterations to compute the MAP points using the neural network surrogate (DINO) trained with the Jacobian data. For the neural network surrogate trained without Jacobian data, we have to use an additional 300 iterations of Pytorch's Adam optimizer with a 0.01 learning rate before LBFGS to obtain a convergent computation of the MAP points. We report the mean of the relative errors averaged over the {\color{black}ten} trials of the neural network training in Figure \ref{fig:perf-dino-nn}. Increasing training size makes the neural network approximations more accurate for all three quantities. Moreover, the neural networks trained with the Jacobian information lead to much more accurate approximations than those trained without Jacobian, especially when the training size is relatively small. See Table \ref{tab:metrics_models} for the relative approximation errors (mean and standard deviation) at the training size of a relatively small size of 1,024 and a larger size of 8,192. 

As we increase the number of samples, the errors of the proposed surrogate for the PtO map $F$ and $J_r$ decrease on the test samples randomly drawn from the prior distribution, ensuring that A2 in Assumption \ref{ass:assumptions} holds. Under this assumption, the MAP point estimation error are bounded by the surrogate errors up to some constants as demonstrated in Theorem \ref{thm:1}, as demonstrated by the results in Table \ref{tab:metrics_models}.

\begin{table}[!htb]
  \footnotesize
  \begin{center}
  \begin{tabular}{|c|c|c|c|c|}
  \hline
  \multicolumn{5}{|c|}{Linear diffusion problem} \\ \hline
  Training size & \multicolumn{2}{c|}{1,024} & \multicolumn{2}{c|}{8,192} \\ \hline
  Loss function & with J & w/o J & with J & w/o J \\ \hline
  $F$ & 4.77\% (0.4\%) & 17.2\% (0.42\%) & 2.23\% (0.34\%) & 5.82\% (0.26\%) \\ \hline
  $J_r$ & 11.26\% (0.23\%) & 40.32\% (0.42\%) & 4.65\% (0.1\%) & 15.4\% (0.28\%) \\ \hline
  MAP ($L_2$) & 12.52\% (0.65\%) & 34.48\% (1.32\%) & 5.39\% (1.39\%) & 17.23\% (0.7\%) \\ \hline
  MAP ($\Gamma_{\text{prior}}^{-1}$) & 26.27\% (0.6\%) & 63.13\% (0.2\%) & 12.06\% (2.02\%) & 26.52\% (1.58\%) \\ \hline
  \multicolumn{5}{|c|}{Nonlinear convection-diffusion-reaction problem} \\ \hline
  Training size & \multicolumn{2}{c|}{1,024} & \multicolumn{2}{c|}{8,192} \\ \hline
  Loss function & with J & w/o J & with J & w/o J \\ \hline
  $F$ & 1.9\% (0.1\%) & 15.94\% (0.93\%) & 1.02\% (0.1\%) & 1.90\% (0.06\%) \\ \hline
  $J_r$ & 15.13\% (0.21\%) & 86.24\% (0.1\%) & 7.7\% (0.21\%) & 17.69\% (0.21\%) \\ \hline
  MAP ($L_2$) & 3.22\% (0.14\%) & 132.1\% (84.73\%) & 1.89\% (0.36\%) & 4.38\% (0.22\%) \\ \hline
  MAP ($\Gamma_{\text{prior}}^{-1}$) & 10.1\% (0.12\%) & 124.7\% (32.97\%) & 9.80\% (0.14\%) & 11.78\% (0.12\%) \\ \hline
  \end{tabular}
  \caption{Mean (and standard deviation) of the relative approximation errors of the neural network surrogates trained with 1,024 and 8,192 training samples for the loss function with Jacobian (DINO) and without (w/o) Jacobian for the two problems.}
  \label{tab:metrics_models}
  \end{center}
\end{table}

\begin{figure}[!htb]
 \centering
  \begin{subfigure}[b]{\textwidth}
\includegraphics[width=\textwidth]{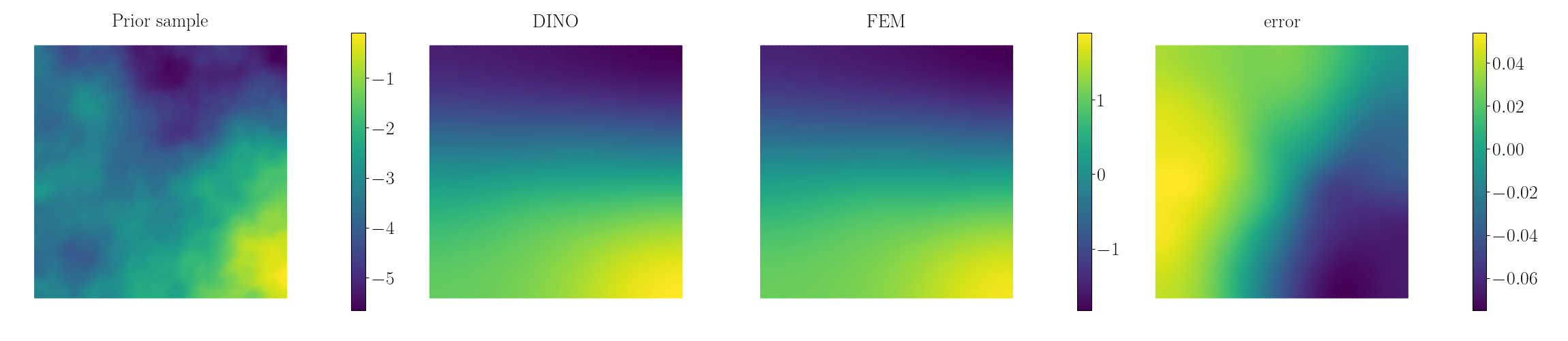}
  \end{subfigure} \\
    \begin{subfigure}[b]{\textwidth}
    \includegraphics[width=\textwidth]{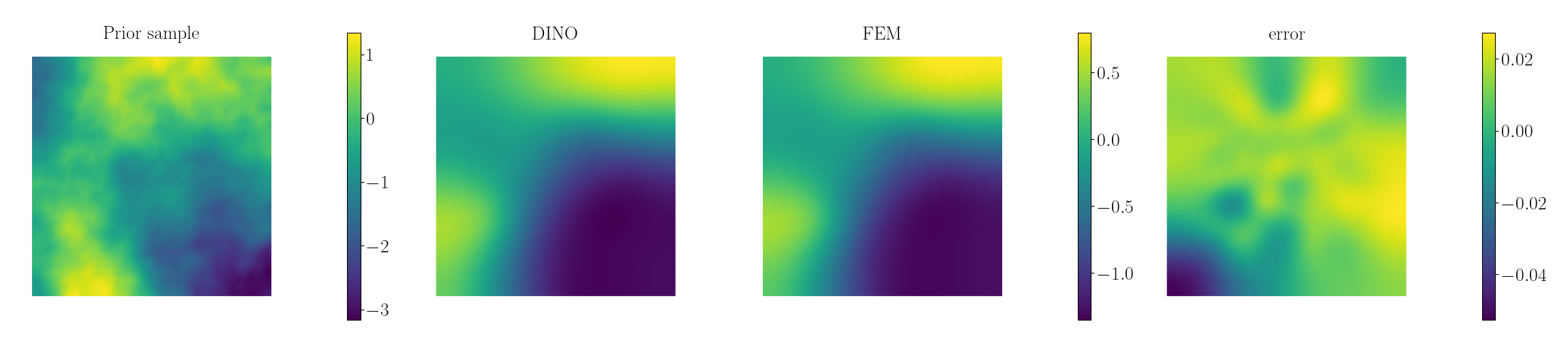}
  \end{subfigure} \\
\hspace{0.2cm}
  \begin{subfigure}[b]{\textwidth}
    \includegraphics[width=0.99\textwidth]{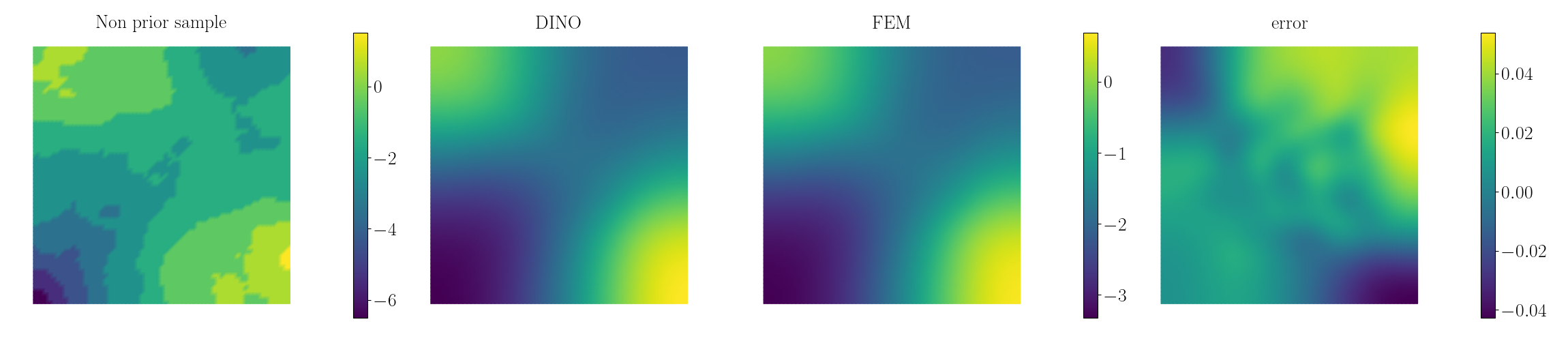}
  \end{subfigure}
    \caption{
    Prior samples (left), MAP points computed by the DINO surrogate and FEM (middle), and the errors of DINO approximation compared to FEM (right) for the diffusion problem (top) and the CDR problem (middle). Bottom: The comparison as above rows for an out-of-distribution sample (piecewise constant) for the CDR problem.
    } \label{fig:MAP}
\end{figure}

To further illustrate the accuracy of the MAP point, we take a random prior sample from the test data set, see left of Figure \ref{fig:MAP}, solve the PDE, and generate synthetic observation data at all candidate sensors to compute the MAP point using FEM and DINO trained with 8,192 samples. The MAP points computed by DINO and FEM, as shown in the middle of Figure \ref{fig:MAP}, are very close to each other with pointwise errors about two orders of magnitude smaller than the MAP point, see right of Figure \ref{fig:MAP}. We can also observe that MAP points look similar but not very close to the prior samples because of the intrinsic ill-posedness of the inverse problems. Note that in solving the Bayesian OED problems, we only need to consider random samples drawn from the prior distribution. Hence, it is sufficient for DINO to be accurate for the prior samples. We also test DINO with an out-of-(prior)distribution sample (piecewise constant binned from a prior sample) in the CDR problem. See the bottom of Figure \ref{fig:MAP}. The MAP point computed by DINO is still very close to the MAP point computed by FEM, which demonstrates the robustness of DINO for this problem. However, we do not expect this to hold for more general problems or more general out-of-distribution samples, especially those far from the prior distribution. 

To show the high accuracy of the reduced Jacobian by DINO in solving generalized eigenvalue problems, we plot the generalized eigenvalues of \eqref{eq:disc-gen-eig} by FEM and of \eqref{eq:rEig} by DINO at the MAP points computed from three random test samples, as shown in Figure \ref{fig:64-eig-possion}. The dominant eigenvalues—the first ten in the leading three orders of magnitude—are almost indistinguishable for the two methods in both the linear Poisson and nonlinear CDR problems, which is a demonstration of Theorem \ref{thm:2} for the bound of the approximation errors of the eigenvalues.


\begin{figure}[!htb]
    \centering
    \includegraphics[width=0.45\textwidth]{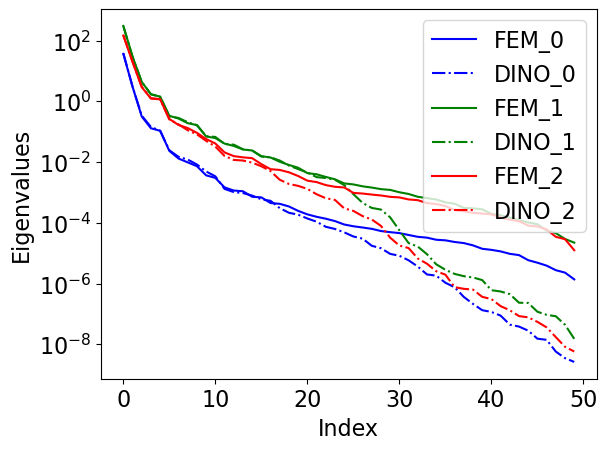} 
    \quad
    \includegraphics[width=0.45\textwidth]{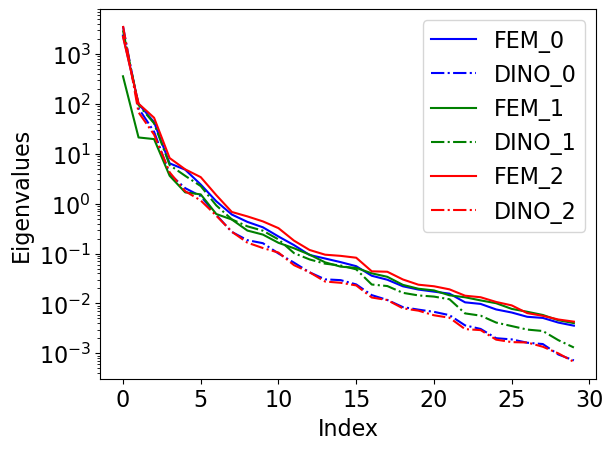} 
    \caption{The generalized eigenvalues of \eqref{eq:disc-gen-eig} computed using FEM and the eigenvalues of \eqref{eq:rEig} computed using DINO at three MAP points (computed from three random prior samples) for the diffusion problem (left) and the CDR problem (right).} \label{fig:64-eig-possion}
\end{figure}

\subsubsection{The scalability of the neural network approximations}

To check the scalability of the neural network approximations by DINO, we increase the parameter dimensions by refining the mesh. Thanks to the dimension reduction using the same reduced dimensions, the neural network size does not change, which leads to the same training cost for different mesh sizes. Figure \ref{fig:mesh-p} indicates that the accuracy of the DINO approximations for increasing mesh sizes of 32 $\times$ 32, 64 $\times$ 64, and 128 $\times$ 128 remains the same for the PtO map $F$, the reduced Jacobian $J_r$, and the MAP point, which implies the scalability (dimension independence) of the DINO approximations with respect to increasing parameter dimensions. This scalability also holds for the approximations of the generalized eigenvalues with the use of the DINO approximations of the reduced Jacobian as shown in Figure \ref{fig:128-32-eig-possion}. These results also confirm that with the scalable neural network approximations of the PtO map and the Jacobian, Theorem \ref{thm:1} and Theorem \ref{thm:2} still hold for different mesh sizes or parameter dimensions on the approximations of the MAP points and the eigenvalues.

\begin{figure}[!htb]
  \centering
    \includegraphics[width=0.49\textwidth]{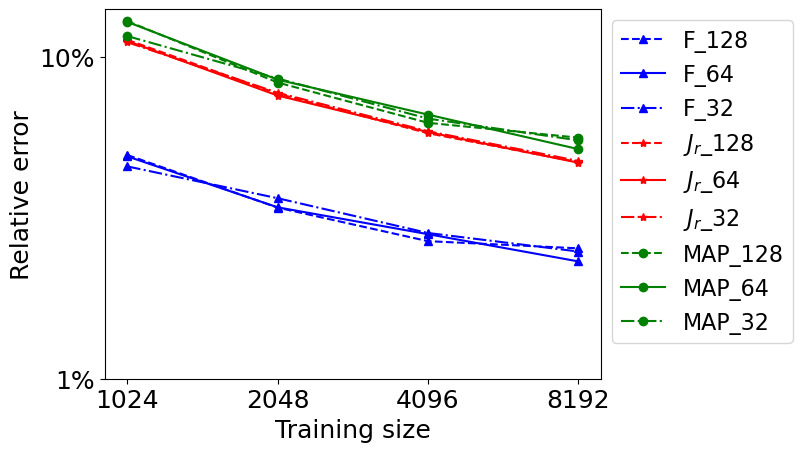}
    \includegraphics[width=0.49\textwidth]{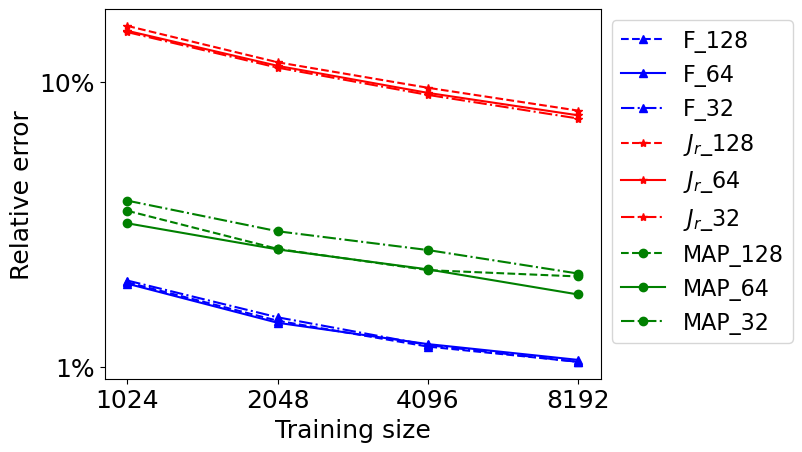}
    \caption{Mean of relative errors of the PtO map $F$, the reduced Jacobian ${J_r}$, and the MAP point with increasing training sizes and mesh sizes 32 $\times$ 32, 64 $\times$ 64, and 128 $\times$ 128 for the diffusion problem (left) and the CDR problem (right).} \label{fig:mesh-p}
\end{figure}

\begin{figure}[!htb]
    \centering
    \includegraphics[width=0.45\textwidth]{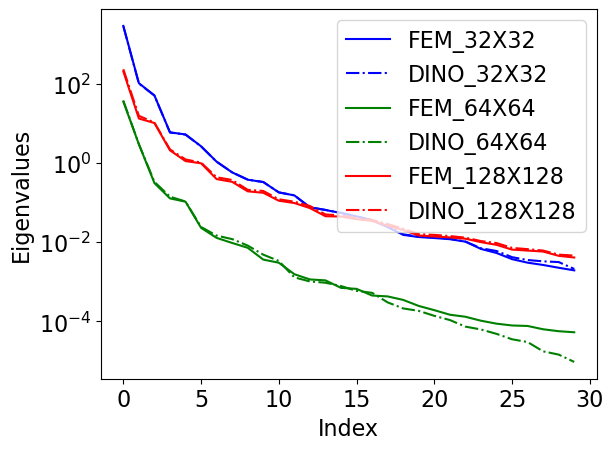} 
    \quad
    \includegraphics[width=0.45\textwidth]{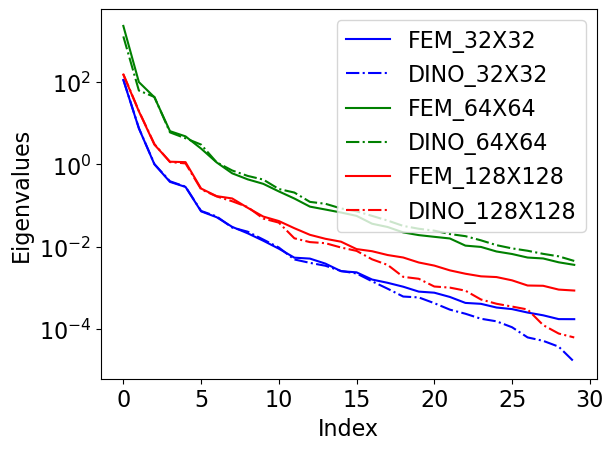}
    \caption{The generalized eigenvalues of \eqref{eq:disc-gen-eig} by FEM and the eigenvalues of \eqref{eq:rEig} by DINO with increasing parameter dimensions with mesh sizes (32 $\times$ 32, 64 $\times$ 64, and 128 $\times$ 128) for the diffusion problem (left) and the CDR problem (right).} \label{fig:128-32-eig-possion}
\end{figure}

\subsection{Accuracy verification and solution of Bayesian OED}


For the Bayesian OED problem, we select 5 sensor locations from 50 candidates for the diffusion problem and 10 sensor locations from 100 candidates for the CDR problem, as shown in Figure \ref{fig:poisson_setup}. 
\begin{figure}[!htb]
  \centering
  \includegraphics[width=0.45\textwidth]{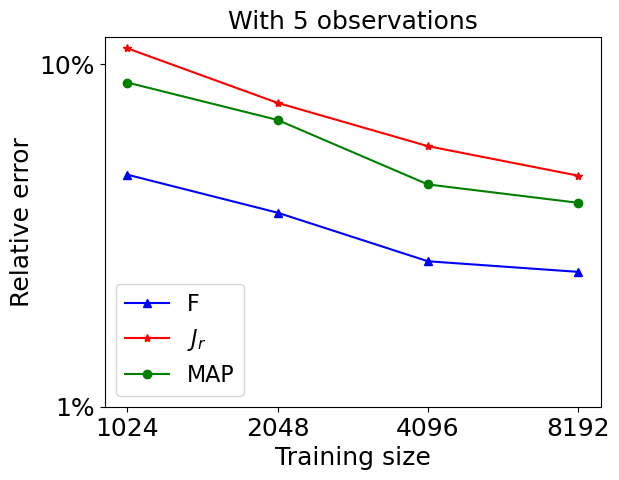} 
  \quad 
  \includegraphics[width=0.45\textwidth]{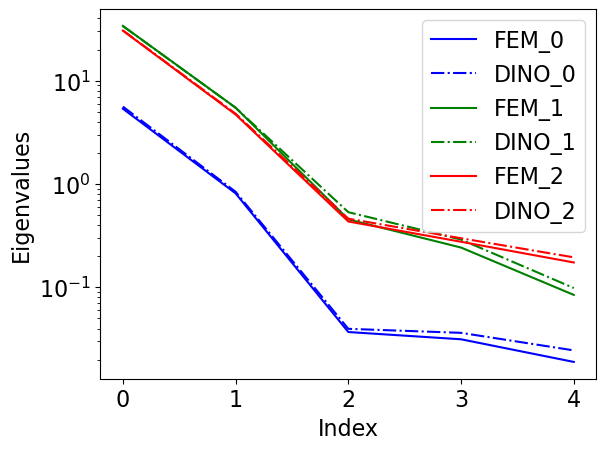} \\ 
  \includegraphics[width=0.45\textwidth]{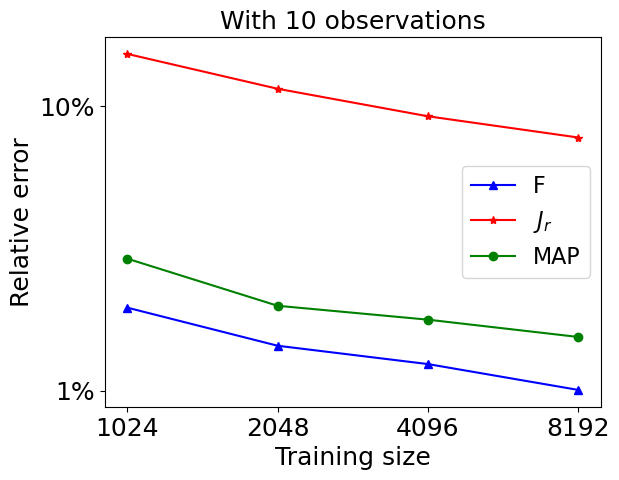} 
  \quad 
  \includegraphics[width=0.45\textwidth]{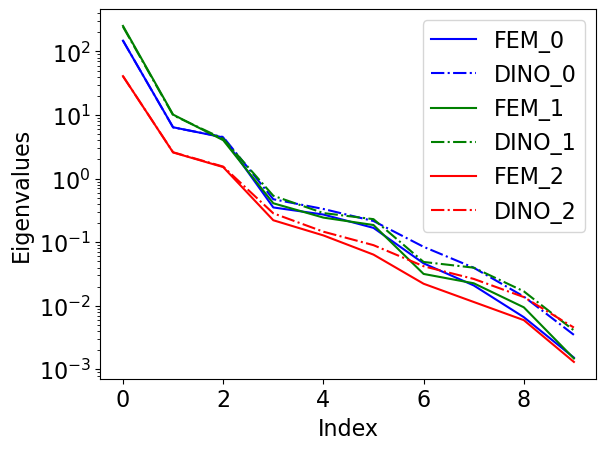}
  \caption{Left: mean of relative errors for the PtO map $F_\xi$, the reduced Jacobian $J_r$, and the MAP point with 5 and 10 observations for the diffusion problem (top) and CDR problem (bottom). Right: the eigenvalues by DINO and FEM.} \label{fig:5obs-pef} 
\end{figure}
To check that the DINO surrogates constructed for the full PtO map $F$ at all candidate observations preserve the approximation accuracy for the selected number of observations, we plot the relative errors of the DINO approximation for the PtO map $F_\xi$ at a random design $\xi$, its reduced Jacobian, the MAP point, and the generalized eigenvalues in Figure \ref{fig:5obs-pef}, which confirms the preserved accuracy for all quantities. Only 5 and 10 generalized eigenvalues exist for the 5 and 10 selected observations. 


To this end, we have constructed and verified the accuracy and scalability of the DINO surrogates to approximate the quantities involved in computing the optimality criteria of the Bayesian OED. We compute these criteria, namely the A-optimality in \eqref{eq:a-final-approx}, the D-optimality in \eqref{eq:d-final-approx}, and the EIG-optimality in \eqref{eq:eig-final-approx} using the eigenvalues from the projected eigenvalue decomposition and the MAP point from the projected optimization. We plot the optimality criteria computed using DINO and FEM approximations at 200 random samples ($\bsm, \bsy$) and report both the $R^2$ scores and the mean and standard deviation of the relative errors in Figure \ref{fig:poisson_cdr_sp_critet}. The results reveal a very high correlation between the DINO and FEM computation of the optimality criteria with $R^2$ score mostly larger than 0.99 and with minor relative errors. On the other hand, the relative errors for the three optimality criteria obtained by the neural network trained without the Jacobian (NN without J) are much larger than those obtained by DINO. This demonstrates the enhanced accuracy of the proposed method by incorporating derivative information in DINO. We use 2048 training samples for both DINO and NN without J for the plots in Figure \ref{fig:poisson_cdr_sp_critet}. We also include the relative errors of approximate optimality criteria with the training sizes 1024, 4096, and 8192 in Table \ref{tab:metrics_mean_std_error}. We observe that as the neural network approximations become more accurate with increasing training sizes, the approximation of the optimality criteria also becomes more accurate. Moreover, DINO leads to much more accurate approximations  than the neural network trained without Jacobian, especially when the training size is relatively small (e.g., at 1024), which demonstrates the theoretical result in Theorem \ref{thm:2}.



\begin{figure}[!htb]
  \centering      
  \begin{subfigure}{0.33\textwidth}
        \includegraphics[width=\linewidth]{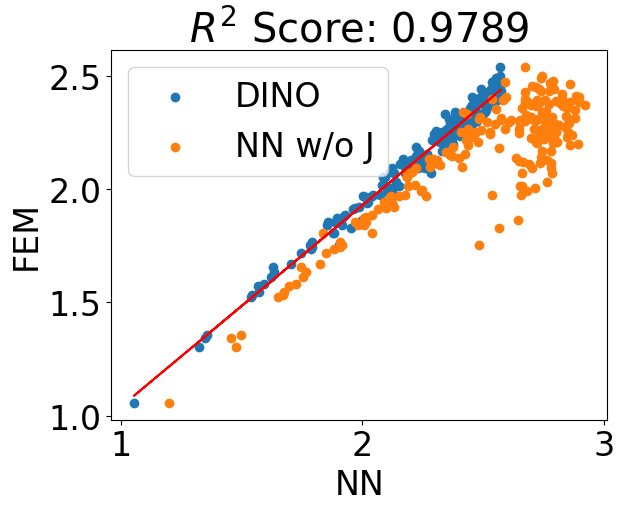}
        \caption{$4.84\% (2.27\%), 16.20\% (8.71\%)$}\label{fig:5obs-trace_poisson}%
    \end{subfigure}
    \hfill
    \begin{subfigure}{0.3\textwidth}
        \includegraphics[width=\linewidth]{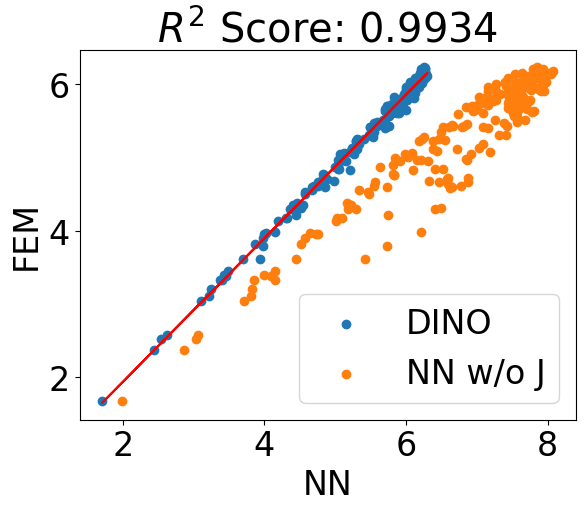}
        \caption{$2.45\% (1.42\%), 27.77\% (7.68\%)$}\label{fig:5obs-det_poisson}
    \end{subfigure}
    \hfill
    \begin{subfigure}{0.3\textwidth}
        \includegraphics[width=\linewidth]{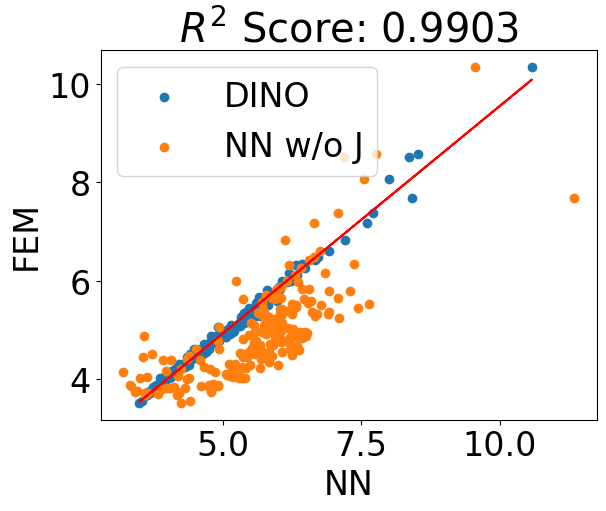}
        \caption{$1.15\% (2.25\%), 14.79\% (14.00\%)$}\label{fig:5obs-eig_poisson} 
    \end{subfigure} \\
    \begin{subfigure}{0.30\textwidth}
        \includegraphics[width=\linewidth]{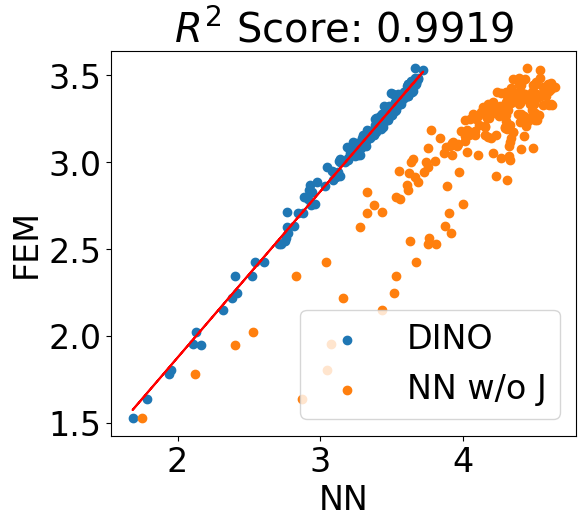}
        \caption{$5.89\% (1.35\%), 32.25\% (8.54\%)$}\label{fig:10obs-trace_cdr}
    \end{subfigure}
    \hfill
    \begin{subfigure}{0.34\textwidth}
        \includegraphics[width=\linewidth]{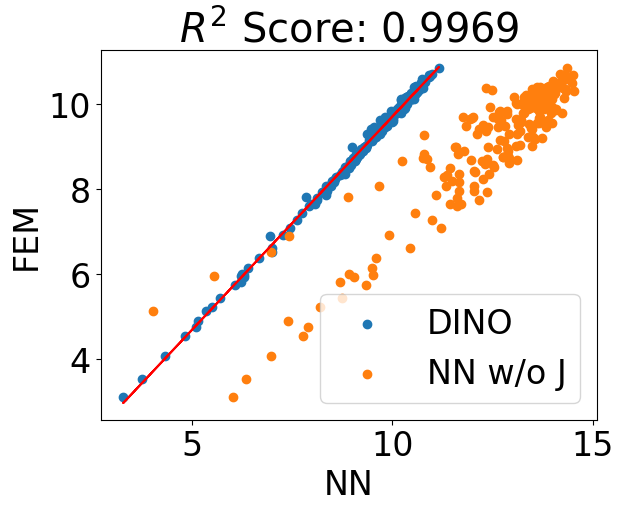}
        \caption{$3.43\% (1.32\%), 37.46\% (11.98\%)$}\label{fig:10obs-det_cdr}
    \end{subfigure}
    \hfill
    \begin{subfigure}{0.3\textwidth}
        \includegraphics[width=\linewidth]{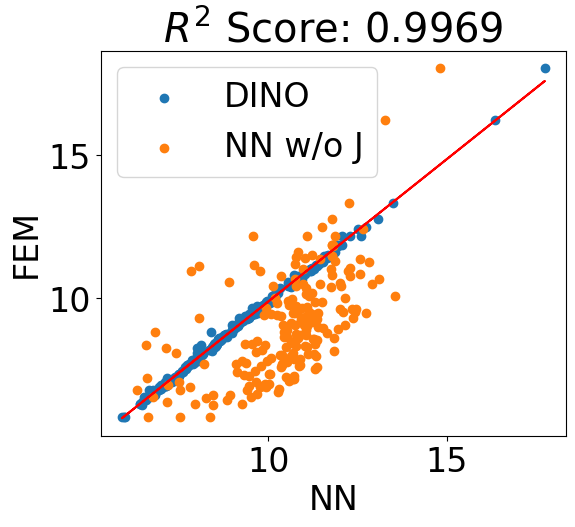}
        \caption{$1.07\% (1.11\%), 17.82\% (16.62\%)$}\label{fig:10obs-eig_cdr}
    \end{subfigure}
  \caption{$R^2$ score (for correlation) of DINO vs FEM evaluations of the optimality criteria, mean and standard deviation of the relative errors obtained by DINO and the neural network trained without Jacobian information (NN w/o J) compared to FEM in the evaluation of the trace (left), determinant (middle), and information gain (right) at 200 random samples. We choose 5 sensors in the diffusion problem (top) and 10 sensors in the CDR problem (bottom).
  }\label{fig:poisson_cdr_sp_critet}
\end{figure}

\begin{table}[!htb]
  \footnotesize
  \begin{center}
  \begin{tabular}{|c|c|c|c|c|c|c|}
  \hline
  \multicolumn{7}{|c|}{Linear diffusion problem} \\ \hline
  Training size & \multicolumn{2}{c|}{1,024} & \multicolumn{2}{c|}{4,096} & \multicolumn{2}{c|}{8,192} \\ \hline
  Optimality & with J & w/o J & with J & w/o J & with J & w/o J \\ \hline
  A \eqref{eq:a-final-approx-W} & 1.41\% (1.12\%) & 3.08\% (3.26\%) & 1.49\% (0.98\%) & 1.43\% (1.21\%) & 1.46\% (1.14\%) & 1.14\% (1.03\%) \\ 
  \hline
  A \eqref{eq:a-final-approx} & 4.52\% (3.03\%) & 16.71\% (10.53\%) & 5.20\% (2.09\%) & 15.81\% (6.34\%) & 5.16\% (2.08\%) & 9.70\% (4.27\%) \\ 
  \hline
  D & 3.95\% (1.58\%) & 27.73\% (9.48\%) & 2.07\% (1.25\%) & 19.64\% (5.14\%) & 1.59\% (1.23\%) & 12.36\% (2.12\%) \\ \hline
  EIG & 2.26\% (2.04\%) & 16.97\% (9.92\%) & 1.33\% (1.08\%) & 13.18\% (7.57\%) & 0.98\% (0.78\%) & 6.76\% (3.83\%) \\ \hline
  \multicolumn{7}{|c|}{Nonlinear convection-diffusion-reaction problem} \\ \hline
  Training size & \multicolumn{2}{c|}{1,024} & \multicolumn{2}{c|}{4,096}& \multicolumn{2}{c|}{8,192} \\ \hline
  Optimality& with J & w/o J  & with J & w/o J & with J & w/o J \\ \hline
  A \eqref{eq:a-final-approx-W} & 0.55\% (0.47\%) & 1.76\% (1.95\%) & 0.41\% (0.21\%) & 0.45\% (0.65\%)& 0.47\% (0.31\%) & 0.42\% (0.26\%) \\ \hline
  A \eqref{eq:a-final-approx} & 6.97\% (2.04\%) & 60.52\% (20.44\%) & 5.69\% (1.18\%) & 11.12\% (2.91\%)& 5.40\% (0.95\%) & 7.54\% (1.72\%) \\ \hline
  D & 4.73\% (2.96\%) & 81.7\% (39.94\%) & 3.33\% (1.17\%) & 10.48\% (3.4\%) & 2.65\% (0.91\%) & 4.98\% (1.54\%)\\ \hline
  EIG & 1.81\% (1.18\%) & 48.61\% (28.84\%) & 1.31\% (0.73\%) & 4.93\% (3.1\%)& 1.0\% (0.54\%) & 2.14\% (1.29\%) \\ \hline
  \end{tabular}
  \caption{Mean (and standard deviation) of the relative errors of optimality criteria using the neural network surrogates trained with 1,024, 4,096, and 8,192 samples with Jacobian (DINO) and without (w/o) Jacobian for the two problems.}
  \label{tab:metrics_mean_std_error}
  \end{center}
  \end{table}
  


\begin{figure}[!htb]
  \centering
  
  \begin{subfigure}[b]{0.45\textwidth}
    \includegraphics[width=\textwidth]{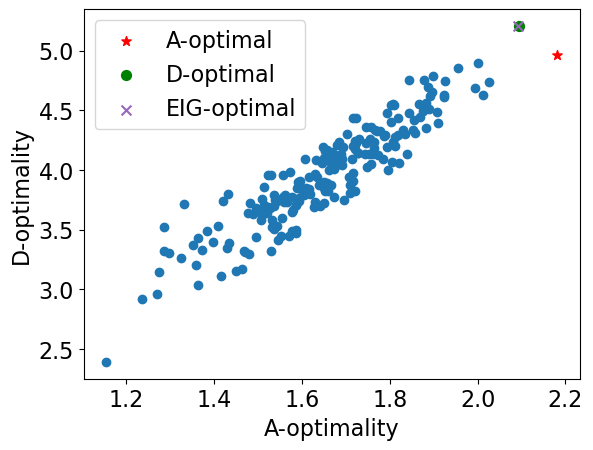}
  \end{subfigure}
  \quad
  \begin{subfigure}[b]{0.45\textwidth}
    \includegraphics[width=\textwidth]{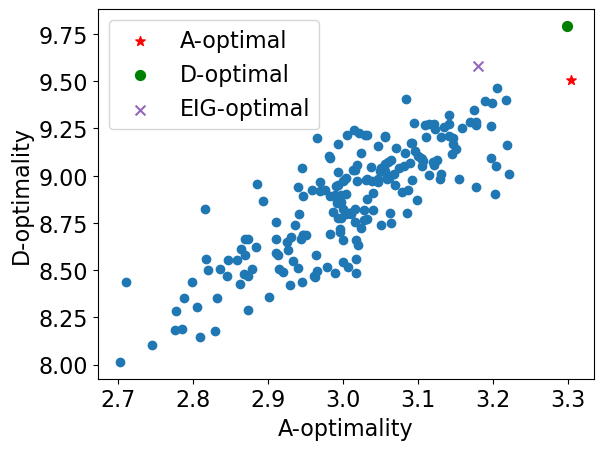}
  \end{subfigure}

  \begin{subfigure}[b]{0.45\textwidth}
    \includegraphics[width=\textwidth]{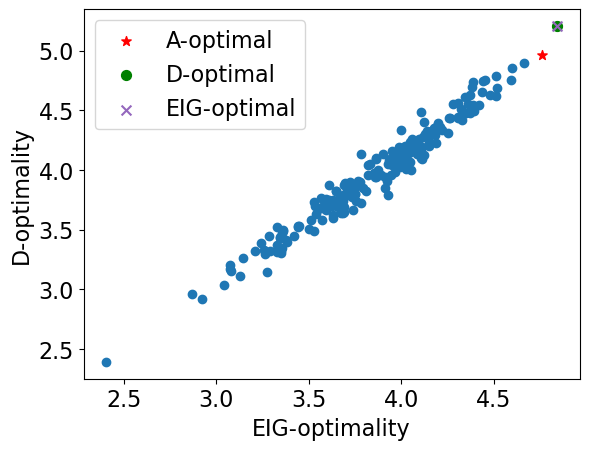}
  \end{subfigure}
  \quad
  \begin{subfigure}[b]{0.45\textwidth}
    \includegraphics[width=\textwidth]{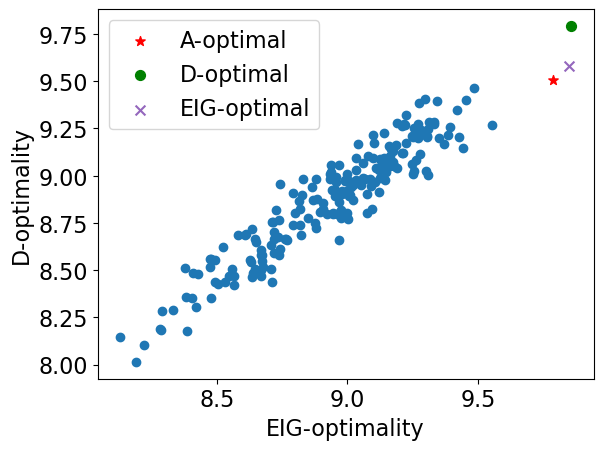}
  \end{subfigure}

  \caption{FEM optimality values at the optimal sensors selected using DINO by the swapping greedy algorithm \ref{alg:buildtree} with different optimality criteria compared to those at 100 random designs. Diffusion (left) and CDR (right) problems.}
  \label{fig:swapping_greedy_perf}
\end{figure}

We run the swapping greedy algorithm (defined by Algorithm \ref{alg:buildtree}) with setting the stopping criteria $k_{\text{max}} = 3$ and $\varepsilon_{\text{min}} = 0.01$, and use DINO to find the optimal sensor locations according to the three optimality criteria of A-optimality, D-optimality, and EIG. To check the optimality of the sensors selected by the algorithm using DINO surrogates, we compute the three optimality criteria using the high-fidelity FEM at the selected sensors compared to 200 randomly selected sensors. The results are shown in Figure \ref{fig:swapping_greedy_perf}, from which we can see that optimality criteria achieve their largest values at the sensors selected using DINO surrogates according to the corresponding optimality criteria, which demonstrates the effectivity of the DINO surrogates. Moreover, larger optimality criteria correspond to smaller uncertainty in the model parameter estimation. The uncertainty indicated by the optimality criteria is effectively reduced by the optimal design from random design. 

The corresponding sensors selected according to different optimality criteria are shown in Figure \ref{fig:poisson-5obs-location} for the two test problems with 5 and 10 sensors, respectively. We can see that different optimality criteria can lead to slightly different sensors. The uncertainty is reduced effectively from the prior as observed in the pointwise prior and posterior variance in Figure \ref{fig:UD-Prior-Post}. Moreover, we can also see that the uncertainty reduction by the (A-)optimal sensor locations is greater than that by a random choice of sensor locations, consistent with the smaller trace values.



\begin{figure}[!htb]
    \centering
    \includegraphics[width=0.31\textwidth]{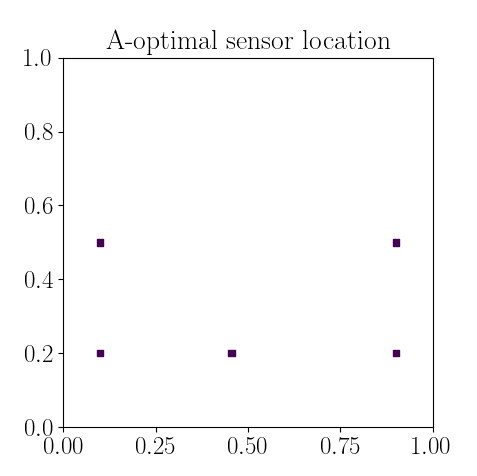} 
    \includegraphics[width=0.32\textwidth]{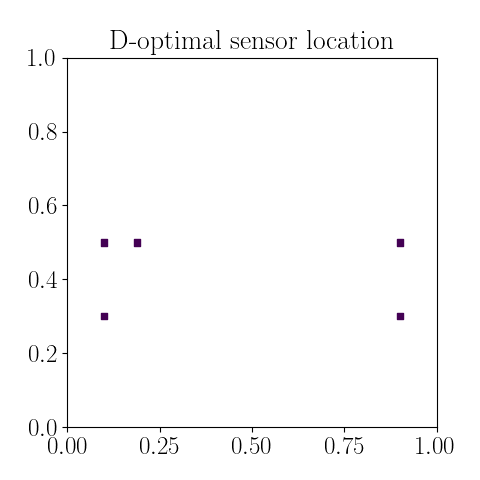} 
    \includegraphics[width=0.32\textwidth]{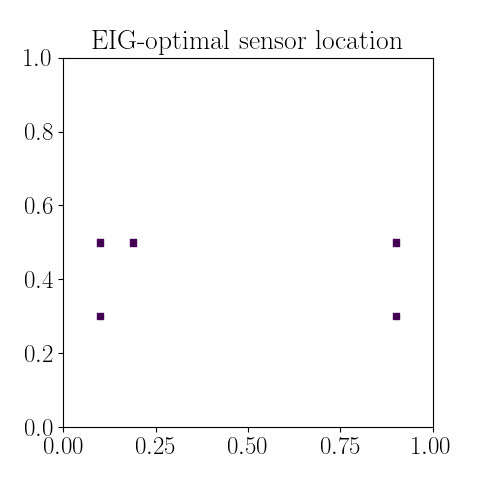} 
    \\
    \includegraphics[width=0.32\textwidth]{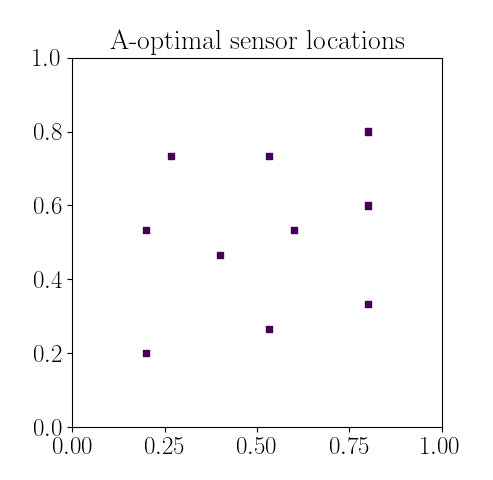} 
    \includegraphics[width=0.32\textwidth]{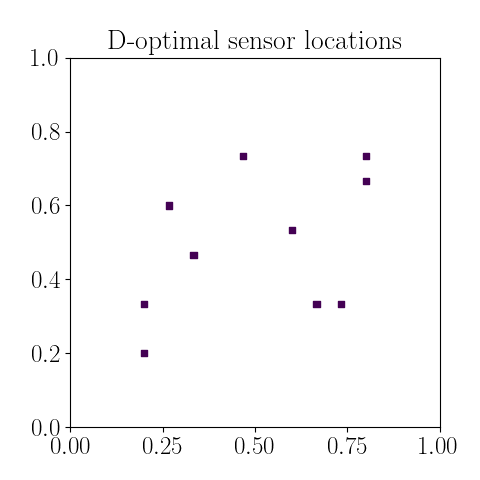} 
    \includegraphics[width=0.32\textwidth]{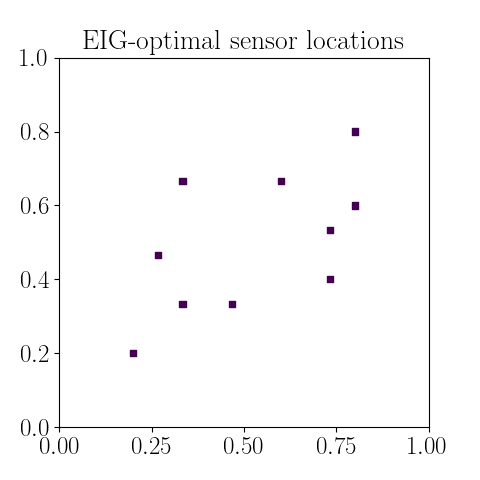} 
    \caption{Optimal sensor locations selected by the swapping greedy algorithm \ref{alg:buildtree} using the A-/D-/EIG-optimality criteria. 5 sensors selected for the diffusion problem (top) and 10 sensors selected for the CDR problem (bottom). 
    }
    \label{fig:poisson-5obs-location}
\end{figure}

\begin{figure}[!htb]
 \centering
 \begin{subfigure}[b]{\textwidth}
 \centering
\includegraphics[width=0.95\textwidth]{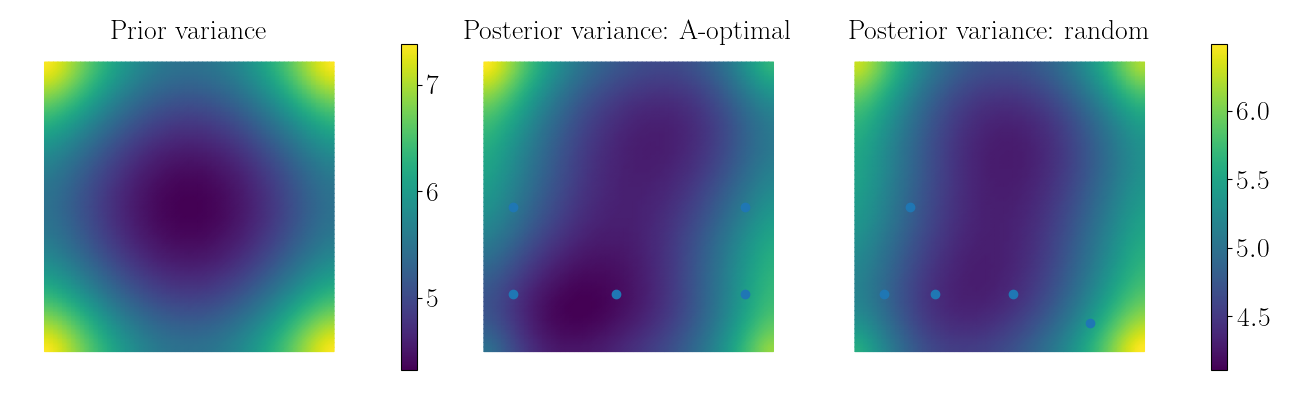}
  \end{subfigure} \\
  \begin{subfigure}[b]{\textwidth}
  \centering
\includegraphics[width=0.95\textwidth]{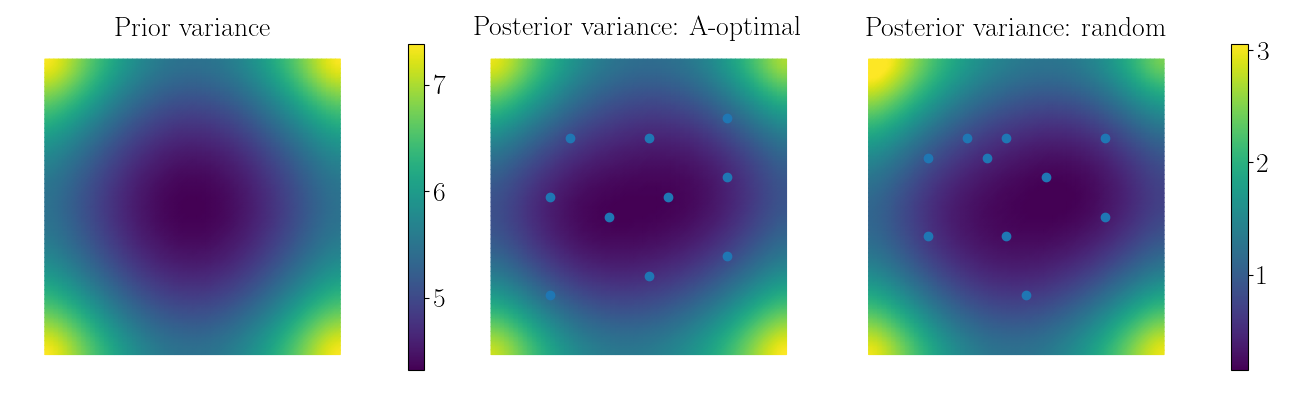}
  \end{subfigure} \\
    \caption{
    Pointwise prior variance (left), posterior variance with A-optimal (middle), and random sensor locations (right) for the diffusion (top) and the CDR problem (bottom). The trace of the posterior variance for A-optimal and random sensor placements are 20,176 and 20,599 for the diffusion problem, and 4,076 and 4,458 for the CDR problem, respectively.} \label{fig:UD-Prior-Post}
\end{figure}

\subsection{Application to a 3D test problem}
We further apply the proposed method to solve a Bayesian OED problem governed by a 3D CDR equation \eqref{eq:cdr_equation} in the physical domain $D = (0,1)^3$. 
We consider a selection of 8 sensor locations out of 512 candidates equidistantly distributed in the domain to infer the model parameter following Gaussian prior $m \sim \cN(0,\cC)$ with Mat\'ern covariance operator $\cC = \cA^{-2}$, where $\cA=-\gamma\Delta + \kappa I$ with $\gamma = 0.4$ and $\kappa = 1.0$. 

For the FEM approximation of both the state and the parameter fields, we use piecewise linear polynomials with a uniform mesh of size 32 $\times$ 32 $\times$ 32, leading to 35,937 dimensional discrete parameters. The FEM solution of the 3D forward problem is computationally much more expensive than that of the 2D forward problem. For the DINO surrogate, we use the same ResNet architecture as that for the 2D diffusion problem with a larger reduced input dimension of 256 and reduced output dimension of 128 for the 512 outputs.

We train the DINO with increasing numbers of training samples, which leads to increasing accuracy of the DINO approximation of the PtO map $F$, the reduced Jacobian ${J_r}$, and MAP point $\bsm_\text{MAP}$, 
as shown in Figure \ref{fig:3d-perf} (left). This figure shows much higher accuracy for all three quantities predicted by the neural networks trained with the Jacobian information than those trained without it. A comparison of the MAP point computed by DINO vs FEM is shown in Figure \ref{fig:cdr_3d}, from which we can observe an accurate approximation by DINO. The eigenvalues at three different MAP points computed by DINO are very accurate compared to FEM, as shown in Figure \ref{fig:3d-perf} (right). 


\begin{figure}[!htb]
  \centering
    \begin{subfigure}{0.48\textwidth}
        \includegraphics[width=\linewidth]{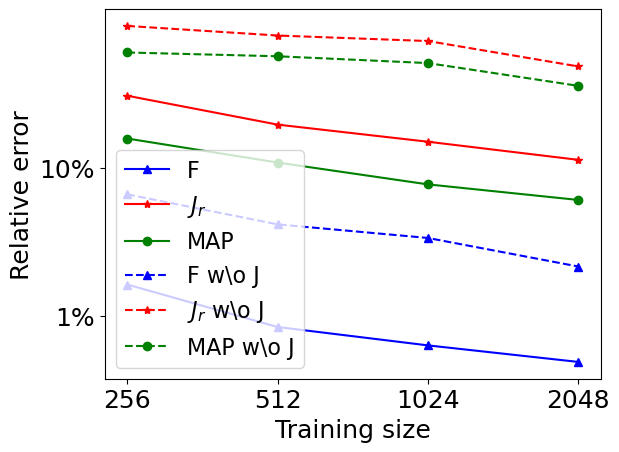}
    \end{subfigure}
    \hfill
    \begin{subfigure}{0.48\textwidth}
        \includegraphics[width=\linewidth]{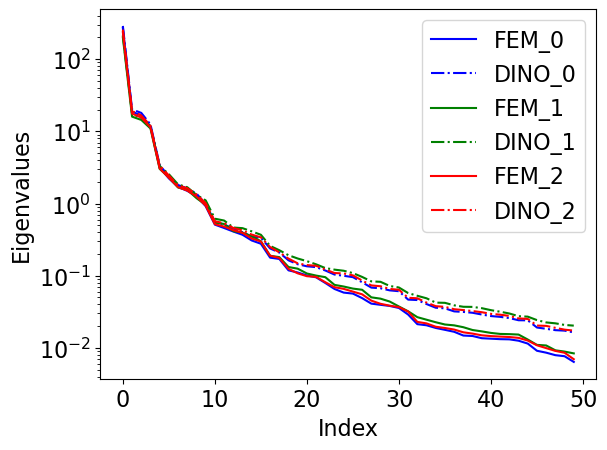}
    \end{subfigure}
    \caption{Mean of relative errors with increasing training size for the DINO approximations of the PtO map $F$, the reduced Jacobian ${J_r}$, and the MAP point $\bsm_\text{MAP}$ (left). Comparison of the generalized eigenvalues by DINO and FEM (right).} \label{fig:3d-perf}
\end{figure}

 
\begin{figure}[!htb]
  \centering      
  \begin{subfigure}{0.3\textwidth}
        \includegraphics[width=\linewidth]{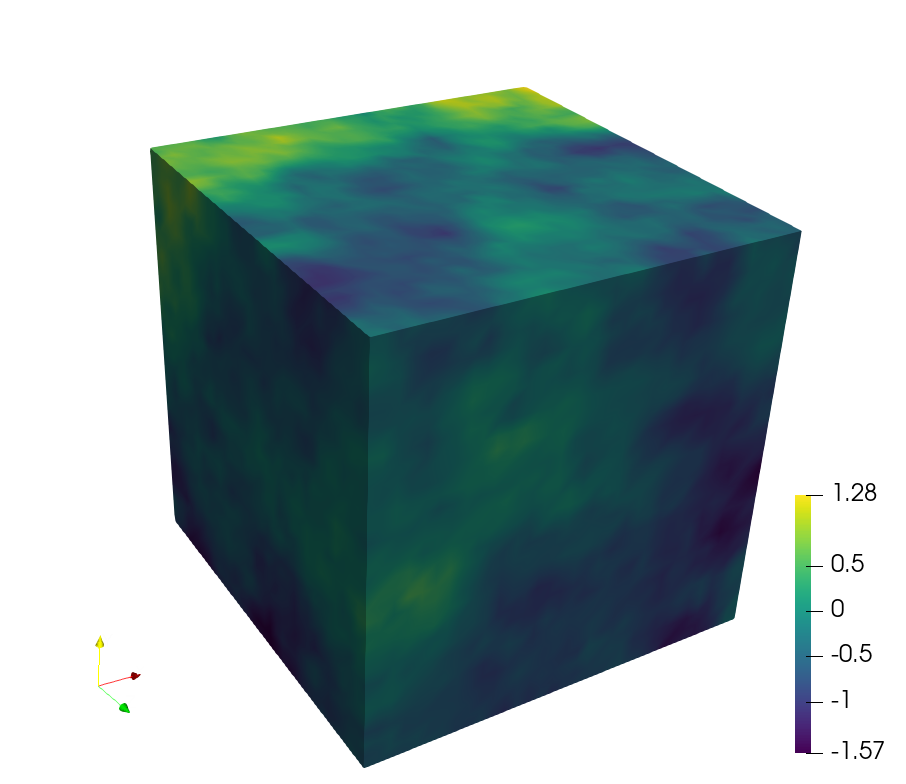}
    \end{subfigure}
    \hfill
    \begin{subfigure}{0.3\textwidth}
        \includegraphics[width=\linewidth]{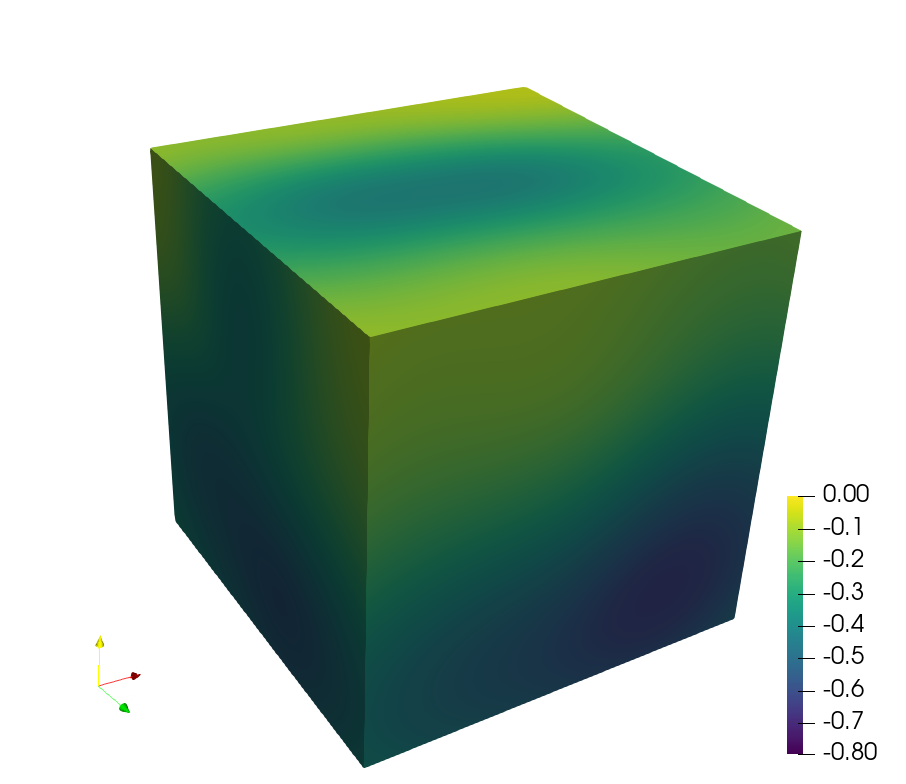}
    \end{subfigure}
    \hfill
    \begin{subfigure}{0.3\textwidth}
        \includegraphics[width=\linewidth]{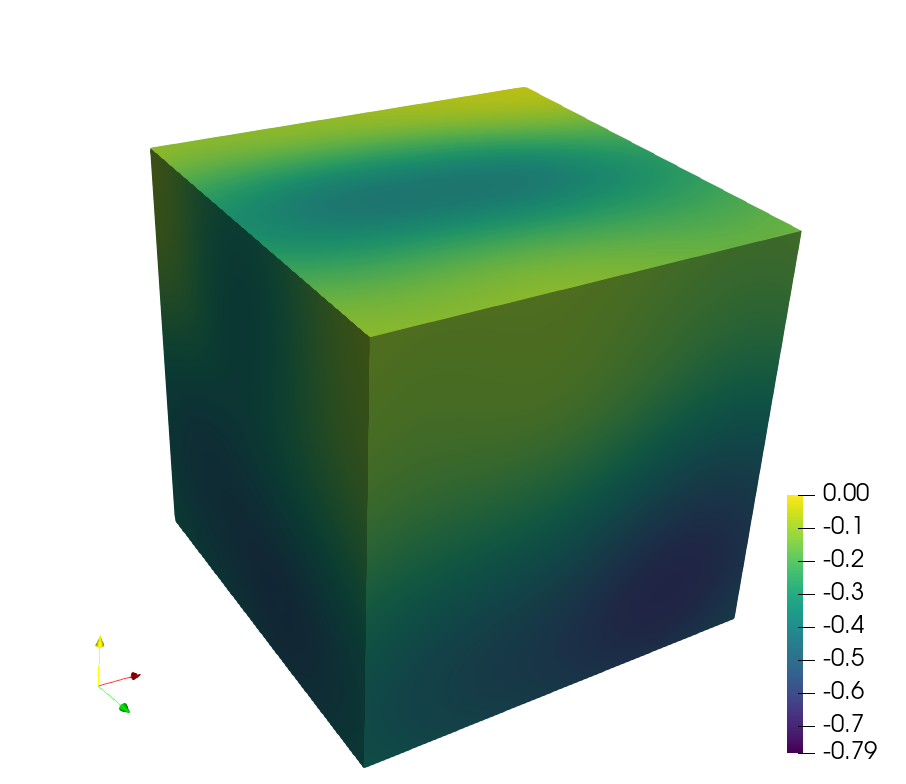}
    \end{subfigure}
  \caption{Prior sample (left), MAP points computed by DINO (middle) and FEM (right).}
  \label{fig:cdr_3d}
\end{figure}

We use the trained DINO surrogate to compute the three optimality criteria, i.e., the expected A/D-optimality and information gain, with 128 random samples for evaluating the expectation of the optimality criteria. We apply the swapping greedy algorithm to optimize the sensor locations, with the eight sensors selected using different optimality criteria shown in Figure \ref{fig:3d-poisson-location}, and the optimality values at the selected sensors shown in Figure \ref{fig:swapping_greedy_perf_3d}. We observe that the optimality values at the selected sensors are optimal with respect to the corresponding optimality criteria and much larger than those at the randomly selected sensors, which demonstrates the effectiveness of the swapping greedy algorithm for the 3D Bayesian OED problem. 

\begin{figure}[!htb]
    \centering
    \includegraphics[width=0.32\textwidth]{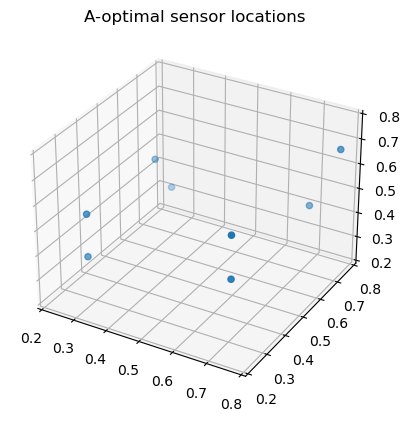} 
    \includegraphics[width=0.32\textwidth]{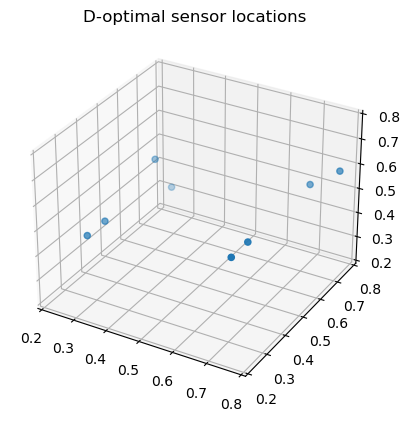} 
    \includegraphics[width=0.32\textwidth]{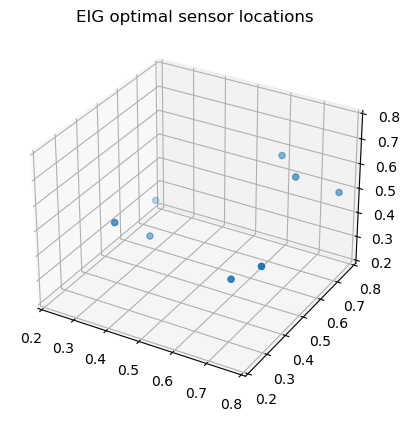} 
    \label{fig:sensor-locations-3d}
\caption{Eight optimal sensor locations selected using DINO by the swapping greedy algorithm \ref{alg:buildtree} with the optimality criteria of A-optimality (left), D-optimality (middle), and EIG-optimality (right).} \label{fig:3d-poisson-location}
\end{figure}


\begin{figure}[!htb]
    \centering
    \includegraphics[width=0.45\textwidth]{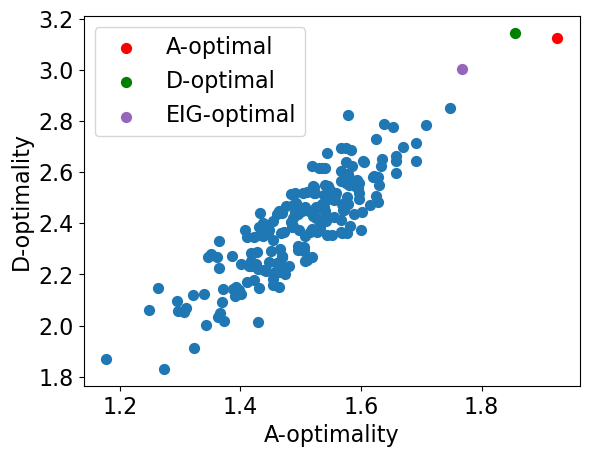} 
    \includegraphics[width=0.45\textwidth]{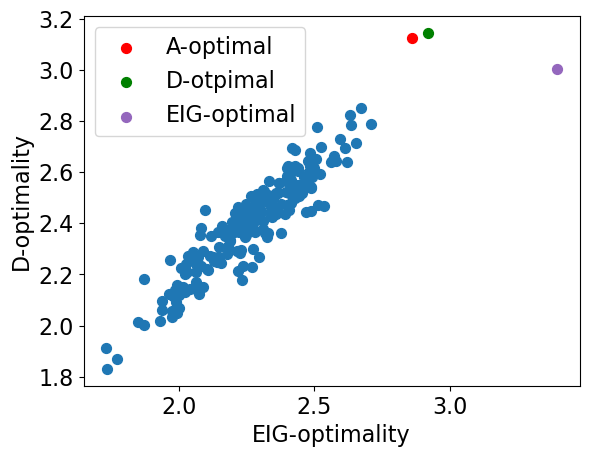}
    \caption{Optimality values at the optimal sensors selected by the swapping greedy algorithm \ref{alg:buildtree} using different optimality criteria compared to those at 200 random designs. The optimality values are efficiently computed using DINO.} \label{fig:swapping_greedy_perf_3d}
\end{figure}

\subsection{Efficiency of the computational framework}


In this section, we demonstrate the efficiency of our proposed approach by reporting the computational cost of generating the training data, training the DINO surrogates, and their use in solving Bayesian OED problems. We compare these costs to those of the high-fidelity computations. We use the CPU of AMD EPYC 7543 with 1 TB memory for high-fidelity computation and neural network evaluation. For the neural network surrogate training, we use the GPU of NVIDIA RTX A6000 with 48 GB memory.



In Table \ref{tb:computing times}, we report the computational complexity analysis numbers for the two (2D and 3D) CDR problems corresponding in Section \ref{sec:complexity}. In particular, for the computation of the MAP point, we use an inexact Newton-CG optimizer from the hIPPYlib library for the high-fidelity approximation (averaged 5 times) and an LBFGS optimizer from PyTorch for the neural network approximation. 

\begin{table}[!htb]
  \begin{center}
    \begin{tabular}{|c|c|c|c|c|c|c|c|c|c|c|c|}
      \hline
      Problem & $N_s$ & $k_s$ & $r_s$ & $d_s$ & $N_{nt}$ & $N_{cg}$ & $d_\bsm$ & $r_\bsm$ & $r_F$ & $N_t$ & $N_b$\\ \hline
      2D CDR &  128 & 2 & 10 & 100 & 5.6 & 4.6 & 16,641 & 128 & 30 & 8,192 & 100  \\ \hline
      3D CDR & 128 & 2 & 8 & 512 &  4.2 & 2.2 & 35,937 & 256 & 128 & 2,048 & 100 \\ \hline
    \end{tabular}
    \caption{
    $N_s$: \# SAA samples, $k_s$: \# swapping loops, $r_s$: \# sensors to select, $d_s$: \# candidate sensors, $N_{nt}$: averaged \# Newton iterations, $N_{cg}$: averaged \# CG iterations per Newton iteration, $d_\bsm$: \# discretized parameters, $r_\bsm$: \# input project bases, $r_F$: \# output projection bases, $N_t$: \# training samples, $N_b$: \# BFGS iterations.}\label{tb:computing times}
  \end{center}
  \end{table}
  
 In Table \ref{tab:ct_nn_fem_1}, we report the computational time (averaged over 5 times) by FEM using FEniCS v.s.\ by DINO using PyTorch for the computation of the PtO map, the MAP point solving the optimization problem \eqref{eq:MAP} v.s.\ \eqref{eq:rMAP}, and the eigenvalue decomposition \eqref{eq:disc-gen-eig} v.s.\ \eqref{eq:rEig}. A speedup from hundreds for the 2D problem to thousands for the 3D problem can be observed. As the evaluation of the optimality criteria is dominated by one MAP optimization and one eigenvalue decomposition for each sample, the combined speedup of the DINO compared to FEM is about $(5.1+0.7)/(0.04+0.007) = 123$ and $(128.7+10.0)/(0.04+0.02) = 2,312$ for the 2D and 3D problems, respectively.
 

\begin{table}[!htb]
  \centering
  \begin{tabular}{|c|c|c|c|}
  \hline
  2D-CDR & PtO & MAP & Eigenpairs \\ \hline
  FEM & 0.8 & 5.1 & 0.7 \\ \hline
  DINO & 0.002 & 0.04 & 0.007 \\ \hline
  Speedup & 400 & 127 & 100 \\ \hline
  \end{tabular}
  \begin{tabular}{|c|c|c|c|}
    \hline
    3D-CDR & PtO &  MAP & Eigenpairs \\ \hline
    FEM & 25.5
     & 128.7 & 10.0\\ \hline
    DINO & 0.003 & 0.04 & 0.02 \\ \hline
    Speedup & 8,500 & 3,218 & 500 \\ \hline
  \end{tabular}
  \caption{Time (in seconds) and speedup by DINO v.s.\ FEM for the computation of the PtO map, the MAP point, and the eigenvalue decomposition, averaged over 5 times.}
  \label{tab:ct_nn_fem_1}
\end{table}

  
%

In Table \ref{tab:ct_nn_fem_2}, we report the offline computational time in (1) computing the (DIS and PCA) projection bases for input and output dimension reduction, (2) computing the training data of the PtO maps and (3) their Jacobians, and (4) training the neural networks. We note that the cost of computing the additional reduced Jacobian data is smaller than that of the PtO map data, as the reduced Jacobian computation is amortized by a direct solver using LU factorization (0.04 seconds for 2D CDR and 2.38 seconds for 3D CDR) and its repeated use in solving $r_t = \min(r_\bsm, r_F)$ linearized PDEs (each takes $0.0008 \ll 0.04$ seconds for 2D CDR and $0.02 \ll 2.38$ seconds for 3D CDR). The training of the neural networks takes much less time (145 seconds for 2D CDR and 413 seconds for 3D CDR in GPU) than that of generating the training data. 

\begin{table}[!htb]
  \centering
  \begin{tabular}{|c|c|c|c|c||c|}
    \hline
    Problem & Bases & PtO & Jacobian & Total & Train (GPU) \\ \hline
    2D-CDR & 437  & 6,554 & 492 & 7,482 & 145 \\ \hline
    3D-CDR & 31,241 & 52,224 & 10,076 & 93,541 & 413  \\ \hline
  \end{tabular}
  \caption{Offline time (in seconds) in computing the input and output projection bases, PtO maps, Jacobians, and training (GPU) the neural networks, averaged over 5 times.}
  \label{tab:ct_nn_fem_2}
\end{table}


Solving the Bayesian OED problem constrained by the 2D CDR equation takes $2.1$ CPU hours to generate the projection bases, the training data of the PtO map, and its Jacobian, and an additional $0.04$ GPU hours for training the neural networks. Then, it takes $3.8 $CPU hours using DINO to run the swapping optimization algorithm with $293,120$ evaluations of the optimality criteria. In contrast, it would take $472.2$ CPU hours using the high-fidelity model to solve the Bayesian OED problem. The speedup in CPU hours by the DINO surrogate (including both offline and online time) over the high-fidelity evaluation is $80.0 (472.2/(3.8+2.1))$. In the case of the 3D CDR, the speedup in CPU hours by DINO surrogate over the high-fidelity model with $1,552,896$ evaluations of the optimality criteria is $1148.3 (59,829.6/(26.1 + 26.0))$. We remark that to further speed up the computation, we have implemented both the data generation and the optimality criteria evaluation in parallel, using $128$ CPU processors.

\section{Conclusion}
\label{sec:conclusion}
In this work, we develop an accurate, scalable, and efficient computational framework based on DINO to solve Bayesian OED problems constrained by PDEs with infinite-dimensional parameter input. We consider the three optimality criteria of A-/D-optimality and EIG, which all require the computation of the Jacobian of the PtO map with Laplace and low-rank approximations of the posterior. We propose to use derivative-informed dimension reduction of the parameter and DINO surrogate of the PtO map to achieve high accuracy and scalability of the approximations for both the PtO map and especially its Jacobian. We derive efficient formulations for computing the MAP point and eigenvalue decomposition in the reduced dimension, which leads to efficient evaluation of the optimality criteria. We provide detailed error analysis for the approximations of the MAP point, the eigenvalues, and the optimality criteria under suitable assumptions of the PtO map and its Jacobian approximations and dimension reduction. To solve the optimization problem of the Bayesian OED, we present a modified swapping greedy algorithm initialized by a greedy algorithm, which achieves fast convergence. 

With two numerical examples of linear diffusion and nonlinear CDR problems in both two- and three-dimensional physical domains, we demonstrate the accuracy, scalability, and efficiency of the proposed method and verify the assumptions and theoretical results. Specifically, we demonstrate that DINO provides more accurate approximations of the PtO map, the reduced Jacobian, and the MAP point than other neural networks trained without Jacobian. This results in significantly higher approximation accuracy for all optimality criteria in Bayesian OED. We demonstrate that the accuracy of DINO approximations with proper dimension reduction is scalable with respect to increasing parameter dimensions for the same number of training samples or the same number of PDE solves. 
We also demonstrated that the proposed method achieved high efficiency with a speedup of 80$\times$ in 2D and 1148$\times$ in 3D, including both the offline data generation and online evaluation time, compared to the high-fidelity approximations for the CDR problems.

This work uses Laplace and low-rank approximations of the posterior distribution for Bayesian OED. It is limited to the condition that the Laplace approximation is a good approximation of the posterior, and the decay of the eigenvalues of the Hessian is fast. In future work, we plan to address these limitations by developing the DINO surrogates in the context of variational approximation of the posterior distribution that can be highly non-Gaussian, and by using nonlinear dimension reduction and hierarchical approximation of the Hessian. We also plan to extend the computational framework to solve sequential Bayesian OED problems with efficient construction of time-dependent DINO surrogates.

\section*{Acknowledgments}
This work is partially funded by National Science Foundation under grants \#2245674, \#2325631, \#2245111, and \#2233032. The first author acknowledges travel support from the Society for Industrial and Applied Mathematics. We thank helpful discussions with Thomas O'Leary-Roseberry, Dingcheng Luo, and Grant Bruer.

\bibliographystyle{elsarticle-num}
\bibliography{references,peng}

\end{document}